\documentclass[11pt,reqno]{amsart}

\usepackage{amsfonts,amssymb,amsmath,gensymb,epic,eepic,float,fullpage}
\usepackage[usenames,dvipsnames]{color}
\usepackage{siunitx}
\usepackage{empheq}
\usepackage[mathscr]{euscript}
\usepackage[T1]{fontenc}
\usepackage{mathtools}
\usepackage{comment}
\usepackage{bm}
\usepackage{caption}
\usepackage{hyperref}
\usepackage{enumerate,enumitem}
\usepackage{tikz}
\usepackage{bbm}
\usepackage{xcolor}
\usepackage{dsfont}
\usepackage{xparse}
\usepackage{physics}

\usetikzlibrary{graphs,patterns,decorations.markings,arrows,matrix}
\usetikzlibrary{calc,decorations.pathmorphing,decorations.markings,
decorations.pathreplacing,patterns,shapes,arrows}
\tikzstyle{block} = [rectangle,draw,text width=10em,text centered,rounded corners,minimum height=4em]
\tikzstyle{line} = [draw, -latex']
\allowdisplaybreaks

\theoremstyle{plain}

\newtheorem{defn}{Definition}
\newtheorem{lem}{Lemma}
\newtheorem{thm}{Theorem}
\newtheorem*{thm*}{Theorem}

\newtheorem{cor}{Corollary}
\newtheorem*{cor*}{Corollary}
\newtheorem{rmk}{Remark}
\newtheorem{prop}{Proposition}
\newtheorem*{prop*}{Proposition}

\numberwithin{equation}{section}
\numberwithin{defn}{section}
\numberwithin{lem}{section}
\numberwithin{thm}{section}
\numberwithin{conj}{section}
\numberwithin{cor}{section}
\numberwithin{rmk}{section}
\numberwithin{prop}{section}
\numberwithin{ex}{section}

\makeatletter
\DeclareFontFamily{OMX}{MnSymbolE}{}
\DeclareSymbolFont{MnLargeSymbols}{OMX}{MnSymbolE}{m}{n}
\SetSymbolFont{MnLargeSymbols}{bold}{OMX}{MnSymbolE}{b}{n}
\DeclareFontShape{OMX}{MnSymbolE}{m}{n}{
    <-6>  MnSymbolE5
   <6-7>  MnSymbolE6
   <7-8>  MnSymbolE7
   <8-9>  MnSymbolE8
   <9-10> MnSymbolE9
  <10-12> MnSymbolE10
  <12->   MnSymbolE12
}{}
\DeclareFontShape{OMX}{MnSymbolE}{b}{n}{
    <-6>  MnSymbolE-Bold5
   <6-7>  MnSymbolE-Bold6
   <7-8>  MnSymbolE-Bold7
   <8-9>  MnSymbolE-Bold8
   <9-10> MnSymbolE-Bold9
  <10-12> MnSymbolE-Bold10
  <12->   MnSymbolE-Bold12
}{}

\let\llangle\@undefined
\let\rrangle\@undefined
\DeclareMathDelimiter{\llangle}{\mathopen}%
                     {MnLargeSymbols}{'164}{MnLargeSymbols}{'164}
\DeclareMathDelimiter{\rrangle}{\mathclose}%
                     {MnLargeSymbols}{'171}{MnLargeSymbols}{'171}
\makeatother

\def\bbZ{\mathbb{Z}}
\def\bbC{\mathbb{C}}

\def\leq{\leqslant}
\def\geq{\geqslant}

\def\id{\text{id}}

\def\dd{\textup{d}}

\definecolor{lgray}{cmyk}{0,0,0,0.2}

\definecolor{yellow}{rgb}{1.0,0.85,0.25}

\newcommand{\ba}{\[\begin{aligned}~}
\newcommand{\ea}{\end{aligned}\]}

\newcommand{\Z}{\mathbb{Z}}

\newcommand{\N}{\mathbb{N}}

\newcommand{\ii}{\mathrm{i}}

\newcommand{\inp}{p^{(0)}}

\newcommand*{\Scale}[2][4]{\scalebox{#1}{$#2$}}
\newcommand{\Pcross}{\mathbb{P}_{s_1,s_2}(\mu;t)}
\newcommand{\PB}{\mathbb{P}_{s_1,s_2}^{\textrm{B}}(t)}
\newcommand{\PS}{\mathbb{P}_{s_1,s_2}^{\textrm{S}}(t)}

\renewcommand{\tikz}[2]{
\begin{tikzpicture}[scale=#1,baseline=(current bounding box.center),>=stealth]
#2
\end{tikzpicture}}

\def\I{\bm{I}}

\def\K{\bm{K}}

\newcommand{\Is}[2]{\I_{[#1,#2]}}

\makeatletter
\g@addto@macro{\endabstract}{\@setabstract}
\newcommand{\authorfootnotes}{\renewcommand\thefootnote{\@fnsymbol\c@footnote}}%
\makeatother

\begin{document}
		\begin{center}
			\LARGE 
			Transition probability and total crossing events in the multi-species asymmetric exclusion process \par \bigskip
			
			\normalsize
			\authorfootnotes
			Jan de Gier \footnote{\href{mailto:jdgier@unimelb.edu.au}{jdgier@unimelb.edu.au}}, William Mead\footnote{\href{mailto:wmead@student.unimelb.edu.au}{wmead@student.unimelb.edu.au}} and
			Michael Wheeler \footnote{\href{mailto:wheelerm@unimelb.edu.au}{wheelerm@unimelb.edu.au}} \par \bigskip
			
			ARC Centre of Excellence for Mathematical and Statistical Frontiers (ACEMS), School of Mathematics and Statistics, University of Melbourne, Victoria 3010, Australia\par \bigskip
			
			June 1, 2023
		\end{center}
	
	\setcounter{footnote}{0}

\begin{abstract}
We present explicit formulas for total crossing events in the multi-species asymmetric exclusion process ($r$-ASEP) with underlying $U_q(\widehat{\mathfrak{sl}}_{r+1})$ symmetry. In the case of the two-species TASEP these can be derived using an explicit expression for the general transition probability on $\mathbb{Z}$ in terms of a multiple contour integral derived from a nested Bethe ansatz approach. For the general $r$-ASEP we employ a vertex model approach within which the probability of total crossing can be derived from partial symmetrization of an explicit high rank rainbow partition function. In the case of $r$-TASEP, the total crossing probability can be show to reduce to a multiple integral over the product of $r$ determinants. For $2$-TASEP we additionally derive convenient formulas for cumulative total crossing probabilities using Bernoulli-step initial conditions for particles of type 2 and type 1 respectively.
\end{abstract}

\section{Introduction}
\label{se:introduction}
The asymmetric simple exclusion processes (ASEP) was originally introduced as a biophysical model for protein synthesis on RNA \cite{Macdonald68,Macdonald69}. As a many-body system where particles perform asymmetric random walks with lattice site exclusion, it exhibits interesting phase behaviour and collective phenomena related to the Kardar-Parisi-Zhang (KPZ) universality class \cite{KPZ1986,BS1995}. The ASEP is considered to be one of the most fundamental processes in the theory of stochastic interacting particle systems \cite{Liggett1985,Liggett1999,Spohn1991}. 

The ASEP is an integrable model \cite{Golinelli_2006} and closely related to the solvable six-vertex lattice model. Bethe ansatz studies of the spectral gap of the asymmetric simple exclusion process (ASEP) on finite geometries with periodic and open boundary conditions have confirmed the KPZ dynamical exponent \cite{PhysRevA.46.844,PhysRevLett.68.725,PhysRevE.52.3512,Gier2005,Gier2006,Gier2008}. The integrability of the ASEP also allows one to solve the master equation exactly and write down a closed form expression for its time-dependent transition probability, or Green's function \cite{schutz_exact_1997,tracy2008integral}. Many recent studies have utilized the integrability to study non-Gaussian fluctuations in the ASEP. In 2000, Johansson showed that the current distribution of the totally asymmetric simple exclusion process (TASEP) on $\Z$ with the step initial condition is given by the GUE Tracy-Widom distribution in the long time limit \cite{J2000}, and that the fluctuation exponent is equal to 1/3, characteristic of the one-dimensional KPZ class. In two dimensions the ASEP has a $\log(t)^{2/3}$ law \cite{Y2004}. There have been a number of generalisations of limit laws for various types of asymmetric exclusion processes, including the ASEP under  several different initial conditions see e.g. \cite{BaikRains00,PS2002,NS2004,Sasamoto2005,TW2009a,Sasamoto2010a,Sasamoto2010b,Sasamoto2010c,bcs,Amir2010,CalabreseDoussal,BC2014}.

Most of these results on limiting distributions have been established based on explicit multiple integral expressions related to random matrix theory \cite{J2000,NS2004,BC2014,bcs}, from which asymptotics can be derived, for example using steepest decent analysis of Fredholm determinant expressions. Almost all models for which limiting distributions have been studied so far have turned out to be a special or a limiting case of a stochastic higher spin six vertex model \cite{MANGAZEEV201470,povolotsky2013,Corwin_2015,BP2016}, whose similarity transformed non-stochastic version has long been known to be (Yang-Baxter) integrable \cite{Kulish1981YangBaxterEA,Kirillov_1987}. 

\subsection{Multi-species exclusion processes}
Much less is known rigorously for multi-species asymmetric exclusion processes. These multi-species models were introduced a long time ago, see e.g. \cite{ferrari_microscopic_1991,Alcaraz_Stochastic_1998,PhysRevE.59.205,Mallick_1999,Derrida_exactsolution,Schutz_Critical_2003}, and several other multi-species stochastic processes have been proposed recently in \cite{kuniba2016multispecies}. In this paper we study asymmetric exclusion processes based on $U_q(\widehat{\mathfrak{sl}}_{r+1})$, which have been most studied. In particular we consider the $r$-ASEP, which has $r$ distinguishable particle species which all evolve in isolation like the ASEP. However, when these different species are adjacent they may interchange positions with rate dictated by the higher species. Their stationary measures were given in \cite{ferrari2007,PEM2009} under periodic boundary conditions and in \cite{Belitsky_Self_2018} under reflecting conditions. While in \cite{Cantini_2015} these were put in the context of Macdonald polynomial theory, the Knizhnik-Zamolodchikov equation. The $n$-TASEP was shown to be related to the combinatorial $R$-matrix and solutions of the tetrahedron equation in \cite{kuniba2016multispecies,Kuniba_2016}. Early work on dynamical properties of multi-species models included dynamic matrix product ansatz analysis \cite{Popkov_2002}, a nested algebraic Bethe ansatz derivation of the Bethe equations \cite{Cantini_Algebraic_2008}, as well as results on Bethe ansatz transition and non-crossing probabilities \cite{Chatterjee_2010}.

In \cite{ChenGW,Kuan_2018,Kuan_2021} methods have been developed to construct multi-species duality functionals, and many other algebraic properties and connections to non-symmetric Macdonald polynomials and partition functions are discussed in \cite{borodin_coloured_2018}. Dualities for the multi-species ASEP are also derived in \cite{Belitsky_Self_2018}. More recently the transition probability and limit laws have been discussed in two-species models with a small number of particles of a second type \cite{tracy2013asymmetric,kuan2019probability,lee2018exact,FNG2019,Nejjar2019}. The fluctuation exponent of $n$-ASEP was addressed in finite size scaling of the gap of its generator \cite{arita2009spectrum}, and the Bethe ansatz for the 2-ASEP with open boundaries was considered in \cite{Zhang_2019}. To our knowledge, full explicit limiting distributions for multi-species models have only been derived recently for the two-species Arndt-Heinzl-Rittenberg (AHR) model \cite{chen_exact_2018,chen_exact_2019,chen_exact_2021}, or in cases where there is a relation to a single species model such as shift colour-position symmetry \cite{borodin2019colorposition,borodin2020shiftinvariance} and the coalescing random walk \cite{kuan_determinantal_2020}. 

In this paper we give explicit multiple integral expressions for the transition probability of the two-species asymmetric exclusion process with arbitrary number of particles of type 1 and type 2, which has an underlying $U_q(\widehat{\mathfrak{sl}}_3)$ symmetry, as well as a representation of the general $r$-ASEP transition probability as a partition function. We furthermore give multiple integral expressions with factorized kernels for certain total crossing probabilities. Our approache is two-fold. We first solve the master equation for the $2$-TASEP using the nested Bethe ansatz, which is a traditional approach to analyse higher rank solvable lattice models \cite{YY1966}. A second approach makes use of a general spin high rank solution of the Yang-Baxter equation \cite{Bosnjak_2016,Garbali2017,borodin_coloured_2018} from which it is possible to express the transition probability for the $n$-ASEP as a partition function, in the spirit of the general Macdonald stochastic processes setup \cite{BC2014}.

\subsection{Model outline}
We consider here a fixed number of two distinguishable species of particles on the infinite one-dimensional lattice $\bbZ$. We have $n$ particles in total, $m$ of which are type 2. This leaves $n-m$ particles of type 1. We dictate that these particle species evolve in continuous time with simple Poisson rates which define the two-species totally asymmetric simple exclusion process (2-TASEP) as considered in \cite{ferrari_microscopic_1991},
\begin{align*}
	(1,0) &\mapsto (0,1) \text{ at rate 1},\\
	(2,0) &\mapsto (0,2) \text{ at rate 1},\\
	(2,1) &\mapsto (1,2) \text{ at rate 1}.
\end{align*}

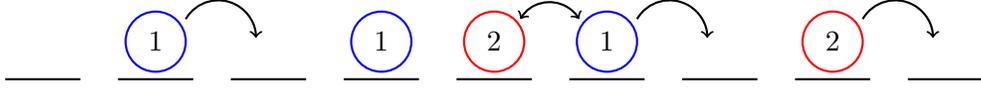
\begin{figure}
	\begin{center}
		\begin{tikzpicture}
		\draw[thick] (0,0) -- (1,0);
		\draw[thick] (1.5,0) -- (2.5,0);
		\draw[thick] (3,0) -- (4,0);
		\draw[blue,thick] (2,0.5) circle (0.4cm);
		\node (1) at (2,0.5) {$1$};
		\draw[thick,->] (2.4,0.8) arc (150:0:0.5);
		
		\draw[thick] (4.5,0) -- (5.5,0);
		\draw[thick] (6,0) -- (7,0);
		\draw[thick] (7.5,0) -- (8.5,0);
		\draw[thick] (9,0) -- (10,0);
		\draw[blue,thick] (5,0.5) circle (0.4cm);
		\node (1) at (5,0.5) {$1$};
		\draw[red,thick] (6.5,0.5) circle (0.4cm);
		\node (2) at (6.5,0.5) {$2$};
		\draw[thick,<->] (6.85,0.8) arc (150:30:0.45);	
		\draw[blue,thick] (8,0.5) circle (0.4cm); 
		\node (1) at (8,0.5) {$1$};
		\draw[thick,->] (8.4,0.8) arc (150:0:0.5);
		
		\draw[thick] (10.5,0) -- (11.5,0);
		\draw[thick] (12,0) -- (13,0);	
		\draw[red,thick] (11,0.5) circle (0.4cm);
		\node (2) at (11,0.5) {$2$};	
		\draw[thick,->] (11.4,0.8) arc (150:0:0.5); 	 	 		 		
		\end{tikzpicture}
	\end{center}
	\caption{2-TASEP on a small lattice showing the possible evolution of the current state.}
	\label{2TASEP rules fig}
\end{figure}
This enforces that only type 2 particles may overtake type 1 particles but not vice versa as demonstrated in Figure \ref{2TASEP rules fig}. Note that if we could not distinguish between the two species of particle, then we would simply find the TASEP evolution. Let $\mathbb{P}(\mathcal{C};t)$ be the probability to find the 2-TASEP in configuration $\mathcal{C}$ at time $t$. The rate of change of $\mathbb{P}(\mathcal{C};t)$ is prescribed by balancing the gain and loss terms in and out of configuration $\mathcal{C}$ through the Markov chain master equation
\begin{equation}
\label{eq:Markovgeneral}
\frac{\dd}{\dd t} \mathbb{P}(\mathcal{C};t) = \sum_{\mathcal{C}'\neq \mathcal{C}} W(\mathcal{C}' \to \mathcal{C}) \mathbb{P}(\mathcal{C}';t) - \sum_{\mathcal{C}'\neq \mathcal{C}} W(\mathcal{C} \to \mathcal{C}') \mathbb{P}(\mathcal{C};t),
\end{equation}
where $W(\mathcal{C}' \to \mathcal{C})$ is the rate of an event from configuration $\mathcal{C}'$ into $\mathcal{C}$.

We now describe this time evolution in more detail. Let us denote the coordinates of all particles by $\nu = \{\nu_1, \dots, \nu_n\} \in \bbZ^n$ with $\nu_1 < \nu_2 < \ldots < \nu_n$, and the indices of type 2 particles by  $p = \{p_1, \dots, p_m\} \in \{1,\dots,n\}^m$ with $p_1 < p_2 < \ldots < p_m$. In other words, a particle at position $\nu_i$ is of type 2 if $i \in p$ and of type 1 if $i \notin p$. We furthermore denote by $\Psi(\nu,p;t)$ an un-normalized general solution of \eqref{eq:Markovgeneral}, where $(\nu,p)$ is a 2-TASEP in configuration at time $t$. 

When all particles are far enough apart, i.e $\nu_{i+1}-\nu_i \ge 2$ for $i=1,\ldots,n-1$, the time evolution of the 2-TASEP is given by the following free evolution equation,
\begin{equation}
\label{2TASEP master equation}
	\frac{\dd}{\dd t} \Psi(\nu,p ;t) = \sum_{i=1}^n \Psi(\{\nu_1,\dots,\nu_i-1, \dots, \nu_n\},p;t)  - n\Psi(\nu,p;t).
\end{equation}
This free evolution is of the same form as that of the single-species TASEP \cite{schutz_exact_1997}, and does not take exclusion into account, nor the possibility that type 2 particles may overtake those of type 1. These interactions can be accounted for by extending the validity of \eqref{2TASEP master equation} to include cases where $\nu_{i+1}-\nu_i=1$ and enforce the following boundary conditions for $1\leq i \leq n-1$,
\begin{enumerate}
	\item If $i \in p$ and $i+1 \notin p$ then
	\begin{equation}
	\label{2TASEP BC1}
		\Psi(\{\nu_1,\dots,\nu_i,\nu_{i+1}=\nu_i,\dots,\nu_n\},p;t)= 0.
	\end{equation}
	
	\item If $i \notin p$ and $i+1 \in p$ then
	\begin{multline}
		\label{2TASEP BC2}
			\Psi(\{\nu_1,\dots,\nu_i,\nu_{i+1}=\nu_i,\dots,\nu_n\},p;t) = \Psi(\{\nu_1,\dots,\nu_i,\nu_{i+1}=\nu_i+1,\dots,\nu_n\},p+\{i\} - \{i+1\} ;t)\\ 
		+ \Psi(\{\nu_1,\dots,\nu_i,\nu_{i+1}=\nu_i+1,\dots,\nu_n\},p;t).
	\end{multline}
	
	\item Otherwise
	\begin{equation}
	\label{2TASEP BC3}
		\Psi(\{\nu_1,\dots,\nu_i,\nu_{i+1}=\nu_i,\dots,\nu_n\},p;t) = \Psi(\{\nu_1,\dots,\nu_i,\nu_{i+1}=\nu_i+1,\dots,\nu_n\},p;t).
	\end{equation}
\end{enumerate}
The un-normalized probability $\Psi(\nu,p;t)$ then forms the general solution to the 2-TASEP master equation.

To interpret the boundary conditions we consider the two-particle system with $\Psi(\{\nu_1,\nu_1+1\},\{1\};t)$; the probability to observe the configuration with two neighbouring particles, one of type 2 at $\nu_1 \in \Z$ and another of type 1 at position $\nu_1+1$. The imposition of the first boundary condition \eqref{2TASEP BC1} corresponds to setting $\Psi(\{\nu_1,\nu_1\},\{1\};t) = 0$ so that equation \eqref{2TASEP master equation} becomes
\begin{equation*}
		\frac{\dd}{\dd t} \Psi\left(\{\nu_1,\nu_1+1\},\{1\};t\right) = \Psi(\{\nu_1-1,\nu_1+1\},\{1\};t) - 2 \Psi(\{\nu_1,\nu_1+1\},\{1\};t).
\end{equation*}
This says that the system may only evolve into the configuration $\nu = \{\nu_1,\nu_1+1\}, p = \{1\}$ by the type 2 particle stepping from $\nu_1-1$ to $\nu_1$, and that there are two ways to evolve out of configuration $\nu$ at rate 1. These are the type 1 particle stepping $\nu_1+1 \mapsto \nu_1+2$ and the two particles interchanging positions.

To examine the second boundary condition we consider $\Psi(\{\nu_1,\nu_1+1\},\{2\};t)$; the same as above but with the particle types swapped. The imposition of boundary condition \eqref{2TASEP BC2} leads the the master equation becoming
\begin{equation*}
		\frac{\dd}{\dd t} \Psi(\{\nu_1,\nu_1+1\},\{2\};t) = \Psi(\{\nu_1-1,\nu_1+1\},\{2\};t) +\Psi(\{\nu_1,\nu_1+1\},\{1\};t) - \Psi(\{\nu_1,\nu_1+1\},\{2\};t),
\end{equation*}
which dictates that there are two ways to evolve into the configuration $\nu$ and only one way to evolve out of it.

The third boundary condition \eqref{2TASEP BC3} is the interaction between two particles of the same type and is identical to the TASEP boundary condition which enforces the general exclusion property between two indistinguishable particles. 

The $r$-ASEP is Yang-Baxter integrable and its generator is constructed from the R-matrix of an underlying $U_q(\widehat{\mathfrak{sl}}_{r+1})$ algebra \cite{AlcarazBariev00,ALCARAZ1993377,PerkSchultz}. It is a consequence of the integrability of $r$-ASEP that all multi-particle interactions factorize into two-particle interactions. The simultaneous imposition of the two-particle boundary conditions therefore enforces the general true time evolution of the 2-TASEP.

\subsection{Main results}
In this section we summarize our main results. The first key result is an explicit expression for the transition probability of the 2-TASEP for arbitrary initial conditions, which we prove in Section~\ref{2TASEP integral formula section}  using the nested Bethe ansatz.  
\begin{thm}[Theorem~\ref{2TASEP Green's function Theorem} in the main text]
	The transition probability for the 2-TASEP with $n$ particles of which $m$ are type 2, initially at positions $\mu$ and in order $p^{(0)}$, and transitioning to positions $\nu$ and order $p$ is given by the transition probability
	\begin{multline}
		\label{2TASEP Greens intro}
		\mathbb{P}(\mu\to\nu,\inp\to p;t) = \\ \oint \prod_{j=1}^m \frac{\dd u_j}{2 \pi \ii}  \prod_{i=1}^n \frac{\dd z_i}{2 \pi \ii}\ P(\nu,p;t)\
		\prod_{i=1}^n z_i^{-\mu_i -1} (1-z_i)^m\prod_{i=1}^m\left[\prod_{j=1}^{\inp_i} \frac{1}{u_i - z_j}\prod_{j=\inp_i+1}^n\frac{1}{1 - z_j}\right],
	\end{multline}
	where $P(\nu,p;t)$ is an eigenfunction of the 2-TASEP generator satisfying the free evolution \eqref{2TASEP master equation} and three boundary conditions, given by 
	\begin{multline}
		\label{2TASEP_P intro}
		P(\nu,p;t) = \sum_{\pi \in S_n} (-1)^{\abs{\pi}} \prod_{i=1}^n \left(\frac{1-z_i}{1-z_{\pi_i}}\right)^i z_{\pi_{i}}^{\nu_{i}} e^{(z^{-1}_i-1)t} \prod_{i=1}^m \prod_{j=1}^{p_i} \frac{1}{1-z_{\pi_j}}
		\\ \times\sum_{\sigma \in S_m} (-1)^{\abs{\sigma}} \prod_{i=1}^m \left(\frac{1-u_i}{1-u_{\sigma_i}} \right)^i  \prod_{i=1}^m \prod_{j=1}^{p_i-1} (u_{\sigma_i} - z_{\pi_j}).
	\end{multline}
	The $z$-contours in the integrals surround the origin but do not enclose the finite poles at $z_i = 1, u_j$ for all $i,j$. The $u$-contours enclose the $z$-contours and residues are taken at the origin after evaluation of the $z$ integration. The transition probability \eqref{2TASEP Greens intro} satisfies the initial condition
	\begin{equation*}
		\mathbb{P}(\mu\to\nu,\inp\to p;0) = \prod_{i=1}^{n} \delta_{\nu_i,\mu_i} \prod_{i=1}^m \delta_{p_i,\inp_i}.
	\end{equation*} 
\end{thm}

\subsubsection{Results for $r$ species}
Formulas such as \eqref{2TASEP Greens intro} and \eqref{2TASEP_P intro} are fully explicit but still contain large sums over symmetric groups and therefore have a significant level of complexity. As we will see later, these formulas simplify for specific initial and final positions of the particles. To gain further insight into why such a simplification occurs, and how this generalizes to arbitrary number of species, we now turn to the vertex model approach developed in \cite{borodin_coloured_2018}.

We show in Section~\ref{se:vertex models} that the transition probability for a system in which all the particles have a different colour\footnote{We refer to this as the rainbow case.} can be realized as a partition function of a vertex model, also when backhopping is allowed, see Theorem~\ref{th:ASEPtransition}. A partial symmetrization then gives transition probabilities for multi-species asymmetric exclusion processes expressed as partition functions of a vertex model.  

A second key result of this paper is the transition probability for total crossing, for which the point of view of vertex models is particularly useful because the rainbow partition function has a convenient expression as a multiple integral with factorized kernel, see equation \eqref{n-colours-total}. For the $r$-species ASEP, a total crossing event is characterized by an initial condition in which all particles of type $r'$ are to the left of all particles of type $r'-1$ for all $r'=1,\ldots,r$, and a final condition at time $t$ in which this order is reversed. Using the vertex model approach of Section~\ref{se:vertex models}, the following expression for the probability of total crossing for the $r$-ASEP can be derived.
\begin{thm}[Theorem~\ref{thm:total-cross} in the main text]
	Fix an integer $r \geq 1$ and let $(n_1,\dots,n_r)$ be a vector of positive integers such that $\sum_{k=1}^{r} n_k = n$. Fix two vectors of strict signatures
	\begin{align*}
		\vec{\mu} &= \left( \mu^{(1)}, \dots, \mu^{(r)} \right),
		\qquad
		\mu^{(i)} = \left( \mu^{(i)}_1 > \cdots > \mu^{(i)}_{n_i} \right),
		\\
		\vec{\lambda} &= \left( \lambda^{(1)}, \dots, \lambda^{(r)} \right),
		\qquad
		\lambda^{(i)} = \left( \lambda^{(i)}_1 > \cdots > \lambda^{(i)}_{n_i} \right),
	\end{align*}
	which satisfy the constraints $\mu^{(i)}_a > \mu^{(j)}_b$, $\lambda^{(i)}_a < \lambda^{(j)}_b$ for all $1\leq i <j \leq r$ and all $a \in \{1,\dots,n_i\}$, $b \in \{1,\dots,n_j\}$ (these constraints imply, in particular, admissibility of both vectors). We then have the following formula:
	\begin{multline}
		\label{main-formula}
		\mathbb{P}(\vec{\mu} \rightarrow \vec{\lambda};t)
		=
		\left(\frac{1-q}{2\pi\ii}\right)^n
		\oint_{C}
		\dd z_1
		\cdots 
		\oint_{C}
		\dd z_n
		\prod_{1 \leq i<j \leq n}
		\left(
		\frac{z_j-z_i}{z_j-q z_i}
		\right)
		\\
		\times
		\prod_{j=1}^{n}
		\exp\left[\frac{(1-q)^2 z_j t}{(1-z_j)(1-q z_j)}\right]
		\frac{1}{(1-z_j)(1-qz_j)}
		\prod_{k=1}^{r}
		\xi_{\mu^{(k)}} 
		\left( z^{(k)}_1,\dots,z^{(k)}_{n_k} \right)
		{\sf F}_{\lambda^{(k)}}
		\left( z^{(k)}_1,\dots,z^{(k)}_{n_k} \right)
	\end{multline}
	where we have defined the rational functions 
	\begin{align}
		\nonumber
		\xi_{\mu}(u_1,\dots,u_N)
		&=
		\prod_{j=1}^{N}
		\left(
		\frac{1-qu_j}{1-u_j}
		\right)^{\mu_j},
		\\
		{\sf F}_{\lambda}(u_1,\dots,u_N)
		&=
		\sum_{\sigma \in S_N}
		\prod_{1 \leq i<j \leq N}
		\left( \frac{u_{\sigma_j} - q u_{\sigma_i}}{u_{\sigma_j} - u_{\sigma_i}} \right)
		\prod_{i=1}^{N}
		\left( \frac{1-u_{\sigma_i}}{1-qu_{\sigma_i}} \right)^{\lambda_i}.
		\label{F}
	\end{align}
	where $q$ is the rate for backhopping, the contours are small circles around $z_i=1$ for all $1\le i \le n$.
\end{thm}

A straightforward corollary brings us to the following result for the multi-species TASEP in which the integrand is expressed as a product of determinants. We will given an alternative proof of this corollary for 2-TASEP in Section~\ref{2TASEP integral formula section} via the transition probability \eqref{2TASEP Greens intro} obtained by the nested Bethe ansatz.

\begin{cor}[Corollary~\ref{cor:pTASEPcross} in the main text]
	\label{cor:pTASEPcross_intro}
	Taking $q = 0$, the permutation sums can be computed explicitly to give determinants. Let $\vec{\mu}$ and $\vec{\lambda}$ be vectors of strict signatures, prescribed in the same way as in the previous Theorem. The probability of total crossing of blocks in the multi-species TASEP is given by
	\begin{multline*}
		\mathbb{P}(\vec{\mu} \rightarrow \vec{\lambda};t)
		=
		\left(\frac{1}{2\pi\ii}\right)^n
		\oint_{C}
		\dd z_1
		\cdots 
		\oint_{C}
		\dd z_n
		\prod_{j=1}^{n}
		\exp\left[\frac{z_j t}{1-z_j}\right]
		\prod_{1 \leq k<\ell \leq r}
		\prod_{i=1}^{n_k}
		\prod_{j=1}^{n_{\ell}}
		\left( z^{(\ell)}_j-z^{(k)}_i \right)
		\\
		\times
		\prod_{k=1}^{r}
		\det\left[ z^{(k)\left(i-j-N_k\right)}_j 
		\left(1-z^{(k)}_j\right)^{\lambda^{(k)}_i-\mu^{(k)}_j-1} \right]_{1 \leq i,j \leq n_k}
	\end{multline*}
	where for each $1 \leq k \leq r$ we define $N_k = \sum_{a=1}^{k-1} n_a$.
\end{cor}

\subsubsection{Cumulative total crossing for 2-TASEP}
A third key result concerns the cumulative total crossing probability in 2-TASEP for Bernoulli initial condition for type 2 particles and step initial condition for type 1 particles, such that at time $t$ all particles of type 1 end up in between position $s_1$ and $s_2$, and all particles of type 2 end up beyond position $s_2$.
\begin{prop}[Proposition~\ref{prop:cumtotalcross} in the main text]
	The cumulative total crossing probability with Bernoulli initial condition with density $\rho$ for type 2 particles and step initial condition for type 1 particles is given by
	\begin{multline}
		\label{Bernoulli Pcross 2}
		\PB = \frac{\rho^m}{m!}\oint_{0,1,1-\rho} \prod_{i=1}^m \frac{\dd z_{i}}{2 \pi \ii} \prod_{i \neq j} (z_j-z_i) \prod_{i=1}^m \frac{e^{(z_i-1)t}z_i^{-s_2-m+1}}{(z_i-1)^{n}(z_i-1+\rho)}\\ 
		\times \oint_{0,1} \prod_{i=1}^{n-m} \frac{\dd w_{i}}{2 \pi \ii}  \prod_{i=1}^m\prod_{j=1}^{n-m} (z_i-w_j) \prod_{i=1}^{n-m} \frac{e^{(w_i-1)t}w_i^{-s_1-m}}{(w_i-1)^{n-m-i+1}} \\ \times \det\left(w_i^{n-m-j} - w_i^{n-m+s_1-s_2-1}\right)_{1\leq i,j \leq n-m}.
	\end{multline}
\end{prop}
If position $s_1\le -m$, the Bernoulli-step initial condition renders the dependence on $s_1$ irrelevant because type 1 particles in 2-TASEP can never reach this position. In this case the cumulative total crossing reduces to the following formula.
\begin{cor}[Corollary~\ref{cor:cumtotalcross} in the main text]
	When $s_1 \leq -m$ we have 
	\begin{equation}
	\PB = \rho^m \oint_{0}\frac{\dd w}{2\pi\ii} \frac{e^{(w-1)t}w^{n-2m-s_2-1}}{w-1} 
	\det(\oint_{0,1,1-\rho}\frac{\dd z}{2\pi\ii} \frac{e^{(z-1)t}z^{i+j-s_2-m-1}}{(z-1)^{m+1}(z-1+\rho)}(w-z))_{1\leq i,j \leq m}.
	\end{equation}
\end{cor}

The remainder of this paper is set up as follows. In Section~\ref{2TASEP integral formula section} we describe the nested Bethe ansatz approach to diagonalize the 2-TASEP generator. This approach provides an explicit expression of the full transition probability of 2-TASEP. In Section~\ref{se:vertex models} we describe the rainbow vertex model approach in which particles all have different colour. Partial symmetrization then leads to results for a model of particles of two different types. 

As highlighted above, Corollary~\ref{cor:pTASEPcross} of the vertex model approach leads to a multiple integral expression over determinants for the transition probability of total crossing, events in which a sequence of blocks of different types of particle reverses. The nested Bethe ansatz approach reproduces this result in the case of 2-TASEP, see Proposition~\ref{2TASEP crossing det formula statement}. In Section~\ref{2TASEP step initial condition section} we study the cumulative distribution function of such events for Bernoulli-step initial conditions.

\section*{Acknowledgments} We warmly thank Zeying Chen, Mario Kieburg and Tomohiro Sasamoto for discussions, and gratefully acknowledge financial support from the Australian Research Council.

\section{Nested Bethe ansatz approach}
\label{2TASEP integral formula section}
In this section we find an exact integral formula for the transition probability (or Green's function) of the 2-TASEP. We also outline some consequences of this formula as well as some useful results for later when we study total crossing events.

\begin{defn}
	The standard regime of a system of $n$ particles, $m$ of which are type 2, with initial coordinates $\mu \in \Z^n$, and type 2 indices $\inp \in \{1,\dots,n\}^m$, satisfying
	\[\mu_1 < \mu_2 < \cdots < \mu_n, \hspace{2cm} \inp_1 < \inp_2 < \cdots < \inp_m, \] 
	is the set of final particle coordinates and indices $(\nu,p)$. These satisfy
	\[\nu_1 < \nu_2 < \cdots < \nu_n, \hspace{2cm} p_1 < p_2 < \cdots < p_m,\]
	along with $\nu_i\geq\mu_i$, $p_j\geq \inp_j$ for all $i,j$.
\end{defn} 

\begin{thm}
	\label{2TASEP Green's function Theorem}	
	The transition probability for the 2-TASEP with $n$ particles of which $m$ are type 2, initially at positions $\mu$ and transitioning to $\nu$ in the standard regime is given by the transition probability
	\begin{multline}
		\label{2TASEP Greens with P}
		\mathbb{P}(\mu\to\nu,\inp\to p;t) = \\ \oint \prod_{j=1}^m \frac{\dd u_j}{2 \pi \ii}  \prod_{i=1}^n \frac{\dd z_i}{2 \pi \ii}\ P(\nu,p;t)\
		\prod_{i=1}^n z_i^{-\mu_i -1} (1-z_i)^m\prod_{i=1}^m\left[\prod_{j=1}^{\inp_i} \frac{1}{u_i - z_j}\prod_{j=\inp_i+1}^n\frac{1}{1 - z_j}\right],
	\end{multline}
	where $P(\nu,p;t)$ is an eigenfunction of the 2-TASEP generator satisfying the free evolution \eqref{2TASEP master equation} and three boundary conditions, given by
	\begin{multline}
		\label{2TASEP_P}
		P(\nu,p;t) = \sum_{\pi \in S_n} (-1)^{\abs{\pi}} \prod_{i=1}^n \left(\frac{1-z_i}{1-z_{\pi_i}}\right)^i z_{\pi_{i}}^{\nu_{i}} e^{(z^{-1}_i-1)t} \prod_{i=1}^m \prod_{j=1}^{p_i} \frac{1}{1-z_{\pi_j}}
		\\ \times\sum_{\sigma \in S_m} (-1)^{\abs{\sigma}} \prod_{i=1}^m \left(\frac{1-u_i}{1-u_{\sigma_i}} \right)^i  \prod_{i=1}^m \prod_{j=1}^{p_i-1} (u_{\sigma_i} - z_{\pi_j}).
	\end{multline}
	The eigenfunction $P(\nu,p;t)$ has an implicit dependence on the spectral parameters $\{z_i\},\{u_i\}\subset \mathbb{C}$. The integrals in \eqref{2TASEP Greens with P} are set-up as in Figure \ref{fig:2TASEP u,z contour diagram}, i.e. the $z$-contours surround the origin but do not enclose the finite poles at $z_i = 1,u_j$ for all $i,j$. The transition probability \eqref{2TASEP Greens with P} satisfies the initial condition
	\begin{equation}
	\label{2TASEP initial condition}
	\mathbb{P}(\mu\to\nu,\inp\to p;0) = \prod_{i=1}^{n} \delta_{\nu_i,\mu_i} \prod_{i=1}^m \delta_{p_i,\inp_i}.
	\end{equation} 
\end{thm}
Due its length and technical nature we defer the proof of Theorem \ref{2TASEP Green's function Theorem} until Appendix \ref{2TASEP Green's function proof section}.

\begin{rmk}
	The above theorem has been derived using the nested Bethe ansatz approach which is the standard method to diagonalize integrable higher rank integrable Hamiltonians and transfer matrices. In Section~\ref{se:vertex models} we will discuss a more recent approach suitable for stochastic generators.
\end{rmk}

\begin{cor}
	\label{cor:Cor2TASEP}
	When $p$ is empty, that is when all particles are of type 1, the transition probability \eqref{2TASEP Greens with P} satisfying the conditions in Theorem \ref{2TASEP Green's function Theorem} reduces to the TASEP transition probability from \cite{schutz_exact_1997}.
\end{cor}

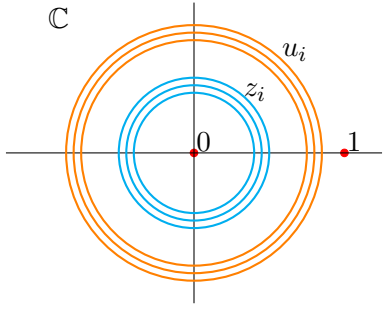
\begin{figure}
	\begin{center}
		\begin{tikzpicture}
		\node (1) at (0.13,0.15) {$0$};
		\node (2) at (2.13,0.15) {$1$};
		\node (3) at (0.81,0.81) {$z_i$};
		\node (4) at (1.34,1.34) {$u_i$};
		\node (5) at (-1.8,1.8) {$\bbC$};
		
		\draw[red,fill=red] (0,0) circle (0.05cm);
		\draw[red,fill=red] (2,0) circle (0.05cm);
		
		\draw[black] (-2.5,0) -- (2.5,0);
		\draw[black] (0,-2) -- (0,2);
		
		\draw[cyan,thick] (0,0) circle (0.8cm);
		\draw[cyan,thick] (0,0) circle (0.9cm); 
		\draw[cyan,thick] (0,0) circle (1cm);
		
		\draw[orange,thick] (0,0) circle (1.5cm);
		\draw[orange,thick] (0,0) circle (1.6cm); 
		\draw[orange,thick] (0,0) circle (1.7cm);    	 	 	 		 		
		\end{tikzpicture}
	\end{center}
	\caption{Diagram showing the setup of the integrals of the transition probability.}
	\label{fig:2TASEP u,z contour diagram}
\end{figure} 

The integration contours for the transition probability \eqref{2TASEP Greens with P} are shown in Figure \ref{fig:2TASEP u,z contour diagram}. Note that we can choose to intergrate along the $z$-contours around $0$ first, and then around the $u$-contours surrounding $0$. Alternatively we can integrate along the $u$-contours first around the poles $u_i=z_j$ before performing any integration in the $z$-variables. Because the $u$-singularities initially enclosed by the contours are simple poles the $u$ integrals are easily evaluated.

\begin{cor}
	For $n$ particles when $\inp=p = \{1,\dots, n\}$, that is when all particles are of type 2, then the transition probability \eqref{2TASEP Greens with P} reduces to the TASEP transition probability from \cite{schutz_exact_1997}.
\end{cor}

\begin{proof}
	Setting $n=m$ and  $p_j = \inp_j = j$ generates the following factor in the integrand of \eqref{2TASEP Greens with P} after simplification \[\prod_{1 \leq i < j \leq n} \frac{u_{\sigma_i} - z_{\pi_j}}{u_i - z_j}\prod_{i=1}^n \frac{1}{u_i - z_i}.\]
	It then follows when integrating in the $u$-variables first in the order $u_1, \dots, u_m$ that the only term which contributes to the integration is $\sigma = \pi$. The result follows with the residue theorem. 
\end{proof}

We can further simplify our computation of the transition probability under certain conditions using the following lemma, which shows that some of the singularities in the $u$-variables are removable.
\begin{lem}
	\label{2TASEP u_i = z_i residue lemma}
	When all type 2 particles are initially to the left of all type 1 particles, that is $\inp_i = i$ for $1 \leq i \leq m$ and performing the integration in the $u$-variables first, the singularities at $u_j = z_i$ for $1 \leq i \leq j-1$ are removable.
\end{lem} 
\noindent
Lemma~\ref{2TASEP u_i = z_i residue lemma} is particularly powerful because the task of evaluating the $u$-integrals now amounts to evaluating the residue of a single simple pole in each $u$-variable. The proof of this result is technical and is deferred until Appendix \ref{u_i=z_i proof section} as a result.

\subsection{Total crossing events}

In this section we focus on the probability of total crossing events using the 2-TASEP transition probability \eqref{2TASEP Greens with P} obtained from the nested Bethe ansatz. We will recover the explicit expression as a multiple integral over the product of two determinants as in Corollary~\ref{cor:pTASEPcross_intro}, which was derived using the methods we will discuss in Section~\ref{se:vertex models}. The nested Bethe ansatz method was also used for the integrable Arndt-Heinzel-Rittenberg (AHR) two-species exclusion processes in \cite{chen_exact_2018}.

To characterize a total crossing event we impose  conditions on the indices $\inp_i =i,p_j=n-m+j$, so that we initially have all type 2 particles to the left of those of type 1. An example of the total crossing a space-time trajectories is given in Figure \ref{Total crossing event fig}. We are then interested in calculating the probability of total crossing.
\begin{figure}
	\begin{center}
		\begin{tikzpicture}
		\draw[->] (0,0) -- (14,0);
		\draw[->] (0,0) -- (-1,0);
		\draw[->] (0,0) -- (0,6);
		
		\draw[fill=red,thick] (1,0) circle (0.3cm); 
		\draw[fill=red,thick] (2,0) circle (0.3cm); 
		\draw[fill=red,thick] (3,0) circle (0.3cm);
		\draw[fill=blue,thick] (4,0) circle (0.3cm); 
		\draw[fill=blue,thick] (5,0) circle (0.3cm); 
		\draw[fill=blue,thick] (6,0) circle (0.3cm);  
		\draw[fill=blue,thick] (7,0) circle (0.3cm);  
		
		\draw[thick,red,->] (1,0) -- (4,3) -- (5.2,3.2) -- (5.8,3.8) -- (7,4) -- (8.2,5.2) -- (9.4,5.4) -- (9.6,5.6) -- (10.8,5.8) -- (11,6);
		\draw[thick,red,->] (2,0) -- (3.5,1.5) -- (4.7,1.7) -- (6,3) -- (7.2,3.2) -- (8.5,4.5) -- (9.7,4.7) -- (10,5) -- (11.2,5.2) -- (12,6);
		\draw[thick,red,->] (3,0) -- (4,1) -- (5.2,1.2) -- (5.8,1.8) -- (7,2) -- (8,3) -- (9.2,3.2) -- (10,4) -- (10.2,4.2) -- (11.4,4.4) -- (13,6);
		\draw[thick,blue,->] (4,0) -- (5,1) -- (4.2,1.2) -- (4.5,1.5) -- (3.7,1.7) -- (5,3) -- (4.2,3.2) -- (7,6);
		\draw[thick,blue,->] (5,0) -- (6.8,1.8) -- (6,2) -- (7,3) -- (6.2,3.2) -- (6.8,3.8) -- (6,4) -- (8,6);
		\draw[thick,blue,->] (6,0) -- (7,1) -- (9,3) -- (8.2,3.2) -- (9.5,4.5) -- (8.7,4.7) -- (9.2,5.2) -- (8.4,5.4) -- (9,6);
		\draw[thick,blue,->] (7,0) -- (11.2,4.2) -- (10.4,4.4) -- (11,5) -- (10.2,5.2) -- (10.6,5.6) -- (9.8,5.8) -- (10,6);		
		
		\node (x) at (14.2,0) {$\nu$};
		\node (t) at (0.2,6) {$t$};
		\end{tikzpicture}
	\end{center}
	\caption{Typical behaviour of a total crossing event with step initial condition on a space-time diagram.}
	\label{Total crossing event fig}
\end{figure}
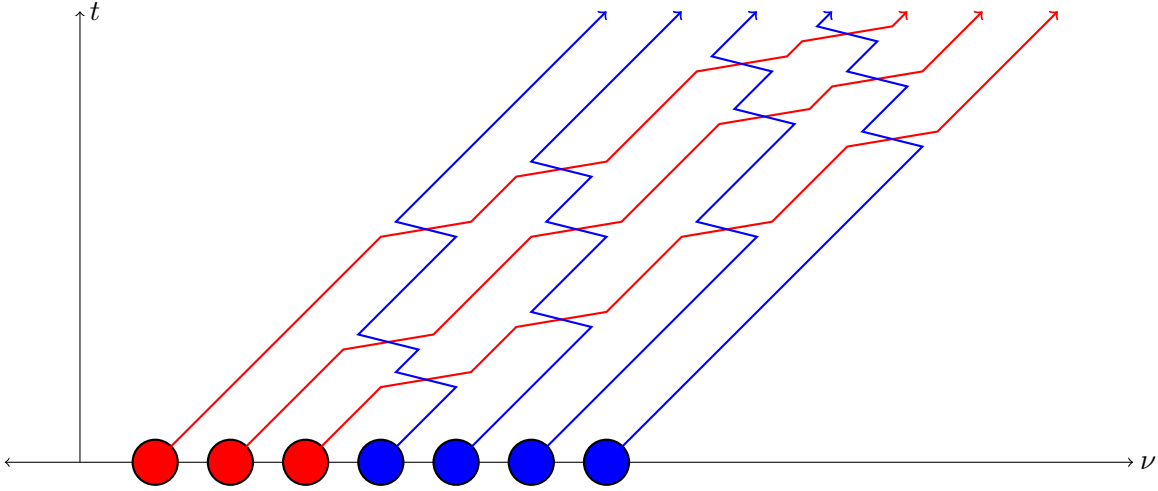

We first state the following symmetrization identity.
\begin{lem}
	\label{2TASEP u factorisation lemma}
	For all $\pi \in S_n$, when $p_j = n-m+j$ the $u$-dependence of the eigenfunction $P(\nu,p;t)$ from \eqref{2TASEP_P} factorizes as 
	\begin{multline}
		\label{2TASEP u factorisation equation}
		\sum_{\sigma \in S_m} (-1)^\abs{\sigma} \prod_{i=1}^{m}\left(\frac{1-u_{i}}{1-u_{\sigma_i}}\right)^{i} \prod_{i=1}^{m} \prod_{j=1}^{p_{i}-1}\left(u_{\sigma_{i}}-z_{\pi_{j}}\right) \\
		=\prod_{i=1}^{m} \prod_{j=1}^{n-m}\left(u_{i}-z_{\pi_{j}}\right) \frac{\prod_{1 \leq i < j \leq m}\left(u_{i}-u_{j}\right)}{\prod_{i=1}^{m} \left(1-u_i\right)^{m-i}} \prod_{i=1}^{m-1}\left(z_{\pi_{n-m+i}}-1\right)^{m-i}.
	\end{multline}
\end{lem}
\begin{proof}
	We first define 
	\begin{equation}
	\label{2 TASEP u facorisation lemma proof f}
	f(u_1,\dots,u_m) = \sum_{\sigma \in S_m} (-1)^\abs{\sigma} \prod_{i=1}^{m}\left(\frac{1-u_{i}}{1-u_{\sigma_i}}\right)^{i} \prod_{j=1}^{i-1}\left(u_{\sigma_{i}}-z_{\pi_{n-m+j}}\right),
	\end{equation}
	and note that the left-hand side of \eqref{2TASEP u factorisation equation} is equal to $	f(u_1,\dots,u_m) \prod_{i=1}^m \prod_{j=1}^{n-m} \left(u_i - z_{\pi_j}\right)$. We then proceed to factorize $f$, which is a rational function in $u_1, \dots , u_m$ where we regard the $z$-variables as constants. This gives 
	\begin{equation}
	f(u_1,\dots,u_m)  = \frac{g(u_1,\dots,u_m)}{\prod_{i=1}^m (1-u_i)^{m -i}},
	\end{equation}
	where
	\begin{equation}
	\label{2TASEP u factorisation lemma proof P}
	g(u_1,\dots,u_m)  =  \sum_{\sigma \in S_m} (-1)^\abs{\sigma} \prod_{j=1}^m (1-u_{\sigma_j})^{m-j} \prod_{k=1}^{j-1} (u_{\sigma_j} - z_{\pi_{n-m+k}}).
	\end{equation}
	We note that $g$ is a polynomial of order $m-1$ in all variables $u_1, \dots, u_m$, and we will now find all of its zeros.
	
	Let $1 \leq a < b \leq m$ and consider the limit of $g$ as $u_a \to u_b$. For a fixed permutation $\sigma \in S_m$, let $\alpha, \beta \in \{1, \dots,m\}$ so that $\sigma_\alpha = a$ and $\sigma_\beta = b$. We then consider the $u_a,u_b$-dependence of the term in \eqref{2TASEP u factorisation lemma proof P} associated with $\sigma$
	\begin{multline}
		\label{2TASEP u factorisation lemma limit 1}
		\lim_{u_a \to u_b} (1-u_a)^{m-\alpha}(1-u_b)^{m-\beta} \prod_{k=1}^{\alpha-1} (u_a - z_{\pi_{n-m+k}}) \prod_{k=1}^{\beta-1} (u_b - z_{\pi_{n-m+k}}) 
		\\ = (1-u_b)^{2m - \alpha - \beta} \prod_{k=1}^{\alpha -1 } (u_b - z_{\pi_{n-m+k}})^2 \prod_{k=\alpha}^{\beta -1 } (u_b - z_{\pi_{n-m+k}}),
	\end{multline}
	where we have assumed without loss of generality that $\alpha < \beta$. Then let $T_{\alpha,\beta}\in S_m$ be the transposition of $\alpha$ and $\beta$. We consider the the same limit as above but now for the term associated with the  permutation $ \sigma T_{\alpha, \beta}$. This limit evaluates as
	\begin{multline}
		\label{2TASEP u factorisation lemma limit 2}
		\lim_{u_a \to u_b} (1-u_a)^{m-\beta}(1-u_b)^{m-\alpha} \prod_{k=1}^{\beta-1} (u_a - z_{\pi_{n-m+k}}) \prod_{k=1}^{\alpha-1} (u_b - z_{\pi_{n-m+k}}) 
		\\ = (1-u_b)^{2m - \alpha - \beta} \prod_{k=1}^{\alpha -1 } (u_b - z_{\pi_{n-m+k}})^2 \prod_{k=\alpha}^{\beta -1 } (u_b - z_{\pi_{n-m+k}}).
	\end{multline}
	We then notice that the limits \eqref{2TASEP u factorisation lemma limit 1},\eqref{2TASEP u factorisation lemma limit 2} are equal, and noting that the signature satisfies $(-1)^\abs{\sigma T_{\alpha, \beta}} = -(-1)^{\abs{\sigma}}$ we may conclude that for all $a \neq b$ the term associated $\sigma$ cancels with the one associated with $\sigma T_{\alpha, \beta}$. This leads to
	\[\lim_{u_a \to u_b} g(u_1,\dots,u_m) = 0.\]
	We conclude that $g$ has zeros at $u_i = u_j$ for all $i \neq j$ which accounts for all $m-1$ zeros in each $u$-variable. This gives 
	\begin{equation}
	\label{2TASEP u factorisation lemma proof c}
	g(u_1,\dots,u_m) = c \prod_{1 \leq i < j \leq m} (u_i - u_j),
	\end{equation}
	for some constant $c$, which can be evaluated as a function of the $z$-variables. This is performed by evaluating the integral
	\begin{equation*}
		\oint_1 \prod_{i=1}^{m-1} \frac{\dd u_i}{2 \pi \ii}  f(u_1,\dots,u_m),
	\end{equation*}
	where each integral is over a contour surrounding the singularity at $u_i = 1$, on which $f$ is analytic. The integral is only in variables $u_1, \dots, u_{m-1}$ as $f$ is analytic when $u_m=1$. Evaluating the contours in the order $u_{m-1},u_{m-2}, \dots,u_1$ means only evaluating the residues of simple poles. This calculation yields $c = \prod_{j=1}^{m-1} (z_{\pi_{n-m+j}}-1)^{m-j}$, and the result follows. 
\end{proof}
Using the factorisation identity Lemma \ref{2TASEP u factorisation lemma} we wish to perform the integration in the $u$-variables. This gives the following result which shows that when considering total crossing events the type 1 and 2 particle bulks are in essence separable.

When studying total crossing events, the systems configuration is now only specified by the initial and final particle coordinates.
\begin{prop}
	\label{2TASEP crossing det formula statement}
	The probability of total crossing is given by the transition probability
	\begin{multline}
		\mathbb{P}(\mu\to\nu,\inp\to p;t) =\oint \prod_{i=1}^{m} \frac{\dd z_i }{2\pi \ii} \prod_{j=1}^{n-m} \frac{\dd w_j }{2\pi \ii} \prod_{i=1}^m\frac{e^{(z_i^{-1}-1)t}}{\left(1-z_{i}\right)^{n-m}} \prod_{i=1}^{n-m}e^{(w_i^{-1}-1)t}\prod_{i=1}^{m} \prod_{j=1}^{n-m} \left(w_{j}-z_{i}\right) \\
		\det \left(z_{i}^{\nu_{n-m+j} - \mu_i-1} \left(1-z_{i}\right)^{i-j}\right)_{1\leq i,j \leq m} \det\left(w_{i}^{\nu_{j}-\mu_{m+i}-1}\left(1-w_{i}\right)^{i-j}\right)_{1 \leq i,j \leq n-m},
	\end{multline}
	where we have relabelled the integration variables $w_i = z_{m+i}$. All contours surround the origin only.
\end{prop}
This result is also derived in an independent manner for the general $r$-species TASEP in Corollary~\ref{cor:pTASEPcross} in the next section.

\begin{proof}[Proof of Proposition \ref{2TASEP crossing det formula statement}]
	We first note that by Lemma \ref{2TASEP u factorisation lemma} when $p_j = n-m +j, \inp_j = j$ we have
	\begin{multline}
		\label{2TASEP det crossing proof with u}
		\mathbb{P}(\mu\to\nu,\inp\to p;t) =  \oint \prod_{i=1}^n \frac{\dd z_i }{2\pi \ii} \prod_{j=1}^m \frac{\dd u_j }{2\pi \ii}\ \sum_{\pi\in S_n} (-1)^{\abs{\pi}}
		\prod_{i=1}^n \left(\frac{1-z_i}{1-z_{\pi_i}}\right)^{i}
		z_{\pi_i}^{\nu_i} z_{i}^{-\mu_i-1} e^{(z^{-1}_i-1)t}
		\\
		\times
		\prod_{i=1}^m \left[ \frac{\prod_{j=1}^{n-m}(u_i-z_{\pi_j})}{\prod_{j=1}^i (u_i-z_j)} \right] \prod_{i=1}^m \frac{(z_i-1)^{m-i+1}}{z_{\pi_{n-m+i}}-1} 
		\prod_{i=1}^{n-m}\frac{1}{(1-z_{\pi_i})^m} 
		\frac{\prod_{1\leq i<j\leq m}(u_i-u_j)}{\prod_{i=1}^{m}(1-u_i)^{m-i}}.
	\end{multline}
	Then Lemma \ref{2TASEP u_i = z_i residue lemma} tells us that first integrating with respect to the $u$-variables amounts to evaluating the residues at $u_i = z_i$ for all $1 \leq i \leq m$. 
	
	By observing the $\prod_{i=1}^m \left[ \frac{\prod_{j=1}^{n-m}(u_i-z_{\pi_j})}{\prod_{j=1}^i (u_i-z_j)} \right]$ factor in \eqref{2TASEP det crossing proof with u} we see that there will only be a singularity at all $u_i = z_i$ for $1 \leq i \leq m$ if $\pi_k > m $ for all $1 \leq k \leq n-m$. This leaves us with the only non-zero terms in \eqref{2TASEP det crossing proof with u} being the ones associated with the permutation $\pi \in S_n$ such that $\pi_k > m$ for $1 \leq k \leq n-m$. This also forces $\pi_{n-m+k} \leq m$ for $1 \leq k \leq m$. The permutations $\pi \in S_n$ which satisfy this requirement are effectively total crossing permutations, and may be split into two smaller subsets of permutations. These are
	\[S_{n-m}^{(1)} = \{\tau: \{1, \dots, n-m \} \to \{m+1, \dots, n\}\} \cong S_{n-m},\]
	\[S_{m}^{(2)} = \{\sigma: \{n-m+1,\dots, n \} \to \{1, \dots, m\}\} \cong S_{m}.\]
	We also note that the signature of these permutations must satisfy $(-1)^{\abs{\pi}} = (-1)^{\abs{\tau}+\abs{\sigma}+m(n-m)}$ when constructing from $\pi \in S_n$.
	
	We re-label the integration variables of the type 2 particles by $z_{m+i} = w_i$ for $i = 1,\dots,n-m$. We can then integrate in the $u$-variables using Lemma \ref{2TASEP u_i = z_i residue lemma}, so that \eqref{2TASEP det crossing proof with u} simplifies to 
	\begin{multline}
		\mathbb{P}(\mu\to\nu,\inp\to p;t) = \oint \prod_{i=1}^{m} \frac{\dd z_i}{2\pi \ii} \prod_{j=1}^{n-m} \frac{\dd w_j }{2\pi \ii} (-1)^{m(n-m)} \prod_{i=1}^me^{(z_i^{-1}-1)t} \prod_{i=1}^{n-m}e^{(w_i^{-1}-1)t}\\	
		\times \sum_{\sigma \in S_m^{(2)}} (-1)^{\abs{\sigma}} \prod_{i=1}^m{z_{\sigma_i}^{\nu_{n-m+i}}z_i^{-\mu_i-1}} \frac{(1-z_i)^i}{(1-z_{\sigma_i})^i} \frac{(z_i-1)^{m-i+1}}{(1-z_i)^{m-i}}\frac{1}{z_{\sigma_i} -1}\\
		\times \sum_{\tau \in S_{n-m}^{(1)}} (-1)^{\abs{\tau}} \prod_{i=1}^{n-m} w_{\tau_i}^{\nu_i} w_{i}^{-\mu_{i+m}-1} \frac{(1-w_i)^{m+i}}{(1-w_{\tau_i})^{m+i}} \prod_{i=1}^{n-m} \frac{1}{(1-w_{\tau_i})^m}  \prod_{i=1}^{m} \prod_{j=1}^{n-m} (z_i - w_{\tau_j}).
	\end{multline}
	We note here that the extra sign outside the sums may be eliminated by bringing all factors into the form $(1-z_i),(1-w_i),(w_i-z_j)$.
	
	After performing the obvious cancellations in the above we then use the isomorphisms $S_{n-m}^{(1)} \cong S_{n-m}, S_m^{(2)} \cong S_m$ and simplify down using the symmetry of some products above to give
	\begin{multline}
		\mathbb{P}(\mu\to\nu,\inp\to p;t) = \oint \prod_{i=1}^{m} \frac{\dd z_i }{2\pi \ii} \prod_{j=1}^{n-m} \frac{\dd w_j }{2\pi \ii} \prod_{i=1}^me^{(z_i^{-1}-1)t} \prod_{i=1}^{n-m}e^{(w_i^{-1}-1)t} \prod_{j=1}^{m} \prod_{k=1}^{n-m} (w_k-z_j) \\	
		\times \sum_{\sigma \in S_m} (-1)^{\abs{\sigma}} \prod_{i=1}^m{z_{\sigma_i}^{\nu_{n-m+i}-\mu_{\sigma_i}-1}} (1-z_{\sigma_i})^{\sigma_i-i} \prod_{i=1}^m \frac{1}{(1-z_i)^{n-m}}\\
		\times \sum_{\tau \in S_{n-m}} (-1)^{\abs{\tau}} \prod_{i=1}^{n-m}{w_{\tau_{i}}^{\nu_{i}-\mu_{m+\tau_{i}}-1}} (1-w_{\tau_i})^{\tau_i-i}  .
	\end{multline}
	Above we have multiplied the sums by $\prod_{i=1}^m (1-z_i)^{m-n} \prod_{i=1}^{n-m} (1-w_i)^m$ to account for the shifting of the factor $\prod_{i=1}^{m}(1-z_i)^i \prod_{i=1}^{n-m}(1-w_i)^i$ under the mapping of groups $S_{n-m}^{(1)} \to S_{n-m}, S_{m}^{(2)} \to S_{m}$. We are then able to write the sums above as determinant to arrive at the result. 
\end{proof}

\section{Vertex models and partition functions}
\label{se:vertex models}

In this section we show how results of \cite{borodin_coloured_2018}, suitably specialized, lead to an alternative formulation of transition probabilities in the multi-species (T)ASEP. We begin by giving some basic definitions related to (higher-spin) coloured vertex models in Section \ref{ssec:coloured-vertex}, before recalling the definition of several families of multivariate rational functions from \cite{borodin_coloured_2018}, in Section \ref{ssec:multivariate}. 

One of these families, $G_{\mu/\nu}$, is symmetric in its alphabet $(y_1,\dots,y_\ell)$ and may be viewed as the entries of a multivariate transfer matrix. In Section \ref{ssec:transition}, we recall that a gauge-transformed version of $G_{\mu/\nu}$ provides the probability of transitioning from a state $\nu$ in the vertex model to the state $\mu$ after $\ell$ discrete time steps; we provide an explicit integral formula for this probability, using the orthogonality results of \cite{borodin_coloured_2018}.

The models introduced to this point are rather generic; aside from the variables $y_i$ assigned to horizontal lattice lines, they depend on a global parameter $s$. In Section \ref{ssec:reduction-asep} we make use of the freedom to appropriately tune these parameters; taking a particular limit of them, we descend to the $n$-species ASEP, with each of the colours $\{1,\dots,n\}$ appearing exactly once in the system. Taking the corresponding limit of the integral from Section \ref{ssec:transition}, we obtain an integral formula of the transition probabilities in the $n$-species ASEP (with $n$ distinguishable particles).

In Section \ref{ssec:symmetrize}, we show how a certain symmetrization procedure applied to the integral in Section \ref{ssec:reduction-asep} provides transition probabilities in the $r$-species ASEP in which the particles need no longer be distinguishable (each of the colours $\{1,\dots,r\}$ may appear multiple times within the system). In Section \ref{ssec:tasep}, we examine the reduction to the $r$-species TASEP by taking $q=0$ in the equations of Section \ref{ssec:symmetrize}.

\subsection{Coloured vertex models}
\label{ssec:coloured-vertex}

In the sequel, we follow closely the conventions of \cite[Chapter 2]{borodin_coloured_2018}.
The vertex models that we consider assign weights to finite collections of paths drawn on a square grid. Each vertex produces a weight that depends on the configuration of all the paths that go through it. The total weight for a collection of paths is the product of weights of the vertices that the paths traverse. We assume the normalization in which the weight of an empty vertex is always equal to $1$.

Our paths are \emph{coloured}; each path carries a colour that is a number between 1 and $n$, where $n\geq 1$ is a fixed parameter. We will assume that each horizontal edge of the underlying square grid can carry no more than one path, while vertical edges can be occupied by multiple paths. Thus, the states of the horizontal edges can be encoded by an integer between $0$ and $n$, with $0$ denoting an edge that is not occupied by a path, while the states of the vertical edges are encoded by $n$-dimensional (nonnegative-valued) vectors which specify the number of times each colour $\{1,\dots,n\}$ appears at that edge.

Paths will always travel upward in the vertical direction, and in the horizontal direction a path can travel rightward or leftward, depending on the partition function being considered.

\subsubsection{Rightward travel}

Vertex weights in the regions of rightward travel are denoted as 
\begin{align}
	\label{generic-L-intro}
	L_{z,q,s}(\I,j; \K,l)
	=
	\tikz{0.7}{
		\draw[lgray,line width=1.5pt,->] (-1,0) -- (1,0);
		\draw[lgray,line width=4pt,->] (0,-1) -- (0,1);
		\node[left] at (-1,0) {\tiny $j$};\node[right] at (1,0) {\tiny $l$};
		\node[below] at (0,-1) {\tiny $\I$};\node[above] at (0,1) {\tiny $\K$};
	}
	,
	\quad
	0\le j,l\le n,
	\quad
	\I,\K \in \{0,1,2,\dots\}^n,
\end{align}
where the vectors $\I = (I_1,\dots,I_n)$, $\K = (K_1,\dots,K_n)$ are chosen such that $I_i$ (respectively, $K_i$) gives the number of paths of colour $i$ present at the bottom (respectively, top) edge of the vertex. Demanding that any path which enters a vertex (via either of the edges marked $\I$ or $j$) must also exit it (via either of the edges marked $\K$ or $l$) leads to the constraint
\begin{align}
	\label{conserve}
	L_{z,q,s}(\I,j; \K,l)
	=
	0,
	\quad
	\text{unless}
	\quad
	\I + \bm{e}_j = \K + \bm{e}_l,
\end{align}
where $\bm{e}_j$ denotes the $j$-th $n$-dimensional Euclidean unit vector (and $\bm{e}_0 = \bm{0}$). We refer to the constraint \eqref{conserve} as {\it path conservation}. The explicit values of the weights \eqref{generic-L-intro} which respect path conservation are provided in the table below:
\begin{align}
	\label{s-weights-intro}
	\begin{tabular}{|c|c|c|}
		\hline
		\quad
		\tikz{0.7}{
			\draw[lgray,line width=1.5pt,->] (-1,0) -- (1,0);
			\draw[lgray,line width=4pt,->] (0,-1) -- (0,1);
			\node[left] at (-1,0) {\tiny $0$};\node[right] at (1,0) {\tiny $0$};
			\node[below] at (0,-1) {\tiny $\I$};\node[above] at (0,1) {\tiny $\I$};
		}
		\quad
		&
		\quad
		\tikz{0.7}{
			\draw[lgray,line width=1.5pt,->] (-1,0) -- (1,0);
			\draw[lgray,line width=4pt,->] (0,-1) -- (0,1);
			\node[left] at (-1,0) {\tiny $i$};\node[right] at (1,0) {\tiny $i$};
			\node[below] at (0,-1) {\tiny $\I$};\node[above] at (0,1) {\tiny $\I$};
		}
		\quad
		&
		\quad
		\tikz{0.7}{
			\draw[lgray,line width=1.5pt,->] (-1,0) -- (1,0);
			\draw[lgray,line width=4pt,->] (0,-1) -- (0,1);
			\node[left] at (-1,0) {\tiny $0$};\node[right] at (1,0) {\tiny $i$};
			\node[below] at (0,-1) {\tiny $\I$};\node[above] at (0,1) {\tiny $\I^{-}_i$};
		}
		\quad
		\\[1.3cm]
		\quad
		$\dfrac{1-s z q^{\Is{1}{n}}}{1-sz}$
		\quad
		& 
		\quad
		$\dfrac{(z-sq^{I_i}) q^{\Is{i+1}{n}}}{1-sz}$
		\quad
		& 
		\quad
		$\dfrac{z(1-q^{I_i}) q^{\Is{i+1}{n}}}{1-sz}$
		\quad
		\\[0.7cm]
		\hline
		\quad
		\tikz{0.7}{
			\draw[lgray,line width=1.5pt,->] (-1,0) -- (1,0);
			\draw[lgray,line width=4pt,->] (0,-1) -- (0,1);
			\node[left] at (-1,0) {\tiny $i$};\node[right] at (1,0) {\tiny $0$};
			\node[below] at (0,-1) {\tiny $\I$};\node[above] at (0,1) {\tiny $\I^{+}_i$};
		}
		\quad
		&
		\quad
		\tikz{0.7}{
			\draw[lgray,line width=1.5pt,->] (-1,0) -- (1,0);
			\draw[lgray,line width=4pt,->] (0,-1) -- (0,1);
			\node[left] at (-1,0) {\tiny $i$};\node[right] at (1,0) {\tiny $j$};
			\node[below] at (0,-1) {\tiny $\I$};\node[above] at (0,1) 
			{\tiny $\I^{+-}_{ij}$};
		}
		\quad
		&
		\quad
		\tikz{0.7}{
			\draw[lgray,line width=1.5pt,->] (-1,0) -- (1,0);
			\draw[lgray,line width=4pt,->] (0,-1) -- (0,1);
			\node[left] at (-1,0) {\tiny $j$};\node[right] at (1,0) {\tiny $i$};
			\node[below] at (0,-1) {\tiny $\I$};\node[above] at (0,1) {\tiny $\I^{+-}_{ji}$};
		}
		\quad
		\\[1.3cm] 
		\quad
		$\dfrac{1-s^2 q^{\Is{1}{n}}}{1-sz}$
		\quad
		& 
		\quad
		$\dfrac{z(1-q^{I_j}) q^{\Is{j+1}{n}}}{1-sz}$
		\quad
		&
		\quad
		$\dfrac{s(1-q^{I_i})q^{\Is{i+1}{n}}}{1-sz}$
		\quad
		\\[0.7cm]
		\hline
	\end{tabular} 
\end{align}
where we assume that $1 \leq i<j \leq n$. Here $q$ is the \emph{quantization parameter}, $s$ is the \emph{spin parameter}, and $z$ is the \emph{spectral parameter}; we have also made use of the notations 
\begin{align*}
	\I^{+}_{i}
	=
	\I + \bm{e}_i,
	\quad
	\I^{-}_{i}
	=
	\I - \bm{e}_i,
	\quad
	\I^{+-}_{ij}
	=
	\I + \bm{e}_i - \bm{e}_j,
	\quad
	\Is{i}{n}
	=
	\sum_{k=i}^{n} I_k.
\end{align*}
At $s=q^{-\frac{N}{2}}$ with $N \in \mathbb{Z}_{\geq1}$, the model forbids vertical edges occupied by more than $N$ paths. More precisely, assuming that the state $\I$ at the bottom vertical edge of the vertex contains exactly $N$ paths 
(whereby $\Is{1}{n} = N$), we see that the weight of the bottom-left vertex in \eqref{s-weights-intro} vanishes when $s=q^{-\frac{N}{2}}$; hence, if we start with $N$ or fewer paths entering the bottom vertical edge of a vertex, we may never transition to more than $N$ paths exiting via the top vertical edge. In the reduction to the multi-species ASEP, we will be interested in the case $N=1$.

\subsubsection{Leftward travel}

In the regions of leftward travel, we will use a different set of weights denoted as 
\begin{align}
	\label{generic-M-intro}
	M_{z,q,s}(\I,j;\K,l)
	=
	\tikz{0.7}{
		\draw[lgray,line width=1.5pt,<-] (-1,0) -- (1,0);
		\draw[lgray,line width=4pt,->] (0,-1) -- (0,1);
		\node[left] at (-1,0) {\tiny $l$};\node[right] at (1,0) {\tiny $j$};
		\node[below] at (0,-1) {\tiny $\I$};\node[above] at (0,1) {\tiny $\K$};
	},
	\quad
	0\le j,l \le n,
	\quad
	\I,\K \in \{0,1,2,\dots\}^n,
\end{align}
and defined by 
\begin{align}
	\label{sym1-intro}
	M_{z,q,s}(\I,j;\K,l)
	&=
	(-s)^{\bm{1}_{j \geq 1}-\bm{1}_{l \geq 1}}
	L_{z^{-1},q^{-1},s^{-1}}(\I,j;\K,l).
\end{align}

\subsubsection{Stochasticity property}

The models \eqref{generic-L-intro} and \eqref{generic-M-intro} have a sum-to-unity property that can be used to construct discrete-time Markov processes. This property is given in the following proposition:
\begin{prop}
	For any fixed $\I \in (\mathbb{Z}_{\geq 0})^n$ and $j \in \{0,1,\dots,n\}$, there holds
	\begin{align}
		\label{sum-unity1}
		\sum_{\K \in (\mathbb{Z}_{\geq 0})^n}
		\sum_{l = 0}^{n}
		L_{z,q,s}(\I,j;\K,l) (-s)^{\bm{1}_{l \geq 1}}
		&=
		1,
		\\
		\label{sum-unity2}
		\sum_{\K \in (\mathbb{Z}_{\geq 0})^n}
		\sum_{l = 0}^{n}
		M_{z,q,s}(\I,j;\K,l) (-s)^{-\bm{1}_{j \geq 1}}
		&=
		1.
	\end{align}
\end{prop}

\begin{proof}
	The first statement \eqref{sum-unity1} may be directly checked using the conservation property \eqref{conserve} and the explicit form \eqref{s-weights-intro} of the vertex weights. The second statement \eqref{sum-unity2} follows from \eqref{sum-unity1} and the definition \eqref{sym1-intro} of the weights for leftward travel.
\end{proof}

We note also the following positivity criterion for the vertex weights:
\begin{prop}
	\label{prop:pos}
	There exist choices of parameters $z,q,s \in \mathbb{R}_{> 0}$ with $z < 1$, $q > 1$, $s < 1$ such that the weights $L_{z,q,s}(\I,j;\K,l) (-s)^{\bm{1}_{l \geq 1}}$ and $M_{z,q,s}(\I,j;\K,l) (-s)^{-\bm{1}_{j \geq 1}}$ are positive for all $\I,j,\K,l$.
\end{prop}

\begin{proof}
	This boils down to analysis of the numerators and denominators that appear in the table of weights; we show how this works in the case of $L_{z,q,s}(\I,j;\K,l) (-s)^{\bm{1}_{l \geq 1}}$. The denominators all take the form $1-sz$, which is automatically positive for $0 < z < 1$, $0 < s < 1$. The positivity of the numerators that appear in $L_{z,q,s}(\I,j;\K,l) (-s)^{\bm{1}_{l \geq 1}}$ is ensured by making each of the combinations
	\begin{align*}
		1-s z q^{\I_{[1,n]}},
		\quad
		s q^{I_i} -z,
		\quad
		q^{I_i}-1,
		\quad
		1-s^2 q^{\I_{[1,n]}}
	\end{align*}
	positive. This is clearly possible as long as $z$ and $s$ are sufficiently small and $z<s$.
\end{proof}

As a consequence of these propositions, we may view the vertex weights $L_{z,q,s}(\I,j;\K,l) (-s)^{\bm{1}_{l \geq 1}}$ and $M_{z,q,s}(\I,j;\K,l) (-s)^{-\bm{1}_{j \geq 1}}$ as discrete probability distributions, describing the probability to transition to outgoing states $\K,l$ given input states $\I,j$.

\subsection{Multivariate rational functions}
\label{ssec:multivariate}

In this section we recall the definition of certain families of multivariate rational functions given in \cite[Chapters 3--4]{borodin_coloured_2018}, as well as some of their key properties.

\subsubsection{The family $f_{\mu}$}

Our first functions are indexed by compositions of nonnegative integers, and are defined via the partition functions \eqref{f-pf}; we subsequently extend their definition to compositions over all integers, via equation \eqref{shift-defn}.

\begin{defn}[The function $f_\mu$]
	Fix a composition $\mu = (\mu_1,\dots,\mu_n) \in (\mathbb{Z}_{\geq 0})^n$ and for all $k \geq 0$ introduce the vectors
	\begin{align}
		\label{AB-states1}
		\bm{A}(k) = \sum_{j=1}^{n} (\bm{1}_{\mu_j = k}) \bm{e}_j.
	\end{align}
	%
	%
	Let $(z_1,\dots,z_n)$ be a collection of complex parameters. We define $f_{\mu}(z_1,\dots,z_n) \equiv f_{\mu}(z_1,\dots,z_n;q,s)$ as the following partition function:
	\begin{align}
		\label{f-pf}
		f_{\mu}(z_1,\dots,z_n)
		&=
		\tikz{0.8}{
			\foreach\y in {1,...,5}{
				\draw[lgray,line width=1.5pt,->] (1,\y) -- (8,\y);
			}
			\foreach\x in {2,...,7}{
				\draw[lgray,line width=4pt,->] (\x,0) -- (\x,6);
			}
			\node[left] at (0.5,1) {$z_1 \rightarrow$};
			\node[left] at (0.5,2) {$z_2 \rightarrow$};
			\node[left] at (0.5,3) {$\vdots$};
			\node[left] at (0.5,4) {$\vdots$};
			\node[left] at (0.5,5) {$z_n \rightarrow$};
			\node[below] at (7,0) {$\cdots$};
			\node[below] at (6.1,0) {$\cdots$};
			\node[below] at (5.1,0) {$\cdots$};
			\node[below] at (4.1,0) {$\cdots$};
			\node[below] at (3,0) {\footnotesize$\bm{0}$};
			\node[below] at (2,0) {\footnotesize$\bm{0}$};
			\node[above] at (7,6) {$\cdots$};
			\node[above] at (6.1,6) {$\cdots$};
			\node[above] at (5.1,6) {$\cdots$};
			\node[above] at (4.1,6) {$\cdots$};
			\node[above] at (3,6) {\footnotesize$\bm{A}(1)$};
			\node[above] at (2,6) {\footnotesize$\bm{A}(0)$};
			\node[right] at (8,1) {$0$};
			\node[right] at (8,2) {$0$};
			\node[right] at (8,3) {$\vdots$};
			\node[right] at (8,4) {$\vdots$};
			\node[right] at (8,5) {$0$};
			\node[left] at (1,1) {$1$};
			\node[left] at (1,2) {$2$};
			\node[left] at (1,3) {$\vdots$};
			\node[left] at (1,4) {$\vdots$};
			\node[left] at (1,5) {$n$};
		}
	\end{align}
	In this picture, a single path of colour $i$ enters via the left edge of row $i$ of the lattice; it exits from the top of column $\mu_i$ (where numbering starts at $0$, for the leftmost column). The zeros on the bottom and right edges mean that no paths enter and exit there. The function $f_{\mu}(z_1,\dots,z_n)$ is computed by summing the weights of all configurations of paths through the lattice, where the weight of a configuration is given by the product of the weights of the individual vertices. The vertex weights are given by \eqref{s-weights-intro}, where we set $z \equiv z_i$ for a vertex in row $i$ of the lattice.
\end{defn}

\subsubsection{Factorization}

\begin{prop}
	\label{prop:factor1}
	Fix an anti-dominant composition $\delta = (\delta_1,\dots,\delta_n) \in (\mathbb{Z}_{\geq 0})^n$ such that $\delta_i \leq \delta_{i+1}$ for all $1 \leq i \leq n-1$. The function $f_{\delta}(z_1,\dots,z_n)$ factorizes as follows:
	\begin{align}
		\label{f-factor}
		f_{\delta}(z_1,\dots,z_n)
		=
		\prod_{j \geq 0} (s^2;q)_{m_j(\delta)}
		\cdot
		\prod_{i=1}^{n}
		\frac{1}{1-s z_i}
		\left( \frac{z_i-s}{1-s z_i} \right)^{\delta_i},
	\end{align}
	where $m_j(\delta) = |\{i: \delta_i=j\}|$ and $(a;q)_m = \prod_{k=1}^{m} (1-a q^{k-1})$ for all $a \in \mathbb{C}$, $m \geq 0$.
\end{prop}

\begin{proof}
	See \cite[Chapter 5, Section 1]{borodin_coloured_2018}; this result is a consequence of the fact that for anti-dominant compositions $\delta$, exactly one configuration of coloured paths is possible in the partition function $f_{\delta}(z_1,\dots,z_n)$. This configuration is the one in which (for all $1 \leq i \leq n$) the path of colour $i$ travels horizontally until column $\delta_i$, where it turns and travels vertically to its exit position at the top of the lattice:
	
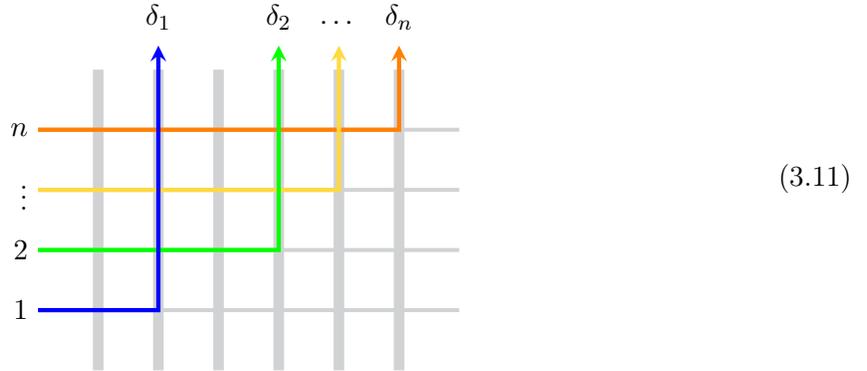
\begin{align}
		\label{factor-pic}
		\tikz{0.8}{
			\foreach\y in {2,...,5}{
				\draw[lgray,line width=1.5pt] (1,\y) -- (8,\y);
			}
			\foreach\x in {2,...,7}{
				\draw[lgray,line width=4pt] (\x,1) -- (\x,6);
			}
			\draw[ultra thick,orange,->] (1,5) -- (7,5) -- (7,6.4);
			\draw[ultra thick,yellow,->] (1,4) -- (6,4) -- (6,6.4);
			\draw[ultra thick,green,->] (1,3) -- (5,3) -- (5,6.4);
			\draw[ultra thick,blue,->] (1,2) -- (3,2) -- (3,6.4);
			\node[above] at (7,6.5) {$\delta_n$};
			\node[above] at (6,6.5) {$\cdots$};
			\node[above] at (5,6.5) {$\delta_2$};
			\node[above] at (3,6.5) {$\delta_1$};
			\node[left] at (1,2) {$1$};
			\node[left] at (1,3) {$2$};
			\node[left] at (1,4) {$\vdots$};
			\node[left] at (1,5) {$n$};
		}
	\end{align}
	Reading the weight of that configuration, one immediately recovers \eqref{f-factor}.
\end{proof}

Now we present an analogue of Proposition \ref{prop:factor1}, which describes a {\it block factorization} property of the functions under consideration.

\begin{prop}
	\label{prop:mu-block}
	Let $\mu = (\mu_1,\dots,\mu_n) \in (\mathbb{Z}_{\geq 0})^n$ be a composition that splits into $r$ blocks, with $1 \leq r \leq n$:
	\begin{align}
		\label{mu-order}
		\mu
		=
		\mu^{(1)}
		\cup
		\mu^{(2)}
		\cup
		\cdots
		\cup
		\mu^{(r)}
		=
		\left( \mu^{(1)}_1,\dots,\mu^{(1)}_{n_1} \right)
		\cup
		\left( \mu^{(2)}_1,\dots,\mu^{(2)}_{n_2} \right)
		\cup
		\cdots
		\cup
		\left( \mu^{(r)}_1,\dots,\mu^{(r)}_{n_r} \right),
	\end{align}
	where $\sum_{k=1}^{r} n_k = n$, and such that $\mu^{(j)}_a < \mu^{(k)}_b$ for all $j<k$, $a \in \{1,\dots,n_j\}$, $b \in \{1,\dots,n_k\}$. Consider the same partitioning of the alphabet $(z_1,\dots,z_n)$:
	\begin{align*}
		(z_1,\dots,z_n)
		=
		\left( z^{(1)}_1,\dots,z^{(1)}_{n_1} \right)
		\cup
		\left( z^{(2)}_1,\dots,z^{(2)}_{n_2} \right)
		\cup
		\cdots
		\cup
		\left( z^{(r)}_1,\dots,z^{(r)}_{n_r} \right).
	\end{align*}
	Under these choices for $\mu$ and $(z_1,\dots,z_n)$, we have
	\begin{align}
		\label{p-factor}
		f_{\mu}(z_1,\dots,z_n)
		=
		\prod_{k=1}^{r}
		f_{\mu^{(k)}}\left(z^{(k)}_1,\dots,z^{(k)}_{n_k}\right).
	\end{align}
\end{prop}

\begin{proof}
	The proof is very similar to that of Proposition \ref{prop:factor1}. In order to sketch it, consider once again the formula \eqref{f-factor}. Note that in the case $\delta = (\delta_1,\dots,\delta_n)$ with $\delta_1 < \cdots < \delta_n$ (strict inequalities between parts) the functions $m_j(\delta)$ are either equal to $0$ or $1$ for all $j \geq 0$. In that case all $q$-Pochhammer functions reduce to $1$ or $1-s^2$, and the formula reads
	\begin{align}
		\label{f-factor2}
		f_{\delta}(z_1,\dots,z_n)
		=
		\prod_{i=1}^{n}
		\frac{1-s^2}{1-s z_i}
		\left( \frac{z_i-s}{1-s z_i} \right)^{\delta_i}
		=
		\prod_{i=1}^{n}
		f_{\delta_i}(z_i),
	\end{align}
	where $f_{\delta_i}(z_i)$ denotes the partition function \eqref{f-pf} in the case of a single row (and accordingly, a single variable). We thus observe a factorization into $n$ single-row partition functions; what yields this result is the fact that the crossings of paths in \eqref{factor-pic} do not produce any extra multiplicative weights.
	
	This special case illustrates how the proof of \eqref{p-factor} works in general. One replaces the strict anti-dominant composition $\delta$ by the composition \eqref{mu-order}, which has strictly ordered blocks. Carrying out the computation as in \eqref{factor-pic}, replacing each path $i$ by a bundle of $n_i$ paths (all of different colours), and the variable $z_i$ by the bundle of variables 
	$\left(z^{(i)}_1,\dots,z^{(i)}_{n_i}\right)$, one easily recovers the formula \eqref{p-factor}.
\end{proof}

\subsubsection{Stability}

\begin{prop}
	For any composition $\mu = (\mu_1,\dots,\mu_n) \in (\mathbb{Z}_{\geq0})^n$ and integer $k \geq 0$, one has the stability property
	\begin{align}
		\label{stab}
		f_{\mu+k^n}(z_1,\dots,z_n)
		=
		\prod_{i=1}^{n}
		\left( \frac{z_i-s}{1-sz_i} \right)^k
		\cdot
		f_{\mu}(z_1,\dots,z_n),
	\end{align}
	where $\mu+k^n = (\mu_1+k,\dots,\mu_n+k)$.
\end{prop}

\begin{proof}
	This is another easy consequence of the partition function definition \eqref{f-pf}; we refer the reader to \cite[Proposition 9.1.1]{borodin_coloured_2018}.
\end{proof}

In view of the property \eqref{stab}, we may extend the definition of the functions $f_{\mu}$ to integer compositions (whose parts may now be negative). For any composition $\mu = (\mu_1,\dots,\mu_n) \in \mathbb{Z}^n$, we define
\begin{align}
	\label{shift-defn}
	f_{\mu}(z_1,\dots,z_n)
	=
	\left( \frac{1-sz_i}{z_i-s} \right)^k
	\cdot
	f_{\mu+k^n}(z_1,\dots,z_n),
	\qquad
	\forall\ k \geq \left|\min_{1 \leq i \leq n}(\mu_i)\right|.
\end{align}
Unless otherwise stated, we will work with generic integer compositions in what follows, for which the definition \eqref{shift-defn} applies. Note that both of the earlier results \eqref{f-factor} and \eqref{p-factor} continue to hold as stated if the participating compositions are now viewed as integer compositions.

\subsubsection{Orthogonality}

The functions $f_{\mu}$ have a remarkable self-orthogonality property with respect to certain $n$-fold contour integration. To state this result, we need to set up a certain class of {\it admissible contours}.

\begin{defn}[Admissible contours]
	\label{def:admiss}
	Let $\{C_1,\dots,C_n\}$ be a collection of contours in the complex plane, and fix two complex parameters $q,s \in \mathbb{C}$. We say that the set $\{C_1,\dots,C_n\}$ is admissible with respect to $(q,s)$ if the following conditions are met:
	\begin{itemize}
		\item The contours $\{C_1,\dots,C_n\}$ are closed, positively oriented and pairwise non-intersecting.
		\item The contours $C_i$ and $q \cdot C_i$ are both contained within contour $C_{i+1}$ for all $1 \leq i \leq n-1$, where $q \cdot C_i$ denotes the image of $C_i$ under multiplication by $q$.
		\item All contours surround the point $s$, and no contour surrounds $s^{-1}$.
	\end{itemize}
	%
\end{defn}

\begin{thm}
	Let $\{C_1,\dots,C_n\}$ be a set of admissible contours. For any two integer compositions $\mu = (\mu_1,\dots,\mu_n) \in \mathbb{Z}^n$ and $\nu = (\nu_1,\dots,\nu_n) \in \mathbb{Z}^n$, there holds
	\begin{align}
		\label{orthog}
		\frac{1}{\left( 2\pi\ii \right)^n}
		\oint_{C_1}
		\frac{\dd z_1}{z_1}
		\cdots 
		\oint_{C_n}
		\frac{\dd z_n}{z_n}
		\prod_{1 \leq i<j \leq n}
		\left(
		\frac{z_j-z_i}{z_j-q z_i}
		\right)
		f_{\nu}(z^{-1}_1,\dots,z^{-1}_n)
		g^{*}_{\mu}(z_1,\dots,z_n)
		=
		\bm{1}_{\mu=\nu},
	\end{align}
	where the function $g^{*}_{\mu}(z_1,\dots,z_n) \equiv g^{*}_{\mu}(z_1,\dots,z_n;q,s)$ is defined via the formula
	\begin{align}
		\label{g}
		g^{*}_{\mu}(z_1,\dots,z_n;q,s)
		=
		\frac{q^{{\rm inv}(\mu)}}{\prod_{j \geq 0} (s^{-2};q^{-1})_{m_j(\mu)}}
		\cdot
		\prod_{i=1}^{n} (-sz_i)^{-1}
		\cdot
		f_{\tilde\mu}(z_n^{-1},\dots,z_1^{-1};q^{-1},s^{-1}),
	\end{align}
	with $\tilde{\mu} = (\mu_n,\dots,\mu_1)$ denoting the reverse ordering of $\mu$ and 
	${\rm inv}(\mu) = |\{i < j: \mu_i < \mu_j\}|$.
\end{thm}

\begin{proof}
	The proof is long and technical; it requires exchange relations for the partition functions \eqref{f-pf} under the action of Hecke algebra generators, as well as an explicit sum formula for $f_{\mu}(z_1,\dots,z_n)$ which displays more directly the structure of its poles. We refer the reader to \cite[Chapters 5, 6, 8]{borodin_coloured_2018} for the details.
\end{proof}

\subsubsection{Symmetrization}

\begin{thm}
	For any signature $\lambda = (\lambda_1 \geq \cdots \geq \lambda_n) \in \mathbb{Z}^n$, define the following rational symmetric function:
	\begin{align}
		\label{F-sym}
		F_{\lambda}(z_1,\dots,z_n)
		&\equiv
		F_{\lambda}(z_1,\dots,z_n;q,s)
		\\
		&=
		\frac{(1-q)^n}{\prod_{i=1}^{n}(1-s z_i)}
		\cdot
		\prod_{j \in \mathbb{Z}}
		\frac{(s^2;q)_{m_j(\lambda)}}{(q;q)_{m_j(\lambda)}}
		\cdot
		\sum_{\sigma \in S_n}
		\sigma \left\{
		\prod_{1 \leq i<j \leq n}
		\left( \frac{z_i-q z_j}{z_i-z_j} \right)
		\prod_{i=1}^{n}
		\left( \frac{z_i-s}{1-sz_i} \right)^{\lambda_i}
		\right\}.
		\nonumber
	\end{align}
	One then has the symmetrization formula
	\begin{align}
		\label{f-to-F}
		\sum_{\mu:\mu^{+}=\lambda}
		f_{\mu}(z_1,\dots,z_n)
		=
		F_{\lambda}(z_1,\dots,z_n),
	\end{align}
	where the sum is taken over all compositions 
	$\mu = (\mu_1,\dots,\mu_n) \in \mathbb{Z}^n$ whose parts may be rearranged in weakly decreasing order to yield $\lambda$.
\end{thm}

\begin{proof}
	See \cite[Proposition 3.4.4]{borodin_coloured_2018}; this theorem is proved in two steps. In the first step, one notices that summing $f_{\mu}$ over compositions $\mu$ translates to summing over partition functions of the form \eqref{f-pf}, in which the boundary conditions along the top of the lattice change. More precisely, summing over $\mu$ such that $\mu^{+} = \lambda$ is equivalent to summing over all ways to permute colours at the top of the lattice \eqref{f-pf}, while holding fixed the positions where paths exit. Such types of sums are known to produce {\it colour-blind} partition functions; see \cite[Chapter 2, Section 4]{borodin_coloured_2018}. The precise statement is that
	\begin{align}
		\label{f-sum}
		\sum_{\mu:\mu^{+}=\lambda}
		f_{\mu}(z_1,\dots,z_n)
		=
		\tikz{0.8}{
			\foreach\y in {1,...,5}{
				\draw[lgray,line width=1.5pt,->] (1,\y) -- (8,\y);
			}
			\foreach\x in {2,...,7}{
				\draw[lgray,line width=4pt,->] (\x,0) -- (\x,6);
			}
			\node[left] at (0.5,1) {$z_1 \rightarrow$};
			\node[left] at (0.5,2) {$z_2 \rightarrow$};
			\node[left] at (0.5,3) {$\vdots$};
			\node[left] at (0.5,4) {$\vdots$};
			\node[left] at (0.5,5) {$z_n \rightarrow$};
			\node[below] at (7,0) {$\cdots$};
			\node[below] at (6.1,0) {$\cdots$};
			\node[below] at (5.1,0) {$\cdots$};
			\node[below] at (4.1,0) {$\cdots$};
			\node[below] at (3,0) {\footnotesize$\bm{0}$};
			\node[below] at (2,0) {\footnotesize$\bm{0}$};
			\node[above] at (7,6) {$\cdots$};
			\node[above] at (6.1,6) {$\cdots$};
			\node[above] at (5.1,6) {$\cdots$};
			\node[above] at (4.1,6) {$\cdots$};
			\node[above] at (3,6) {\footnotesize$\bm{m}_1$};
			\node[above] at (2,6) {\footnotesize$\bm{m}_0$};
			\node[right] at (8,1) {$0$};
			\node[right] at (8,2) {$0$};
			\node[right] at (8,3) {$\vdots$};
			\node[right] at (8,4) {$\vdots$};
			\node[right] at (8,5) {$0$};
			\node[left] at (1,1) {$1$};
			\node[left] at (1,2) {$1$};
			\node[left] at (1,3) {$\vdots$};
			\node[left] at (1,4) {$\vdots$};
			\node[left] at (1,5) {$1$};
		}
	\end{align}
	where the partition function on the right hand side of \eqref{f-sum} differs from that of \eqref{f-pf} in two ways: 1. The paths which enter via the left edges of the lattice all have the same colour; 2. The state at the top of column $i$ is defined to be $\bm{m}_i = m_i(\lambda) \bm{e}_1$, for all $i \geq 0$.
	
	In the second step, one explicitly evaluates the resulting partition function; see \cite{borodin2017family,BorodinPetrov} for a direct evaluation of \eqref{f-sum}, that yields the formula \eqref{F-sym}.
	
\end{proof}

\subsubsection{The family $G_{\mu/\nu}$}

The other functions that we consider are indexed by pairs of compositions, and defined in terms of the partition function \eqref{G-pf}. The compositions may be taken, directly, over all integers.

\begin{defn}[The function $G_{\mu/\nu}$]
	Fix a pair of integer compositions $\mu = (\mu_1,\dots,\mu_n) \in \mathbb{Z}^n$, $\nu = (\nu_1\dots,\nu_n) \in \mathbb{Z}^n$ such that $\mu_i \geq \nu_i$ for all $1 \leq i \leq n$. Similarly to \eqref{AB-states1}, for all $k \in \mathbb{Z}$ define vectors
	\begin{align}
		\label{AB-states2}
		\bm{A}(k) = \sum_{j=1}^{n} (\bm{1}_{\mu_j = k}) \bm{e}_j,
		\quad\quad
		\bm{B}(k) = \sum_{j=1}^{n} (\bm{1}_{\nu_j = k}) \bm{e}_j.
	\end{align}
	Let $(y_1,\dots,y_{\ell})$ be a collection of complex parameters, where $\ell \geq 1$ is arbitrary (in particular, it is independent of $n$). We define $G_{\mu/\nu}(y_1,\dots,y_\ell) \equiv G_{\mu/\nu}(y_1,\dots,y_\ell;q,s)$ as the following partition function:
	\begin{align}
		\label{G-pf}
		G_{\mu/\nu}(y_1,\dots,y_\ell)
		&=
		\tikz{0.8}{
			\foreach\y in {1,...,5}{
				\draw[lgray,line width=1.5pt,<-] (0,\y) -- (8,\y);
			}
			\foreach\x in {1,...,7}{
				\draw[lgray,line width=4pt,->] (\x,0) -- (\x,6);
			}
			\node[left] at (-0.5,1) {$y_1 \leftarrow$};
			\node[left] at (-0.5,2) {$y_2 \leftarrow$};
			\node[left] at (-0.5,3) {$\vdots$};
			\node[left] at (-0.5,4) {$\vdots$};
			\node[left] at (-0.5,5) {$y_\ell \leftarrow$};
			\node[above] at (7,6) {$\cdots$};
			\node[above] at (6,6) {$\cdots$};
			\node[above] at (5,6) {\footnotesize$\bm{B}(1)$};
			\node[above] at (4,6) {\footnotesize$\bm{B}(0)$};
			\node[above] at (2.9,6) {\footnotesize$\bm{B}(-1)$};
			\node[above] at (1.9,6) {$\cdots$};
			\node[above] at (0.9,6) {$\cdots$};
			\node[below] at (7,0) {$\cdots$};
			\node[below] at (6,0) {$\cdots$};
			\node[below] at (5,0) {\footnotesize$\bm{A}(1)$};
			\node[below] at (4,0) {\footnotesize$\bm{A}(0)$};
			\node[below] at (2.9,0) {\footnotesize$\bm{A}(-1)$};
			\node[below] at (1.9,0) {$\cdots$};
			\node[below] at (0.9,0) {$\cdots$};
			\node[right] at (8,1) {$0$};
			\node[right] at (8,2) {$0$};
			\node[right] at (8,3) {$\vdots$};
			\node[right] at (8,4) {$\vdots$};
			\node[right] at (8,5) {$0$};
			\node[left] at (0,1) {$0$};
			\node[left] at (0,2) {$0$};
			\node[left] at (0,3) {$\vdots$};
			\node[left] at (0,4) {$\vdots$};
			\node[left] at (0,5) {$0$};
		}
	\end{align}
	In this picture, columns are assigned integer labels. A path of colour $i$ enters via the bottom edge of column $\mu_i$ of the lattice; it leaves via the top edge of column $\nu_i$. Since the parts of the indexing compositions are bounded above and below, all motion of paths happens between columns $b$ and $a$, where $b = \min_{1 \leq i \leq n}(\nu_i)$ and $a = \max_{1 \leq i \leq n}(\mu_i)$; it is in this sense that the partition function \eqref{G-pf} is well-defined, despite being constructed using rows of vertices which are infinite in both directions.  The vertex weights are given by \eqref{sym1-intro}, where we set $z \equiv y_i$ for a vertex in row $i$ of the lattice.
\end{defn}

\subsubsection{Cauchy identity}

The functions $f_{\mu}$ and $G_{\mu/\nu}$ satisfy a Cauchy-type summation identity; this, together with the orthogonality result \eqref{orthog} allows one to write an integral formula for $G_{\mu/\nu}(y_1,\dots,y_{\ell})$ for arbitrary $\mu,\nu$ and $\ell$, as we see in the next section.

\begin{thm}
	Fix two collections of complex parameters $(y_1,\dots,y_\ell)$ and $(z_1,\dots,z_n)$ such that
	\begin{align*}
		\left| \frac{y_i-s}{1-s y_i} \cdot \frac{z_j-s}{1-s z_j} \right| < 1,
	\end{align*} 
	for all $1 \leq i \leq \ell$ and $1 \leq j \leq n$. Then for any composition 
	$\nu = (\nu_1,\dots,\nu_n) \in \mathbb{Z}^n$ we have the identity
	\begin{align}
		\label{cauchy}
		\sum_{\kappa \in \mathbb{Z}^n}
		f_{\kappa}(z_1,\dots,z_n)
		G_{\kappa/\nu}(y_1,\dots,y_\ell)
		=
		q^{-\ell n}
		\prod_{i=1}^{\ell}
		\prod_{j=1}^{n}
		\frac{1-q y_i z_j}{1- y_i z_j}
		\cdot
		f_{\nu}(z_1,\dots,z_n).
	\end{align}
\end{thm}

\begin{proof}
	The proof is via the Yang--Baxter equation for the vertex model \eqref{s-weights-intro}, which we do not discuss in the present work; see \cite[Proposition 4.5.1]{borodin_coloured_2018}.
\end{proof}

\subsection{A discrete-time Markov process}
\label{ssec:transition}


Let us now define a discrete-time Markov process, via the partition functions \eqref{G-pf}. The states of the Markov process will be integer compositions. For any two compositions $\mu,\nu \in \mathbb{Z}^n$, the probability 
$\mathbb{P}_y(\mu \rightarrow \nu)$ of transitioning from state $\mu$ to state $\nu$ after a single time step is defined to be
\begin{align}
	\label{prob-wt}
	\mathbb{P}_y(\mu \rightarrow \nu)
	=
	(-s)^{|\nu|-|\mu|}
	G_{\mu/\nu}(y),
\end{align}
where the weights \eqref{prob-wt} depend on the parameters $y,q,s \in \mathbb{R}_{>0}$, chosen such that $y < s < 1 < q$. In order to verify that \eqref{prob-wt} defines a valid probability measure, we need the following result:
\begin{prop}
	For any fixed composition $\mu \in \mathbb{Z}^n$, one has
	\begin{align}
		\label{G-prob}
		\sum_{\nu \in \mathbb{Z}^n}
		(-s)^{|\nu|-|\mu|}
		G_{\mu/\nu}(y)
		=
		1,
	\end{align}
	and the summand is non-negative for all $\nu$. Here we have defined the weight of a composition by $|\mu| = \sum_{i=1}^{n} \mu_i$.
\end{prop} 

\begin{proof}
	The quantity $(-s)^{|\nu|-|\mu|} G_{\mu/\nu}(y)$ may be computed as a partition function of the form \eqref{G-pf}, with $\ell=1$, $y_1 = y$, where the vertex weights are of modified form $M_{y,q,s}(\I,j;\K,l) (-s)^{-\bm{1}_{j \geq 1}}$. Indeed, the factor $(-s)^{-\bm{1}_{j \geq 1}}$ assigns an extra multiplicative weight of $(-s)^{-1}$ whenever the incoming horizontal edge of a vertex \eqref{generic-M-intro} is occupied by a path; the total number of such edges in the partition function $G_{\mu/\nu}(y)$ is $|\mu| - |\nu|$.
	
	The sum-to-unity property \eqref{G-prob} then follows by applying the identity \eqref{sum-unity2} iteratively to each vertex in the one-row partition function, working from right to left; the non-negativity of the summand follows immediately from Proposition \ref{prop:pos}.
\end{proof}

In a similar vein, let $\mathbb{P}_{\{y_1,\dots,y_{\ell}\}}(\mu \rightarrow \nu)$ denote the probability of transitioning from state $\mu$ to state $\nu$ after $\ell$ time steps, where the transition probabilities at the $i$-th time step are chosen to be \eqref{prob-wt} with $y=y_i$. By concatenation of one-row partition functions, we conclude that
\begin{align}
	\label{G-prob-ell}
	\mathbb{P}_{\{y_1,\dots,y_{\ell}\}}(\mu \rightarrow \nu)
	=
	(-s)^{|\nu|-|\mu|}
	G_{\mu/\nu}(y_1,\dots,y_{\ell}).
\end{align}
Equation \eqref{G-prob-ell} is the cornerstone result in this part of the text. As we show in the following section, it has a limit to the (continuous-time) multi-species ASEP, allowing us to compute transition probabilities within the latter. Before moving on to this, we provide an explicit integral formula for the probabilities \eqref{G-prob-ell} in the case $\mu_1 \geq \cdots \geq \mu_n$:

\begin{thm}
	\label{thm}
	Let $\mu \in \mathbb{Z}^n$ be chosen such that $\mu_1 \geq \cdots \geq \mu_n$, and fix a further composition $\nu \in \mathbb{Z}^n$ (with no restrictions imposed upon the ordering of its parts). We then have the integral formula
	\begin{multline}
		\label{Gmunu-int}
		\mathbb{P}_{\{y_1,\dots,y_{\ell}\}}(\mu \rightarrow \nu)
		=
		\frac{(-s)^{|\nu|-|\mu|} q^{-\ell n}}{(2\pi\ii)^n}
		\oint_{C_1}
		\frac{\dd z_1}{z_1}
		\cdots 
		\oint_{C_n}
		\frac{\dd z_n}{z_n}
		\\
		\prod_{1 \leq i<j \leq n}
		\left(
		\frac{z_j-z_i}{z_j-q z_i}
		\right)
		\prod_{i=1}^{\ell}
		\prod_{j=1}^{n}
		\frac{z_j - q y_i}{z_j - y_i}
		\prod_{i=1}^{n}
		\frac{1}{1-s z_i}
		\left(
		\frac{z_i-s}{1-sz_i}
		\right)^{\mu_i}
		f_{\nu}(z^{-1}_1,\dots,z^{-1}_n),
	\end{multline}
	provided that all points in the set $\{y_1,\dots,y_{\ell}\}$ are sufficiently close to $s$, whereby they are enclosed by each of the contours $\{C_1,\dots,C_n\}$, which are admissible with respect to $(q,s)$.
	
\end{thm}

\begin{proof}
	Start from the the Cauchy identity \eqref{cauchy} with $z_i \mapsto z_i^{-1}$ for all $1 \leq i \leq n$ and multiply it by $g^{*}_{\mu}(z_1,\dots,z_n)$, as given by \eqref{g}, where $\mu \in \mathbb{Z}^n$ is an arbitrary composition: 
	\begin{multline*}
		\sum_{\kappa \in \mathbb{Z}^n}
		f_{\kappa}(z_1^{-1},\dots,z_n^{-1})
		G_{\kappa/\nu}(y_1,\dots,y_\ell)
		g^{*}_{\mu}(z_1,\dots,z_n)
		=
		\\
		q^{-\ell n}
		\prod_{i=1}^{\ell}
		\prod_{j=1}^{n}
		\frac{z_j-q y_i}{z_j- y_i}
		\cdot
		f_{\nu}(z_1^{-1},\dots,z_n^{-1})
		g^{*}_{\mu}(z_1,\dots,z_n).
	\end{multline*}
	Integrating as in equation \eqref{orthog}, the sum over $\kappa$ collapses to a single term corresponding to $\kappa = \mu$, in view of the indicator function on the right hand side of \eqref{orthog}. We then have
	\begin{multline}
		\label{G-intermediate}
		G_{\mu/\nu}(y_1,\dots,y_{\ell})
		=
		\frac{q^{-\ell n}}{(2\pi\ii)^n}
		\oint_{C_1}
		\frac{\dd z_1}{z_1}
		\cdots 
		\oint_{C_n}
		\frac{\dd z_n}{z_n}
		\\
		\prod_{1 \leq i<j \leq n}
		\left(
		\frac{z_j-z_i}{z_j-q z_i}
		\right)
		\prod_{i=1}^{\ell}
		\prod_{j=1}^{n}
		\frac{z_j - q y_i}{z_j - y_i}
		f_{\nu}(z^{-1}_1,\dots,z^{-1}_n)
		g^{*}_{\mu}(z_1,\dots,z_n).
	\end{multline}
	Finally, choose $\mu$ such that $\mu_1 \geq \cdots \geq \mu_n$; by virtue of \eqref{g} and \eqref{f-factor} we find that
	\begin{multline*}
		g^{*}_{\mu}(z_1,\dots,z_n)
		=
		\frac{q^{{\rm inv}(\mu)}}{\prod_{j \geq 0} (s^{-2};q^{-1})_{m_j(\mu)}}
		\cdot
		\prod_{i=1}^{n} (-sz_i)^{-1}
		\\
		\times
		\prod_{j \geq 0} (s^{-2};q^{-1})_{m_j(\mu)}
		\cdot
		\prod_{i=1}^{n}
		\frac{1}{1-s^{-1} z_i^{-1}}
		\left( \frac{z_i^{-1}-s^{-1}}{1-s^{-1} z_i^{-1}} \right)^{\mu_i}.
	\end{multline*}
	Since ${\rm inv}(\mu) = 0$, this simplifies to 
	\begin{align*}
		g^{*}_{\mu}(z_1,\dots,z_n)
		=
		\prod_{i=1}^{n}
		\frac{1}{1-s z_i}
		\left( \frac{z_i-s}{1-s z_i} \right)^{\mu_i}.
	\end{align*}
	Substituting this into \eqref{G-intermediate} and using \eqref{G-prob-ell}, we deduce the result \eqref{Gmunu-int}.
	
\end{proof}

\subsection{Reduction to the multi-species asymmetric simple exclusion process}
\label{ssec:reduction-asep}

In this section we degenerate the results \eqref{G-prob-ell} and \eqref{Gmunu-int} to the multi-species ASEP on the integer lattice. This is achieved via a certain limiting procedure of the parameters $\{y_1,\dots,y_{\ell}\}$ and $s$. In Section \ref{sssec:reduction} we examine what happens to the stochastic vertex weights $M_{y,q,s}(\I,j;\K,l) (-s)^{-\bm{1}_{j \geq 1}}$ under these specific choices of $y$ and $s$, and from this we infer Proposition \ref{prop:asep-prob}, which yields the reduction of \eqref{G-prob-ell} to the multi-species ASEP. In Section \ref{sssec:int} we carry out this limiting procedure directly on the integral \eqref{Gmunu-int}, thus obtaining an integral formula for the transition probabilities in the multi-species ASEP. 

\subsubsection{Reduction of the weights}
\label{sssec:reduction}

Let us begin by fixing some new notation for our vertex weights. We write
\begin{align}
	{\sf M}_y(\I,j;\K,l)
	=
	M_{q^{-\frac{1}{2}}y,q,q^{-\frac{1}{2}}}(\I,j;\K,l) (-q^{\frac{1}{2}})^{\mathbf{1}_{j\ge 1}}.
\end{align}
These are simply the stochastic weights $M_y(\I,j;\K,l)(-s)^{-\mathbf{1}_{j\ge 1}} $ with the replacements $y \mapsto q^{-\frac{1}{2}} y$ and $s \mapsto q^{-\frac{1}{2}}$. Using \eqref{s-weights-intro} and the definition \eqref{sym1-intro}, we may explicitly tabulate ${\sf M}_y(\I,j;\K,l)$ for all choices of $\I,j,\K,l$ which respect the conservation requirement \eqref{conserve}. We obtain the table below:
\begin{align}
	\label{dual-s-weights}
	\begin{tabular}{|c|c|c|}
		\hline
		\quad
		\tikz{0.7}{
			\draw[lgray,line width=1.5pt,<-] (-1,0) -- (1,0);
			\draw[lgray,line width=4pt,->] (0,-1) -- (0,1);
			\node[left] at (-1,0) {\tiny $0$};\node[right] at (1,0) {\tiny $0$};
			\node[below] at (0,-1) {\tiny $\I$};\node[above] at (0,1) {\tiny $\I$};
		}
		\quad
		&
		\quad
		\tikz{0.7}{
			\draw[lgray,line width=1.5pt,<-] (-1,0) -- (1,0);
			\draw[lgray,line width=4pt,->] (0,-1) -- (0,1);
			\node[left] at (-1,0) {\tiny $i$};\node[right] at (1,0) {\tiny $i$};
			\node[below] at (0,-1) {\tiny $\I$};\node[above] at (0,1) {\tiny $\I$};
		}
		\quad
		&
		\quad
		\tikz{0.7}{
			\draw[lgray,line width=1.5pt,<-] (-1,0) -- (1,0);
			\draw[lgray,line width=4pt,->] (0,-1) -- (0,1);
			\node[left] at (-1,0) {\tiny $i$};\node[right] at (1,0) {\tiny $0$};
			\node[below] at (0,-1) {\tiny $\I$};\node[above] at (0,1) {\tiny $\I^{-}_i$};
		}
		\quad
		\\[1.3cm]
		\quad
		$\dfrac{q^{-\Is{1}{n}}-q^{-1}y}{1-q^{-1}y}$
		\quad
		& 
		\quad
		$\dfrac{(1-yq^{-I_i}) q^{-\Is{i+1}{n}}}{1-q^{-1}y}$
		\quad
		& 
		\quad
		$\dfrac{(1-q^{-I_i}) q^{-\Is{i+1}{n}}}{1-q^{-1}y}$
		\quad
		\\[0.7cm]
		\hline
		\quad
		\tikz{0.7}{
			\draw[lgray,line width=1.5pt,<-] (-1,0) -- (1,0);
			\draw[lgray,line width=4pt,->] (0,-1) -- (0,1);
			\node[left] at (-1,0) {\tiny $0$};\node[right] at (1,0) {\tiny $i$};
			\node[below] at (0,-1) {\tiny $\I$};\node[above] at (0,1) {\tiny $\I^{+}_i$};
		}
		\quad
		&
		\quad
		\tikz{0.7}{
			\draw[lgray,line width=1.5pt,<-] (-1,0) -- (1,0);
			\draw[lgray,line width=4pt,->] (0,-1) -- (0,1);
			\node[left] at (-1,0) {\tiny $j$};\node[right] at (1,0) {\tiny $i$};
			\node[below] at (0,-1) {\tiny $\I$};\node[above] at (0,1) {\tiny $\I^{+-}_{ij}$};
		}
		\quad
		&
		\quad
		\tikz{0.7}{
			\draw[lgray,line width=1.5pt,<-] (-1,0) -- (1,0);
			\draw[lgray,line width=4pt,->] (0,-1) -- (0,1);
			\node[left] at (-1,0) {\tiny $i$};\node[right] at (1,0) {\tiny $j$};
			\node[below] at (0,-1) {\tiny $\I$};\node[above] at (0,1) 
			{\tiny $\I^{+-}_{ji}$};
		}
		\quad
		\\[1.3cm]
		\quad
		$\dfrac{y(q^{-\Is{1}{n}}-q^{-1})}{1-q^{-1}y}$
		\quad
		& 
		\quad
		$\dfrac{(1-q^{-I_j}) q^{-\Is{j+1}{n}}}{1-q^{-1}y}$
		\quad
		&
		\quad
		$\dfrac{y(1-q^{-I_i})q^{-\Is{i+1}{n}}}{1-q^{-1}y}$
		\quad
		\\[0.7cm]
		\hline
	\end{tabular} 
\end{align}
where it is assumed that $1 \leq i < j \leq n$. As we remarked earlier, the bottom leftmost of these weights vanishes when $\Is{1}{n} = 1$, thereby eliminating the possibility to inject more than one path into the vertical edge of a vertex (this is the only vertex that allows such injections). Accordingly, we may restrict our attention to states on vertical edges that contain at most one path; computing ${\sf M}_y(\bm{e}_i,j;\bm{e}_k,l)$ for $i,j,k,l$ such that $\bm{e}_i+\bm{e}_j = \bm{e}_k+\bm{e}_l$, we obtain the following table:
\begin{align}
	\label{fundamental-wt}
	\begin{tabular}{|c|c|c|c|c|}
		\hline
		\quad
		\tikz{0.7}{
			\draw[lgray,line width=1.5pt,<-] (-1,0) -- (1,0);
			\draw[lgray,line width=4pt,->] (0,-1) -- (0,1);
			\node[left] at (-1,0) {\tiny $i$};\node[right] at (1,0) {\tiny $i$};
			\node[below] at (0,-1) {\tiny $\bm{e}_i$};
			\node[above] at (0,1) {\tiny $\bm{e}_i$};
		}
		\quad
		&
		\quad
		\tikz{0.7}{
			\draw[lgray,line width=1.5pt,<-] (-1,0) -- (1,0);
			\draw[lgray,line width=4pt,->] (0,-1) -- (0,1);
			\node[left] at (-1,0) {\tiny $i$};\node[right] at (1,0) {\tiny $i$};
			\node[below] at (0,-1) {\tiny $\bm{e}_j$};
			\node[above] at (0,1) {\tiny $\bm{e}_j$};
		}
		\quad
		&
		\quad
		\tikz{0.7}{
			\draw[lgray,line width=1.5pt,<-] (-1,0) -- (1,0);
			\draw[lgray,line width=4pt,->] (0,-1) -- (0,1);
			\node[left] at (-1,0) {\tiny $j$};\node[right] at (1,0) {\tiny $j$};
			\node[below] at (0,-1) {\tiny $\bm{e}_i$};
			\node[above] at (0,1) {\tiny $\bm{e}_i$};
		}
		\quad
		&
		\quad
		\tikz{0.7}{
			\draw[lgray,line width=1.5pt,<-] (-1,0) -- (1,0);
			\draw[lgray,line width=4pt,->] (0,-1) -- (0,1);
			\node[left] at (-1,0) {\tiny $j$};\node[right] at (1,0) {\tiny $i$};
			\node[below] at (0,-1) {\tiny $\bm{e}_j$};
			\node[above] at (0,1) {\tiny $\bm{e}_i$};
		}
		\quad
		&
		\quad
		\tikz{0.7}{
			\draw[lgray,line width=1.5pt,<-] (-1,0) -- (1,0);
			\draw[lgray,line width=4pt,->] (0,-1) -- (0,1);
			\node[left] at (-1,0) {\tiny $i$};\node[right] at (1,0) {\tiny $j$};
			\node[below] at (0,-1) {\tiny $\bm{e}_i$};
			\node[above] at (0,1) {\tiny $\bm{e}_j$};
		}
		\quad
		\\[1.3cm]
		\quad
		$1$
		\quad
		& 
		\quad
		$\dfrac{q^{-1}(1-y)}{1-q^{-1}y}$
		\quad
		&
		\quad
		$\dfrac{1-y}{1-q^{-1}y}$
		\quad
		&
		\quad
		$\dfrac{1-q^{-1}}{1-q^{-1}y}$
		\quad
		&
		\quad
		$\dfrac{y(1-q^{-1})}{1-q^{-1}y}$
		\quad
		\\[0.7cm]
		\hline
	\end{tabular}
\end{align}
with $0 \leq i < j \leq n$. The tabulated weights are those of the $U_{q^{-1}}(\widehat{\mathfrak{sl}}_{n+1})$ stochastic vertex model with fundamental representations of both horizontal and vertical lattice lines; see \cite[Chapter 2]{borodin_coloured_2018}.

Now let us take the parameter $y$ to be given by
\begin{align}
	y = 1 + (1-q) \epsilon,
\end{align} 
where $\epsilon \in \mathbb{R}_{>0}$ is a small positive real number. Making this replacement in \eqref{fundamental-wt} and truncating to first order in $\epsilon$, we recover the following:
\begin{align*}
	\begin{tabular}{|c|c|c|c|c|}
		\hline
		\quad
		\tikz{0.7}{
			\draw[lgray,line width=1.5pt,<-] (-1,0) -- (1,0);
			\draw[lgray,line width=4pt,->] (0,-1) -- (0,1);
			\node[left] at (-1,0) {\tiny $i$};\node[right] at (1,0) {\tiny $i$};
			\node[below] at (0,-1) {\tiny $\bm{e}_i$};
			\node[above] at (0,1) {\tiny $\bm{e}_i$};
		}
		\quad
		&
		\quad
		\tikz{0.7}{
			\draw[lgray,line width=1.5pt,<-] (-1,0) -- (1,0);
			\draw[lgray,line width=4pt,->] (0,-1) -- (0,1);
			\node[left] at (-1,0) {\tiny $i$};\node[right] at (1,0) {\tiny $i$};
			\node[below] at (0,-1) {\tiny $\bm{e}_j$};
			\node[above] at (0,1) {\tiny $\bm{e}_j$};
		}
		\quad
		&
		\quad
		\tikz{0.7}{
			\draw[lgray,line width=1.5pt,<-] (-1,0) -- (1,0);
			\draw[lgray,line width=4pt,->] (0,-1) -- (0,1);
			\node[left] at (-1,0) {\tiny $j$};\node[right] at (1,0) {\tiny $j$};
			\node[below] at (0,-1) {\tiny $\bm{e}_i$};
			\node[above] at (0,1) {\tiny $\bm{e}_i$};
		}
		\quad
		&
		\quad
		\tikz{0.7}{
			\draw[lgray,line width=1.5pt,<-] (-1,0) -- (1,0);
			\draw[lgray,line width=4pt,->] (0,-1) -- (0,1);
			\node[left] at (-1,0) {\tiny $j$};\node[right] at (1,0) {\tiny $i$};
			\node[below] at (0,-1) {\tiny $\bm{e}_j$};
			\node[above] at (0,1) {\tiny $\bm{e}_i$};
		}
		\quad
		&
		\quad
		\tikz{0.7}{
			\draw[lgray,line width=1.5pt,<-] (-1,0) -- (1,0);
			\draw[lgray,line width=4pt,->] (0,-1) -- (0,1);
			\node[left] at (-1,0) {\tiny $i$};\node[right] at (1,0) {\tiny $j$};
			\node[below] at (0,-1) {\tiny $\bm{e}_i$};
			\node[above] at (0,1) {\tiny $\bm{e}_j$};
		}
		\quad
		\\[1.3cm]
		\quad
		$1$
		\quad
		& 
		\quad
		$\epsilon + O(\epsilon^2)$
		\quad
		&
		\quad
		$q \epsilon + O(\epsilon^2)$
		\quad
		&
		\quad
		$1-\epsilon+O(\epsilon^2)$
		\quad
		&
		\quad
		$1-q\epsilon+O(\epsilon^2)$
		\quad
		\\[0.7cm]
		\hline
	\end{tabular}
\end{align*}
with $0 \leq i < j \leq n$. As $\epsilon \rightarrow 0$, the second and third of these vertices have weights which tend to zero. In this limit paths will therefore tend to travel in zig-zag fashion through the lattice \eqref{G-pf}; either entering a vertex from the bottom edge and leaving via the left one, or entering a vertex from the right edge and leaving via the top one. To make the limit non-trivial, one needs to combine the choice $y_i = q^{-\frac{1}{2}}[1+(1-q)\epsilon]$, $1 \leq i \leq \ell$ of parameters in \eqref{G-pf} with a rescaling by $\epsilon^{-1}$ of the number of rows. In this limit, the discrete-time Markov process introduced in Section \ref{ssec:transition} reduces to the (continuous-time) multi-species ASEP, as we formalize below.  

\begin{defn}[Transition probability $\mathbb{P}(\mu \rightarrow \nu;t)$]
	Consider the multi-species ASEP on the integer lattice, with $n$ particles, all of different types. Let $(i,j)$ denote the occupation data of two neighbouring sites in the lattice, where $0 \leq i \not= j \leq n$. Then $(i,j) \mapsto (j,i)$ in time $t$, where $t$ is a random variable chosen from the exponential distribution with parameter $q$ if $0 \leq i < j \leq n$ and with parameter $1$ if $n \geq i > j \geq 0$.
	
	A state of the system is described by a strict\footnote{Having pairwise distinct parts.} integer composition $\mu = (\mu_1,\dots,\mu_n)$, where $\mu_i \in \mathbb{Z}$ is the coordinate of particle $i$. We denote the probability that the system is in state $\nu$ at time $t$, having started in state $\mu$ initially, by $\mathbb{P}(\mu \rightarrow \nu;t)$.
\end{defn}

\begin{prop}
	\label{prop:asep-prob}
	Fix two integer compositions $\mu,\nu \in \mathbb{Z}^n$ and a small parameter $\epsilon \in \mathbb{R}_{>0}$. Define
	\begin{align}
		\label{sfG}
		{\sf P}_{\ell;\epsilon}(\mu \rightarrow \nu)
		=
		\mathbb{P}_{\{y_1,\dots,y_{\ell}\}}(\mu \rightarrow \nu-\ell^n)
		\Big|_{s \mapsto q^{-\frac{1}{2}}}
		\Big|_{y_1,\dots,y_{\ell} \mapsto q^{-\frac{1}{2}} [1+(1-q)\epsilon]}
	\end{align}
	where every variable $y_1,\dots,y_{\ell}$ is substituted to the value shown.\footnote{For $q \in \mathbb{R}_{>1}$ it is clear that these substitutions yield $0 < y_i < s < 1$, which allows us to claim the positivity of weights as in Proposition \ref{prop:pos}.} The function ${\sf P}_{\ell;\epsilon}(\mu \rightarrow \nu)$ has analytic dependence on $\ell$; we analytically continue $\ell$ to generic complex values. For fixed $t \in \mathbb{R}_{>0}$ we then have
	\begin{align}
		\label{eps-limit}
		\mathbb{P}(\mu \rightarrow \nu;t)
		=
		\lim_{\epsilon \rightarrow 0}
		{\sf P}_{t/\epsilon;\epsilon}(\mu \rightarrow \nu).
	\end{align}
\end{prop}

\begin{proof}
	The analyticity of the function \eqref{sfG} in $\ell$ is proved in the following subsection; see equation \eqref{analytic-ell}, which gives an explicit formula for ${\sf P}_{\ell;\epsilon}(\mu \rightarrow \nu)$, from which the claim is immediate. The final statement \eqref{eps-limit} is well known in the literature; see, for example, \cite{aggarwal2017convergence}. 
\end{proof}

\subsubsection{Integral formula for ASEP transition probabilities}
\label{sssec:int}

\begin{thm}
	\label{th:ASEPtransition}
	Let $\mu \in \mathbb{Z}^n$ be chosen such that $\mu_1 > \cdots > \mu_n$, and fix a further strict composition $\nu \in \mathbb{Z}^n$ (with no restrictions imposed upon the ordering of its parts). We then have the integral formula
	\begin{multline}
		\label{asep-int}
		\mathbb{P}(\mu \rightarrow \nu;t)
		=
		\frac{(-q^{-\frac{1}{2}})^{|\nu|}}{(2\pi\ii)^n}
		\oint_{C}
		\frac{\dd z_1}{z_1}
		\cdots 
		\oint_{C}
		\frac{\dd z_n}{z_n}
		\prod_{1 \leq i<j \leq n}
		\left(
		\frac{z_j-z_i}{z_j-q z_i}
		\right)
		\\
		\times
		\prod_{j=1}^{n}
		\exp\left[\frac{(1-q)^2 z_j t}{(1-z_j)(1-q z_j)}\right]
		\frac{1}{1-z_j}
		\left(
		\frac{1-qz_j}{1-z_j}
		\right)^{\mu_j}
		f_{\nu}\left(q^{-\frac{1}{2}}z^{-1}_1,\dots,q^{-\frac{1}{2}}z^{-1}_n;q,q^{-\frac{1}{2}}\right),
	\end{multline}
	where all integration contours are the same; $C$ is a small negatively oriented circle centred on $1$.
\end{thm}

\begin{proof}
	We begin by explicitly writing down ${\sf P}_{\ell;\epsilon}(\mu \rightarrow \nu)$. By direct computation using \eqref{Gmunu-int} and \eqref{sfG}, we have
	\begin{multline}
		\label{P-intermediate}
		{\sf P}_{\ell;\epsilon}(\mu \rightarrow \nu)
		=
		\frac{(-q^{-\frac{1}{2}})^{|\nu|-\ell n-|\mu|} q^{-\ell n}}{(2\pi\ii)^n}
		\oint_{C_1}
		\frac{\dd z_1}{z_1}
		\cdots 
		\oint_{C_n}
		\frac{\dd z_n}{z_n}
		\prod_{1 \leq i<j \leq n}
		\left(
		\frac{z_j-z_i}{z_j-q z_i}
		\right)
		\\
		\times
		\prod_{j=1}^{n}
		\left( 
		\frac{z_j-q^{\frac{1}{2}}[1+(1-q)\epsilon]}
		{z_j-q^{-\frac{1}{2}}[1+(1-q)\epsilon]} 
		\right)^\ell
		\frac{1}{1-q^{-\frac{1}{2}} z_j}
		\left(
		\frac{z_j-q^{-\frac{1}{2}}}{1-q^{-\frac{1}{2}}z_j}
		\right)^{\mu_j}
		f_{\nu-\ell^n}\left(z^{-1}_1,\dots,z^{-1}_n;q,q^{-\frac{1}{2}}\right).
	\end{multline}
	In order to respect the admissibility requirements of the integration contours (Definition \ref{def:admiss}), each contour $C_i$ is the union of a positively oriented outer contour $O_i$ that surrounds both $s=q^{-\frac{1}{2}}$ and $s^{-1} = q^{\frac{1}{2}}$, and a negatively oriented inner contour $I$ that surrounds only $q^{\frac{1}{2}}$. Next, we use the shift formula \eqref{stab} with $k \mapsto \ell$, $\mu \mapsto \nu-\ell^n$, $s \mapsto q^{-\frac{1}{2}}$, $z_i \mapsto z_i^{-1}$: 
	\begin{align*}
		f_{\nu-\ell^n}\left(z^{-1}_1,\dots,z^{-1}_n;q,q^{-\frac{1}{2}}\right)
		=
		\prod_{j=1}^{n}
		\left( \frac{1-q^{-\frac{1}{2}}z_j^{-1}}{z_j^{-1}-q^{-\frac{1}{2}}} \right)^\ell
		f_{\nu}\left(z^{-1}_1,\dots,z^{-1}_n;q,q^{-\frac{1}{2}}\right).
	\end{align*}
	Applying this formula to the integrand of \eqref{P-intermediate} and rearranging factors, we obtain
	\begin{multline}
		\label{analytic-ell}
		{\sf P}_{\ell;\epsilon}(\mu \rightarrow \nu)
		=
		\frac{(-q^{-\frac{1}{2}})^{|\nu|}}{(2\pi\ii)^n}
		\oint_{C_1}
		\frac{\dd z_1}{z_1}
		\cdots 
		\oint_{C_n}
		\frac{\dd z_n}{z_n}
		\prod_{1 \leq i<j \leq n}
		\left(
		\frac{z_j-z_i}{z_j-q z_i}
		\right)
		\\
		\times
		\prod_{j=1}^{n}
		\left[ 1 + \epsilon \frac{(1-q)^2 z_j}{(q^{\frac{1}{2}}-z_j)(1-q^{\frac{1}{2}} z_j)} \right]^\ell
		\frac{1}{1-q^{-\frac{1}{2}} z_j}
		\left(
		\frac{1-q^{\frac{1}{2}} z_j}{1-q^{-\frac{1}{2}}z_j}
		\right)^{\mu_j}
		f_{\nu}\left(z^{-1}_1,\dots,z^{-1}_n;q,q^{-\frac{1}{2}}\right).
	\end{multline}
	The result \eqref{asep-int} follows by computing the limit \eqref{eps-limit} and sending each outer contour $O_i$ to infinity (the residue at infinity is zero); we also rescale our integration variables, by making the replacements $z_i \mapsto q^{\frac{1}{2}}z_i$ for all $1 \leq i \leq n$ in the preceding integral (this maps $I$ to $C$).
\end{proof}

\begin{rmk}
	An important special case of \eqref{asep-int} is recovered by imposing the ordering constraint $\nu_1 < \cdots < \nu_n$ on the parts of the outgoing composition $\nu$. In that case we may invoke the factorization formulae \eqref{f-factor} for the function $f_{\nu}$; applying this result within the integrand of \eqref{asep-int}, after some simplifications we obtain
	\begin{multline}
		\label{n-colours-total}
		\mathbb{P}(\mu \rightarrow \nu;t)
		=
		\left(\frac{1-q}{2\pi\ii}\right)^n
		\oint_{C}
		\dd z_1
		\cdots 
		\oint_{C}
		\dd z_n
		\prod_{1 \leq i<j \leq n}
		\left(
		\frac{z_j-z_i}{z_j-q z_i}
		\right)
		\\
		\times
		\prod_{j=1}^{n}
		\exp\left[ \frac{(1-q)^2 z_j t}{(1-z_j)(1-qz_j)}\right]
		\frac{1}{(1-z_j)(1-q z_j)}
		\left(
		\frac{1-qz_j}{1-z_j}
		\right)^{\mu_j-\nu_j}.
	\end{multline}
	This formula gives the probability that a system of $n$ distinguishable particles with initial positions $\mu_1 > \cdots > \mu_n$ transitions, after time $t$, to the state with particles at positions $\nu_1 < \cdots < \nu_n$; it is our first example of a {\it total crossing probability}, in the sense the particles totally reverse their relative ordering after time $t$.
	
	One striking feature of the formula \eqref{n-colours-total} is that it is manifestly invariant under the simultaneous shifting
	\begin{align*}
		\mu_i \mapsto \mu_i + \Delta,
		\qquad
		\nu_i \mapsto \nu_i + \Delta,
	\end{align*}
	for any $i \in [1,n]$ and $\Delta \in \mathbb{Z}$, provided that this shifting does not alter the total ordering of particles; in other words, provided that $\mu_{i-1} > \mu_i + \Delta > \mu_{i+1}$ and $\nu_{i-1} < \nu_i + \Delta < \nu_{i+1}$. This is an explicit example of the {\it shift invariance} phenomenon which has been topical in recent literature \cite{borodin2019colorposition,borodin2020shiftinvariance,galashin2020symmetries}.
\end{rmk}

\subsection{Symmetrization}
\label{ssec:symmetrize}

\begin{defn}[Admissibility and the transition probability $\mathbb{P}_t(\vec{\mu} \rightarrow \vec{\nu})$]
	Fix an integer $r \geq 1$ and let $(n_1,\dots,n_r)$ be a vector of positive integers. Let $\vec{\mu}$ be a $r$-dimensional vector, given by
	\begin{align*}
		\vec{\mu} = \left( \mu^{(1)}, \dots, \mu^{(r)} \right),
	\end{align*}
	where for all $1 \leq i \leq r$ one has
	\begin{align}
		\label{strict}
		\mu^{(i)} = \left( \mu^{(i)}_1 > \cdots > \mu^{(i)}_{n_i} \right);
	\end{align}
	that is, the elements of $\vec{\mu}$ are strict signatures. We say that $\vec{\mu}$ is {\it admissible} if it satisfies, in addition to \eqref{strict}, the pairwise distinctness property $\mu^{(i)}_a \not= \mu^{(j)}_b$, for all $i \not= j$ and all integers 
	$a \in \{1,\dots,n_i\}$ and $b \in \{1,\dots,n_j\}$.
	
	Consider the multi-species ASEP on the integer lattice with $n = n_1 + \cdots + n_r$ particles, where $n_i$ particles are of type $i$ (that is, they have the colour $i$), for all $1 \leq i \leq r$. States in this system are naturally indexed by admissible vectors $\vec{\mu}$. We let $\mathbb{P}_t(\vec{\mu} \rightarrow \vec{\nu})$ denote the probability that the system is in state $\vec{\nu}$ at time $t$, given that it started in state $\vec{\mu}$, where both $\vec{\mu}$ and $\vec{\nu}$ are admissible.
\end{defn}

\begin{prop}
	Fix an integer $r \geq 1$ and let $(n_1,\dots,n_r)$ be a vector of positive integers such that $\sum_{k=1}^{r} n_k = n$. Let $\mu$ be a strict integer composition of length $n$ that we partition into $r$ blocks:
	\begin{align}
		\label{mu-blocks}
		\mu
		=
		\mu^{(1)}
		\cup
		\cdots
		\cup
		\mu^{(r)}
		=
		\left( \mu^{(1)}_1,\dots,\mu^{(1)}_{n_1} \right)
		\cup
		\cdots
		\cup
		\left( \mu^{(r)}_1,\dots,\mu^{(r)}_{n_r} \right)
		.
	\end{align}
	To this we associate an admissible vector $\vec{\mu}^{+}$, obtained by taking the dominant ordering of each sub-composition $\mu^{(i)}$:
	\begin{align*}
		\vec{\mu}^{+}
		=
		\left( \mu^{(1)+}, \dots, \mu^{(r)+} \right).
	\end{align*}
	Fix a further admissible vector
	\begin{align*}
		\vec{\lambda}
		=
		\left( \lambda^{(1)}, \dots, \lambda^{(r)} \right),
	\end{align*}
	where each $\lambda^{(i)}$ is a strict signature. We have the following identity between transition probabilities in the $n$-species ASEP (with $n$-distinct particles) and the $r$-species ASEP:
	\begin{align}
		\label{n-to-p}
		\sum_{\nu:\nu^{(1)+} = \lambda^{(1)},\dots, \nu^{(r)+} = \lambda^{(r)}}
		\mathbb{P}_t(\mu \rightarrow \nu)
		=
		\mathbb{P}_t\left(\vec{\mu}^{+} \rightarrow \vec{\lambda}\right),
	\end{align}
	where the composition $\nu$ appearing on the right hand side is given by
	\begin{align}
		\label{nu-blocks}
		\nu
		=
		\nu^{(1)}
		\cup
		\cdots
		\cup
		\nu^{(r)}
		=
		\left( \nu^{(1)}_1,\dots,\nu^{(1)}_{n_1} \right)
		\cup
		\cdots
		\cup
		\left( \nu^{(r)}_1,\dots,\nu^{(r)}_{n_r} \right),
	\end{align}
	and where the sum is taken over blocks $\nu^{(i)}$ which are permutable to $\lambda^{(i)}$, for each $1 \leq i \leq r$.
\end{prop}

\begin{proof}
	This is a manifestation of a key property of coloured stochastic vertex models, known as {\it colour blindness}; see \cite[Chapter 2]{borodin_coloured_2018}. 
	
	A direct explanation of \eqref{n-to-p} may be sketched as follows. Consider the left hand side of \eqref{n-to-p}, and in particular the first block $\nu^{(1)}$ of the composition $\nu$. The quantity $\sum_{\nu:\nu^{(1)+} = \lambda^{(1)}} \mathbb{P}_t(\mu \rightarrow \nu)$ computes the probability that the multi-species ASEP, initialized in state $\mu$, has colours $\{1,\dots,n_1\}$ distributed over positions $\{\lambda^{(1)}_1,\dots,\lambda^{(1)}_{n_1}\}$ at time $t$, without paying heed to the relative ordering of these colours within the final state. Ignoring relative ordering of this contiguous block of colours is equivalent to setting all of these colours equal in the dynamics. A similar argument applies when imposing restrictions on all of the blocks $\nu^{(1)},\dots,\nu^{(r)}$; setting colours equal within each of the $r$ blocks yields the right hand side of \eqref{n-to-p}.
\end{proof}

\begin{thm}[Total crossing of blocks]
	\label{thm:total-cross}
	Fix an integer $r \geq 1$ and let $(n_1,\dots,n_r)$ be a vector of positive integers such that $\sum_{k=1}^{r} n_k = n$. Fix two vectors of strict signatures
	\begin{align*}
		\vec{\mu} &= \left( \mu^{(1)}, \dots, \mu^{(r)} \right),
		\qquad
		\mu^{(i)} = \left( \mu^{(i)}_1 > \cdots > \mu^{(i)}_{n_i} \right),
		\\
		\vec{\lambda} &= \left( \lambda^{(1)}, \dots, \lambda^{(r)} \right),
		\qquad
		\lambda^{(i)} = \left( \lambda^{(i)}_1 > \cdots > \lambda^{(i)}_{n_i} \right),
	\end{align*}
	which satisfy the constraints $\mu^{(i)}_a > \mu^{(j)}_b$, $\lambda^{(i)}_a < \lambda^{(j)}_b$ for all $1\leq i <j \leq r$ and all $a \in \{1,\dots,n_i\}$, $b \in \{1,\dots,n_j\}$ (these constraints imply, in particular, admissibility of both vectors). We then have the following formula:
	\begin{multline}
		\label{main-formula}
		\mathbb{P}(\vec{\mu} \rightarrow \vec{\lambda};t)
		=
		\left(\frac{1-q}{2\pi\ii}\right)^n
		\oint_{C}
		\dd z_1
		\cdots 
		\oint_{C}
		\dd z_n
		\prod_{1 \leq i<j \leq n}
		\left(
		\frac{z_j-z_i}{z_j-q z_i}
		\right)
		\\
		\times
		\prod_{j=1}^{n}
		\exp\left[\frac{(1-q)^2 z_j t}{(1-z_j)(1-q z_j)}\right]
		\frac{1}{(1-z_j)(1-qz_j)}
		\prod_{k=1}^{r}
		\xi_{\mu^{(k)}} 
		\left( z^{(k)}_1,\dots,z^{(k)}_{n_k} \right)
		{\sf F}_{\lambda^{(k)}}
		\left( z^{(k)}_1,\dots,z^{(k)}_{n_k} \right)
	\end{multline}
	where we have defined the rational functions
	\begin{align}
		\nonumber
		\xi_{\mu}(u_1,\dots,u_N)
		&=
		\prod_{j=1}^{N}
		\left(
		\frac{1-qu_j}{1-u_j}
		\right)^{\mu_j},
		\\
		{\sf F}_{\lambda}(u_1,\dots,u_N)
		&=
		\sum_{\sigma \in S_N}
		\prod_{1 \leq i<j \leq N}
		\left( \frac{u_{\sigma_j} - q u_{\sigma_i}}{u_{\sigma_j} - u_{\sigma_i}} \right)
		\prod_{i=1}^{N}
		\left( \frac{1-u_{\sigma_i}}{1-qu_{\sigma_i}} \right)^{\lambda_i}.
		\label{F}
	\end{align}
\end{thm}

\begin{proof}
	Start with two integer compositions $\mu$ and $\nu$; $\mu$ is given by \eqref{mu-blocks} and $\nu$ is given by \eqref{nu-blocks}. We impose on the parts of $\mu$ the constraints $\mu^{(i)}_a > \mu^{(i)}_{a+1}$ and $\mu^{(i)}_a > \mu^{(j)}_b$ for all $i<j$  (this is equivalent to demanding that $\mu_1 > \cdots > \mu_n$). On the parts of $\nu$ we impose the constraints $\nu^{(i)}_a < \nu^{(j)}_b$ for all $i<j$ (we say nothing about the ordering of parts within a common block).
	
	We then apply equation \eqref{asep-int} to compute 
	$\mathbb{P}_t(\mu \rightarrow \nu)$; the block-ordering assumption $\nu^{(i)}_a < \nu^{(j)}_b$ allows us to invoke Proposition \ref{prop:mu-block}, and we have
	\begin{multline*}
		\mathbb{P}(\mu \rightarrow \nu;t)
		=
		\frac{(-q^{-\frac{1}{2}})^{|\nu|}}{(2\pi\ii)^n}
		\oint_{C}
		\frac{\dd z_1}{z_1}
		\cdots 
		\oint_{C}
		\frac{\dd z_n}{z_n}
		\prod_{1 \leq i<j \leq n}
		\left(
		\frac{z_j-z_i}{z_j-q z_i}
		\right)
		\prod_{j=1}^{n}
		\exp\left[\frac{(1-q)^2 z_j t}{(1-z_j)(1-q z_j)}\right]
		\frac{1}{1-z_j}
		\\
		\times
		\prod_{k=1}^{r}
		\xi_{\mu^{(k)}} 
		\left( z^{(k)}_1,\dots,z^{(k)}_{n_k} \right)
		f_{\nu^{(k)}}
		\left(
		q^{-\frac{1}{2}}/z^{(k)}_1,\dots,q^{-\frac{1}{2}}/z^{(k)}_{n_k};q,q^{-\frac{1}{2}}
		\right).
	\end{multline*}
	We then perform the summation \eqref{n-to-p}; namely, we sum each block $\nu^{(i)}$ over all compositions which are permutable to $\lambda^{(i)}$, where $\lambda^{(1)},\dots,\lambda^{(r)}$ are the signatures given in the statement of the theorem. We may carry out these sums using \eqref{f-to-F}:
	\begin{multline*}
		\sum_{\nu:\nu^{(1)+} = \lambda^{(1)},\dots, \nu^{(r)+} = \lambda^{(r)}}
		\mathbb{P}(\mu \rightarrow \nu;t)
		=
		\frac{(-q^{-\frac{1}{2}})^{|\vec{\lambda}|}}{(2\pi\ii)^n}
		\oint_{C}
		\frac{\dd z_1}{z_1}
		\cdots 
		\oint_{C}
		\frac{\dd z_n}{z_n}
		\prod_{1 \leq i<j \leq n}
		\left(
		\frac{z_j-z_i}{z_j-q z_i}
		\right)
		\\
		\times
		\prod_{j=1}^{n}
		\exp\left[\frac{(1-q)^2 z_j t}{(1-z_j)(1-q z_j)}\right]
		\frac{1}{1-z_j}
		\\
		\times
		\prod_{k=1}^{r}
		\xi_{\mu^{(k)}} 
		\left( z^{(k)}_1,\dots,z^{(k)}_{n_k} \right)
		F_{\lambda^{(k)}}
		\left(
		q^{-\frac{1}{2}}/z^{(k)}_1,\dots,q^{-\frac{1}{2}}/z^{(k)}_{n_k};q,q^{-\frac{1}{2}}
		\right),
	\end{multline*}
	where $F_{\lambda^{(k)}}$ is given by \eqref{F-sym} and $|\vec{\lambda}| = \sum_{k=1}^{r} |\lambda^{(k)}|$. The claim \eqref{main-formula} follows from \eqref{n-to-p} and the identity
	\begin{align*}
		(-q^{-\frac{1}{2}})^{|\lambda|}
		\prod_{i=1}^{N} u_i^{-1}
		\cdot
		F_{\lambda}
		\left(q^{-\frac{1}{2}}u_1^{-1},\dots,q^{-\frac{1}{2}}u_N^{-1};q,q^{-\frac{1}{2}}\right)
		=
		\prod_{i=1}^{N}
		\frac{1-q}{1-q u_i}
		\cdot
		{\sf F}_{\lambda}(u_1,\dots,u_N),
	\end{align*}
	which holds for all strict signatures $\lambda = (\lambda_1 > \cdots > \lambda_N)$.
\end{proof}

A number of special cases of the formula \eqref{main-formula} are of interest. One of these has already been seen in \eqref{n-colours-total}; this is the $r=n$ version of \eqref{main-formula}, when all composition blocks have size $1$, and leads to a totally factorized integrand. 

At the other extreme, we may also consider $r=1$, when there is a single block consisting of $n$ parts. In that case we have the following result:
\begin{cor}
	Let $\vec{\lambda} = (\lambda)$ and $\vec{\mu} = (\mu)$ be one-dimensional vectors, where $\lambda$ and $\mu$ are strict signatures of length $n$. Transition probabilities in the (single-species) ASEP are given by the formula
	\begin{multline*}
		\mathbb{P}(\vec{\mu} \rightarrow \vec{\lambda};t)
		=
		\left(\frac{1-q}{2\pi\ii}\right)^n
		\oint_{C}
		\dd z_1
		\cdots 
		\oint_{C}
		\dd z_n
		\prod_{1 \leq i<j \leq n}
		\left(
		\frac{z_j-z_i}{z_j-q z_i}
		\right)
		\\
		\times
		\prod_{j=1}^{n}
		\exp\left[\frac{(1-q)^2 z_j t}{(1-z_j)(1-q z_j)}\right]
		\frac{1}{(1-z_j)(1-qz_j)}
		\xi_{\mu} 
		\left( z_1,\dots,z_n \right)
		{\sf F}_{\lambda}
		\left( z_1,\dots,z_n \right).
	\end{multline*}
	This recovers the Tracy--Widom formula for ASEP transition probabilities \cite{tracy2008integral}.
\end{cor}

\begin{proof}
	This is the $r=1$ degeneration of \eqref{main-formula}.
\end{proof}

\subsection{Reduction to TASEP}
\label{ssec:tasep}

We conclude by studying the $q \rightarrow 0$ case of the results from the previous section; this limit is straightforwardly taken.

\begin{prop}
	At $q=0$, the function \eqref{F} is given by a determinant:
	\begin{align}
		\label{F-det}
		{\sf F}_{\lambda}(u_1,\dots,u_N;q=0)
		=
		\prod_{1 \leq i<j \leq N}
		\frac{1}{u_j-u_i}
		\cdot
		\det\left[
		u_j^{i-1}
		(1-u_j)^{\lambda_i}
		\right]_{1 \leq i,j \leq N}.
	\end{align}
\end{prop}

\begin{proof}
	We substitute $q=0$ into \eqref{F}, which yields
	\begin{align*}
		{\sf F}_{\lambda}(u_1,\dots,u_N;q=0)
		&=
		\prod_{1 \leq i<j \leq N}
		\frac{1}{u_j-u_i}
		\cdot
		\sum_{\sigma \in S_N}
		(-1)^{\abs{\sigma}}
		\cdot
		\prod_{i=1}^{N}
		u_{\sigma_i}^{i-1}
		\left(1-u_{\sigma_i}\right)^{\lambda_i},
	\end{align*}
	which is the Laplace expansion of the determinant \eqref{F-det}
\end{proof}

\begin{cor}
	\label{cor:pTASEPcross}
	Let $\vec{\mu}$ and $\vec{\lambda}$ be vectors of strict signatures, prescribed in the same way as in the statement of Theorem \ref{thm:total-cross}. The probability of total crossing of blocks in the multi-species TASEP is given by
	\begin{multline*}
		\mathbb{P}(\vec{\mu} \rightarrow \vec{\lambda};t)
		=
		\left(\frac{1}{2\pi\ii}\right)^n
		\oint_{C}
		\dd z_1
		\cdots 
		\oint_{C}
		\dd z_n
		\prod_{j=1}^{n}
		\exp\left[\frac{z_j t}{1-z_j}\right]
		\prod_{1 \leq k<\ell \leq r}
		\prod_{i=1}^{n_k}
		\prod_{j=1}^{n_{\ell}}
		\left( z^{(\ell)}_j-z^{(k)}_i \right)
		\\
		\times
		\prod_{k=1}^{r}
		\det\left[ z^{(k)\left(i-j-N_k\right)}_j 
		\left(1-z^{(k)}_j\right)^{\lambda^{(k)}_i-\mu^{(k)}_j-1} \right]_{1 \leq i,j \leq n_k}
	\end{multline*}
	where for each $1 \leq k \leq r$ we define $N_k = \sum_{a=1}^{k-1} n_a$.
\end{cor}

\begin{proof}
	This is the $q=0$ reduction of \eqref{main-formula}; it makes use of the determinant formula \eqref{F-det}, and arises after noting partial cancellation between Vandermonde factors in the integrand and those appearing in \eqref{F-det}.
\end{proof}

\section{Cumulative total crossing probability for the 2-TASEP under step and Bernoulli initial conditions}
\label{2TASEP step initial condition section}

A key observable for the 2-TASEP is the probability that all particles of type 2 start out to the left of all particles of type 1 and end up to the right of those of type 1. The latter is a stable final configuration because particles of type 1 are not allowed to overtake those of type 2. We thus consider the following cumulative probability of total crossing subject to both Bernoulli and step initial conditions using the transition probability derived in Proposition \ref{2TASEP crossing det formula statement} and Corollary \ref{cor:pTASEPcross}. 
\begin{defn}
	\label{2TASEP Pcross step definition statement}
	For some $s_1,s_2 \in \Z$, we define the cumulative total crossing probability for 2-TASEP with initial condition $\mu\in \Z^n$ as 
	\begin{equation}
	\label{2TASEP Pcross definition}
	\Pcross 
	= \sum_{s_1\leq \nu_1<\cdots<\nu_{n-m}< s_2 }\hspace{0.1cm}\sum_{s_2\leq \nu_{n-m+1}<\cdots<\nu_n} \mathbb{P}(\mu\to\nu,\inp\to p;t),
	\end{equation}
	where $\inp_j=j$ and $p_j = n-m+j$.
\end{defn}

The probability $\Pcross$ is a sum over all configurations where all particles of type 1 end up between position $s_1$ and $s_2$, and all particles of type 2 end up to the right of $s_2$. We may regard $s_1$ and $s_2$ as fixed walls in space or consider them moving in time if we scale $s_1$ and $s_2$ accordingly. Note that to get a non-zero probability we must have $s_2-s_1\geq n-m$.

We now proceed to evaluate $\Pcross$.
\begin{prop}
	\label{2TASEP step Pcross formula proposition}
	\begin{multline}
		\label{2TASEP Pcross 1}
		\Pcross=\oint \prod_{i=1}^{m} \frac{\dd z_{i}}{2 \pi \ii} \prod_{i=1}^{n-m} \frac{\dd w_{i}}{2 \pi \ii} \prod_{i=1}^{m} \prod_{j=1}^{n-m}\left(w_{j}-z_{i}\right)\prod_{1\leq i<j\leq m}(z_j-z_i) \\ \times \prod_{i=1}^{m} \frac{e^{(z_i^{-1}-1)t}z_{i}^{s_2-1-\mu_i}}{\left(1-z_{i}\right)^{n-i+1}} \prod_{i=1}^{n-m}\frac{e^{(w_i^{-1}-1)t}w_{i}^{s_1-1-\mu_{i+m}}}{\left(1-w_{i}\right)^{n-m-i+1}}\operatorname{det}\left(w_{i}^{j-1}-w_{i}^{s_2-s_1}\right)_{1 \leq i,j \leq n-m}.
	\end{multline}
\end{prop}

\begin{proof}
	Consider initially Proposition \ref{2TASEP crossing det formula statement}, from which we find the crossing probability as
	\begin{multline}
		\label{2TASEP Pcross 1 Pcross}
		\Pcross = \oint \prod_{i=1}^{m} \frac{\dd z_i }{2\pi \ii} \prod_{j=1}^{n-m} \frac{\dd w_j }{2\pi \ii} \prod_{i=1}^m\frac{e^{(z_i^{-1}-1)t}z_i^{-\mu_i-1}}{\left(1-z_{i}\right)^{n-m-i}} \prod_{i=1}^{n-m}\frac{e^{(w_i^{-1}-1)t}w_i^{-\mu_{m+i}-1}}{\left(1-w_i\right)^{-i}}\prod_{i=1}^{m} \prod_{j=1}^{n-m} \left(w_{j}-z_{i}\right) \\ \times A(z_1,\dots,z_m) B(w_1,\dots,w_{n-m}),
	\end{multline}
	where
	\begin{align*}
		A(z_1,\dots,z_m) & = \sum_{\nu_n = s_2+m-1}^\infty \sum_{\nu_{n-1} = s_2+m-2}^{\nu_n-1}  \cdots \sum_{\nu_{n-m+1}=s_2}^{\nu_{n-m+2}-1} \det \left(\frac{z_{i}^{\nu_{n-m+j}}}{ \left(1-z_{i}\right)^{j}}\right)_{1\leq i,j \leq m},\\
		B(w_1,\dots,w_{n-m}) & = \sum_{\nu_{n-m} = s_1+n-m-1}^{s_2-1} \cdots \sum_{\nu_1 = s_1}^{\nu_2 -1} \det\left(\frac{w_{i}^{\nu_{j}}}{\left(1-w_{i}\right)^{j}}\right)_{1 \leq i,j \leq n-m}.
	\end{align*}
	
	To evaluate $A$ we make use of the following two lemmas. We first express the determinant in $A$ as a sum over permutations and apply the following nested geometric series identity.
	\begin{lem}
		\label{summation identity 1}
		\begin{equation}
		\label{summation identity 1 equation}
		\sum_{\nu_m=s_2+m-1}^\infty \cdots \sum_{\nu_1=s_2}^{\nu_2-1} \prod_{i=1}^m z_i^{\nu_i} = \prod_{i=1}^m\frac{z_i^{s_2+i-1}}{1-\prod_{j=i}^m z_j}.
		\end{equation}
	\end{lem}
	\noindent The proof of this result is given in Appendix \ref{Geometric series proof appendix}. This result gives 
	\begin{equation}
	\label{2TASEP Pcross formula proof 2}
	A(z_1,\dots,z_m) = \sum_{\sigma\in S_m}(-1)^\abs{\sigma} \prod_{i=1}^m \left[\frac{1}{1-\prod_{j=i}^m z_{\sigma_j}}\frac{z_{\sigma_i}^{s_2+i -1}}{(1-z_{\sigma_i})^{i}}\right],
	\end{equation}
	\noindent Next we apply a symmetrization identity which is adapted from a result in \cite{tracy_integral_2008}.
	\begin{lem}
		\label{symmetrization identity 1}
		\begin{equation}
		\label{2TASEP Pcross 1 proof A}
		A(z_1,\dots,z_m)=\sum_{\sigma\in S_m}(-1)^\abs{\sigma} \prod_{i=1}^m \left[\frac{1}{1-\prod_{j=i}^m z_{\sigma_j}}\frac{z_{\sigma_i}^{s_2+i -1}}{(1-z_{\sigma_i})^{i}}\right] = \prod_{1\leq i < j \leq m}\left(z_j-z_i\right) \prod_{i=1}^m\frac{z_i^{s_2}}{(z_i-1)^{m+1}}.
		\end{equation}
	\end{lem}
	\noindent A proof of this result is given in Appendix \ref{Symmetrization identity proof appendix}. To evaluate $B$ we factorize 
	\begin{multline}
		B(w_1,\dots,w_{n-m}) = \\ \prod_{i=1}^{n-m}\frac{1}{(1-w_i)^{n-m+1}}\sum_{\nu_{n-m} = s_1+n-m-1}^{s_2-1} \cdots \sum_{\nu_1 = s_1}^{\nu_2 -1} \det\left(w_{i}^{\nu_{j}}\left(1-w_{i}\right)^{n-m-j+1}\right)_{1 \leq i,j \leq n-m}.
		\label{2TASEP Pcross 1 proof B 1}
	\end{multline}
	Noting that the $j$-th column in the determinant in \eqref{2TASEP Pcross 1 proof B 1} depends only on $\nu_j$, we evaluate the sum over $\nu_1$ as follows 
	\begin{align*}
		\sum_{\nu_1 = 1}^{\nu_2 -1} & \det\left(w_{i}^{\nu_{j}}\left(1-w_{i}\right)^{n-m-j+1}\right)_{1 \leq i,j \leq n-m} \\ & \hspace{0.5cm} = \det\left(\left[\left(w_i^{s_1}-w_i^{\nu_2}\right)(1-w_i)^{n-m-1}\right]_{1\leq i \leq n-m}\Big|\left[w_{i}^{\nu_{j}}\left(1-w_{i}\right)^{n-m-j+1}\right]_{\substack{1\leq i \leq n-m \\2\leq j \leq n-m}}\right)
		\\ & \hspace{0.5cm} = \det\left(\left[ w_i^{s_1}(1-w_i)^{n-m-1}\right]_{1\leq i \leq n-m}\Big|\left[w_{i}^{\nu_{j}}\left(1-w_{i}\right)^{n-m-j+1}\right]_{\substack{1\leq i \leq n-m \\2\leq j \leq n-m}}\right),
	\end{align*}
	where in the last line we performed the column operation $C_1 \mapsto C_1 + C_2$. Continuing with the procedure of summing over $\nu_j$ and the performing the column operation $C_j \mapsto C_j + C_{j+1}$ for $1\leq j \leq n-m-1$, and summing over $\nu_{n-m}$ without performing a column operation gives	
	\begin{multline}
		\label{2TASEP Pcross 1 proof B 2}
		B(w_1,\dots,w_{n-m})=\prod_{i=1}^{n-m}\frac{1}{(1-w_i)^{n-m+1}}\\
		\times \det(\left(w_i^{s_1+j-1}(1-w_i)^{n-m-j}\right)_{\substack{1\leq i \leq n-m \\1\leq j \leq n-m-1}}\Big| \left(w_i^{s_1+n-m-1}-w_i^{s_2}\right)_{1\leq i \leq n-m}).
	\end{multline}
	The determinant above may be expressed in an alternative form by again undergoing a series of column operations. The $i$-th step in the column operation scheme maps
	\[C_j \mapsto C_j+C_{j+1} \text{ for $1 \leq j \leq n-m-i-1$ and }C_{n-m-i} \mapsto C_{n-m-i}+C_{n-m}.\]
	Performing $n-m-1$ steps in the scheme reduces \eqref{2TASEP Pcross 1 proof B 2} to 
	\begin{equation}
	\label{2TASEP Pcross 1 proof B 3}
	B(w_1,\dots,w_{n-m})=\prod_{i=1}^{n-m}\frac{1}{(1-w_i)^{n-m+1}} \det(w_i^{j+s_1-1}-w_i^{s_2})_{1\leq i \leq n-m}.
	\end{equation}
	Finally combining \eqref{2TASEP Pcross 1 proof A} and \eqref{2TASEP Pcross 1 proof B 3} into \eqref{2TASEP Pcross 1 Pcross} gives the result. 
\end{proof}

A further simplification is obtained by choosing special initial conditions. We will choose Bernoulli distributed initial conditions for type 2 particles and step initial conditions for particles of type 1.
\begin{defn}
	Consider a random variable $\chi_x$ which has probability $\rho$ of a particle occupancy and $\rho-1$ of a vacancy at site $x\in \bbZ$. The Bernoulli initial data is the initial occupancy
	\begin{equation}
	\eta_x(0) = \begin{cases}
	2\chi_x & \text{for } x< 0 \\
	1 & \text{for } 0\leq x \leq n-m-1 \\
	0 & \text{for } x \geq n-m \\
	\end{cases}.
	\end{equation}
\end{defn}
The probability of a complete set of initial conditions at the coordinates $\mu \in \bbZ^n$ is then given by
\begin{multline}
	\mathbb{P}(\mu;0) =\\ \rho^m(1-\rho)^{-\mu_m-1}\prod_{i=1}^{m-1}(1-\rho)^{\mu_{m-i+1}-\mu_{m-i}-1} \prod_{i=1}^{n-m} \delta_{\mu_{m+i},i-1}  = \rho^m(1-\rho)^{-\mu_1-m} \prod_{i=1}^{n-m} \delta_{\mu_{m+i},i-1} .
\end{multline}
\begin{defn}
	The cumulative total crossing probability of the 2-TASEP subject to the Bernoulli initial data for particles of type 2 is defined by
	\begin{equation}
	\label{Pcross Bernoulli defn}
	\PB := \sum_{\mu} \mathbb{P}(\mu;0)\Pcross = \sum_{\mu_1 < \mu_2 < \cdots <\mu_m<0} \mathbb{P}(\mu;0)\Pcross,
	\end{equation}
	where in the second equality we have $\mu_{m+i}=i-1$. We refer to this as the Bernoulli total crossing probability.
\end{defn}
\begin{prop}
	\label{prop:cumtotalcross}
	The cumulative total crossing probability with Bernouilli-step initial condition is given by
	\begin{multline}
		\label{Bernoulli Pcross 1}
		\PB = \frac{\rho^m}{m!}\oint \prod_{i=1}^{m} \frac{\dd z_{i}}{2 \pi \ii} \prod_{i=1}^{n-m} \frac{\dd w_{i}}{2 \pi \ii} \prod_{i=1}^{m} \prod_{j=1}^{n-m}\left(w_{j}-z_{i}\right)\prod_{ i\neq j}(z_j-z_i) \\ \times \prod_{i=1}^{m} \frac{e^{(z_i^{-1}-1)t}z_{i}^{s_2}}{\left(1-z_{i}\right)^{n}(1-(1-\rho)z_i)} \prod_{i=1}^{n-m}\frac{e^{(w_i^{-1}-1)t}w_{i}^{s_1-i}}{\left(1-w_{i}\right)^{n-m-i+1}}\operatorname{det}\left(w_{i}^{j-1}-w_{i}^{s_2-s_1}\right)_{1 \leq i,j \leq n-m}.
	\end{multline}
\end{prop}

\begin{proof}
	Initially consider $\PB$ as defined in \eqref{Pcross Bernoulli defn} where we use Proposition \ref{2TASEP step Pcross formula proposition}. We fix the initial positions of the type 1 particles at $\mu_{m+i}=i-1$. Then making use of the following two corollaries, we perform a similar process to the proof of Proposition \ref{2TASEP step Pcross formula proposition}.
	\begin{cor}[of Lemma \ref{summation identity 1}]
		\begin{equation}
		\sum_{\mu_1 < \mu_2 < \cdots <\mu_m<0} (1-\rho)^{-\mu_1}\prod_{i=1}^m z_i^{-\mu_i} = (1-\rho)^m \prod_{i=1}^m \frac{z_i^{m-i+1}}{1-(1-\rho)\prod_{j=1}^i z_j}.
		\end{equation}
	\end{cor}
	\noindent Incorporating this lemma into \eqref{Pcross Bernoulli defn} gives the sum over initial conditions. We then symmetrize the integrand using the following result.
	\begin{cor}[of Lemma \ref{symmetrization identity 1}, see Lemma 4.4 from \cite{chen_exact_2019}].
		\begin{equation}
		\sum_{\sigma\in S_m}(-1)^\abs{\sigma} \prod_{i=1}^m \left[\left(\frac{1-z_{\sigma_i}}{z_{\sigma_i}}\right)^i \frac{1}{1-(1-\rho)\prod_{j=1}^i z_{\sigma_j}}\right] = \prod_{1\leq i < j \leq m}\left(z_i-z_j\right) \prod_{i=1}^m\frac{1-z_i}{z_i^m(1-(1-\rho)z_i)}.
		\end{equation}
	\end{cor}
	\noindent This gives the result. 
\end{proof}
Finally, it is a simple corollary to obtain the result for step-step initial condition:
\begin{cor}
	The total crossing probability for the step initial condition can be calculated by taking the limit
	\begin{equation}
	\PS = \lim_{\rho\to 1}\PB.
	\end{equation}
\end{cor}

\subsection{The case of one relevant wall}
When $s_1 \leq -m$ this wall can be considered trivial and of no effect as no configurations will be excluded because particles can no longer be to the left of $s_1$. This effect is illustrated in Figure \ref{fig:trivial wall diagram}. In this case the cumulative total crossing probability further simplifies.

First note that we can express the cumulative total crossing probability \eqref{2TASEP Pcross 1} in an alternative form by inverting the integration variables.
\begin{lem}
	\label{Bernoulli Pcross integration inversion}
	Making the change of variables $z_i \mapsto z_i^{-1},w_i \mapsto w_i^{-1}$ in \eqref{Bernoulli Pcross 1} gives the total crossing probability as 
	\begin{multline}
		\label{Bernoulli Pcross 2}
		\PB = \frac{\rho^m}{m!}\oint_{0,1,1-\rho} \prod_{i=1}^m \frac{\dd z_{i}}{2 \pi \ii} \prod_{i \neq j} (z_j-z_i) \prod_{i=1}^m \frac{e^{(z_i-1)t}z_i^{-s_2-m+1}}{(z_i-1)^{n}(z_i-1+\rho)}\\ 
		\times \oint_{0,1} \prod_{i=1}^{n-m} \frac{\dd w_{i}}{2 \pi \ii}  \prod_{i=1}^m\prod_{j=1}^{n-m} (z_i-w_j) \prod_{i=1}^{n-m} \frac{e^{(w_i-1)t}w_i^{-s_1-m}}{(w_i-1)^{n-m-i+1}} \\
		\times \det\left(w_i^{n-m-j} - w_i^{n-m+s_1-s_2-1}\right)_{1\leq i,j \leq n-m}.
	\end{multline}
\end{lem}
\noindent The proof of this lemma is straightforward.

We may now simplify this expression for the crossing probability using our next result. 
\begin{prop}
	\label{2TASEP Pcross integral collapse}
	When $s_1 \leq -m$ the $(n-m)$-fold integral over the $w$-variables in Proposition \ref{Bernoulli Pcross integration inversion} collapses into one integral 
	\begin{multline}
		\label{2TASEP Pcross 3}
		\PB= \frac{\rho^m}{m!} \oint_{0,1,1-\rho} \prod_{i=1}^{m} \frac{\dd z_{i}}{2 \pi \ii} \prod_{i \neq j} \left(z_{j}-z_{i}\right) \prod_{i=1}^{m} \frac{\mathrm{e}^{\left(z_{i}-1\right)t}z_{i}^{-s_2-m+1}}{\left(z_{i}-1\right)^{m+1}(z_i-1+\rho)} \\
		\times \oint_{0} \frac{\dd w}{2 \pi \ii} \frac{\mathrm{e}^{(w-1) t}w^{n-2m-s_2-1}}{(w-1)} \prod_{i=1}^{m}\left(w-z_{i}\right).
	\end{multline}
\end{prop}
\begin{figure}
	\begin{center}
		\begin{tikzpicture}[scale=0.70]
		\draw[->] (0,0) -- (8,0);
		\draw[->] (0,0) -- (-6,0);
		\draw[->] (0,0) -- (0,8);
		
		\draw[fill=red,thick] (-3,0) circle (0.3cm); 
		\draw[fill=red,thick] (-2,0) circle (0.3cm); 
		\draw[fill=red,thick] (-1,0) circle (0.3cm);
		\draw[fill=blue,thick] (0,0) circle (0.3cm); 
		\draw[fill=blue,thick] (1,0) circle (0.3cm); 
		\draw[fill=blue,thick] (2,0) circle (0.3cm);  
		
		\draw[thick,red,->] (-3,0) -- (-3,2) -- (2.5,7.5);
		\draw[thick,red,->] (-2,0) -- (-2,1) -- (4.5,7.5);
		\draw[thick,red,->] (-1,0) -- (6.5,7.5);
		\draw[thick,blue,->] (0,0) -- (-3,3) -- (-3,5) -- (-0.5,7.5); 
		\draw[thick,blue,->] (1,0) -- (1,1) -- (-2,4) -- (-2,5) -- (0.5,7.5);
		\draw[thick,blue,->] (2,0) -- (2,2) -- (-1,5) -- (1.5,7.5);  
		
		\draw[green] (-4,0) -- (-4,8);
		
		\node (x) at (8.2,0) {$\nu$};
		\node (t) at (0.2,8) {$t$};
		\node (s1) at (-4.2,4.5) {$s_1$};
		
		\node (2) at (-3,0) {2};
		\node (2) at (-2,0) {2};
		\node (2) at (-1,0) {2};
		\node (1) at (0,0) {1};
		\node (1) at (1,0) {1};
		\node (1) at (2,0) {1};
		\end{tikzpicture}
	\end{center}
	\caption{Space-time diagram total crossing configuration with an irrelevant wall at $s_1\leq -m$.}
	\label{fig:trivial wall diagram}
\end{figure}
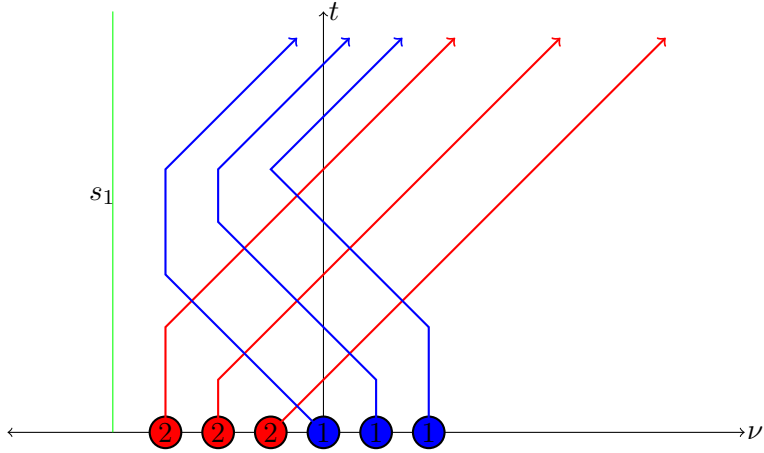

\begin{rmk}
	Consider the case of 1 particle of type 1 by setting $n=m+1$. We take the limit $\rho \to 1$ so that we deal with the step initial condition and impose the restriction $s_1 \leq -m$. By re-labelling the integration variable $w=z_n$, it is observed that the crossing probability has the form
	\begin{equation}
	\label{2TASEP connection with TASEP crossing prob}
	\PS = \Gamma_{n-1}(s_2) - \Gamma_{n}(s_2).
	\end{equation}
	Here $\Gamma_n(s)$ is the probability that the all particles of the $n$-particle single species TASEP with step initial condition ($\mu_i = i$) cross a wall at position $s>n$, which can be obtained by summing the TASEP transition probability \cite{schutz_exact_1997} over configurations where all particles are beyond $s$ and then symmetrizing the result. This TASEP crossing probability is given by
	\begin{equation}
	\Gamma_{n}(s) = \frac{1}{n!}\oint_{0,1}\prod_{i=1}^n \frac{\dd z_i}{2 \pi \ii} \prod_{i=1}^n \frac{e^{(z_i-1)t}z_i^{1-s}}{(z_i-1)^{n}} \prod_{i \neq j} (z_j-z_i).
	\end{equation}
	Under these restrictions with an appropriate asymptotic analysis on $\Gamma_{n}(s)$, the Tracy-Widom for the GUE, $F_2$, emerges. In the scaling limit we should recover the results of Nejjar \cite{nejjar_kpz_2020} from \eqref{2TASEP connection with TASEP crossing prob}.
\end{rmk}

\begin{proof}[Proof of Proposition \ref{2TASEP Pcross integral collapse}]
	Starting with \eqref{Bernoulli Pcross 2} with $s_1 \leq -m$, we permute the columns in the determinant by mapping $j\mapsto n-m-j+1$ which has signature $(-1)^{(n-m)(n-m-1)/2}$
	\begin{multline}
	\det\left(w_i^{n-m-j} - w_i^{n-m+s_1-s_2-1}\right)_{1\leq i,j \leq n-m} \\
	= (-1)^{(n-m)(n-m-1)/2} \det\left(w_i^{j-1} - w_i^{n-m+s_1-s_2-1}\right)_{1\leq i,j \leq n-m}.
	\end{multline}
	We then notice that the singularity at $w_{n-m} = 1$ is removable due to a zero in the determinant. This can be seen by performing the column operations $C_i \mapsto C_i - C_1$ for $i=2,\dots,n-m$ on the determinant as follows
	\begin{multline*}
		\det\left(w_i^{j-1} - w_i^{n-m+s_1-s_2-1}\right)_{1\leq i,j \leq n-m} \\ = \Scale[0.9]{\left|
			\begin{array}{c c c c}
				1 - w_1^{n-m+s_1-s_2-1} & w_1 - w_1^{n-m+s_1-s_2-1} & \cdots & w_1^{n-m-1} - w_1^{n-m+s_1-s_2-1} \\
				\vdots & \vdots & & \vdots \\
				1 - w_{n-m}^{n-m+s_1-s_2-1} & w_{n-m} - w_{n-m}^{n-m+s_1-s_2-1} & \cdots & w_{n-m}^{n-m-1} - w_{n-m}^{n-m+s_1-s_2-1} \\
			\end{array}\right|} \\
		= \left|
		\begin{array}{c c c c}
			1 - w_1^{n-m+s_1-s_2-1} & w_1 -1 & \cdots & w_1^{n-m-1} - 1 \\
			\vdots & \vdots & & \vdots \\
			1 - w_{n-m}^{n-m+s_1-s_2-1} & w_{n-m} - 1 & \cdots & w_{n-m}^{n-m-1} - 1 \\
		\end{array}\right|,
	\end{multline*}
	where is becomes clear that $w_{n-m} = 1$ is a root of every entry along the bottom row of the determinant which cancels the $(1-w_{n-m})^{-1}$ factor in the integrand. Proceeding to expand the determinant along the bottom row gives
	\begin{multline}
		\label{2TASEP Pcross 4}
		\det\left(w_i^{j-1} - w_i^{n-m+s_1-s_2-1}\right)_{1\leq i,j \leq n-m} \\ 
		= (-1)^{n-m-1} \left(1 - w_{n-m}^{n-m+s_1-s_2-1} \right) \det(w_i^j-1)_{1 \leq i,j \leq n-m-1} \\
		+ \sum_{k=2}^{n-m} \left(w_{n-m}^{k-1}-1\right)(-1)^{n-m-k} Q_k(w_1, \dots, w_{n-m-1}),
	\end{multline}
	for some rational functions $Q_k$, none of which depend on $w_{n-m}$. We then notice that the only term in \eqref{2TASEP Pcross 4} which gives a singularity at $w_{n-m} = 0$ is the first one since $s_2-s_1>n-m$.
	
	The remaining determinant may be evaluated as
	\begin{equation}
	\det(w_i^j-1)_{1 \leq i,j \leq n-m-1} = \prod_{i=1}^{n-m-1} (w_i-1) \prod_{1 \leq i < j \leq n-m-1} (w_j-w_i), 
	\end{equation}
	where we have used column operations and the identity for $\ell\in \N$ 
	\[\zeta^\ell-1 = (\zeta-1) \sum_{k=0}^{\ell-1}\zeta^k.\]
	It follows that \eqref{Bernoulli Pcross 2} becomes
	\begin{multline}
		\label{2TASEP Pcross 5}
		\PB = (-1)^{(n-m)(n-m+1)/2} \frac{\rho^m}{m!} \oint_{0,1,1-\rho} \prod_{i=1}^m \frac{\dd z_{i}}{2 \pi \ii} \prod_{i \neq j} (z_j-z_i) \prod_{i=1}^m \frac{e^{(z_i-1)t} z_i^{-s_2-m+1}}{(z_i-1)^{n}(z_i-1+\rho)}\\ 
		\times \oint_{1} \prod_{i=1}^{n-m-1} \frac{\dd w_{i}}{2 \pi \ii} \oint_{0} \frac{\dd w_{n-m}}{2 \pi \ii} w_{n-m}^{n-m+s_1-s_2-1}\prod_{i=1}^m\prod_{j=1}^{n-m} (w_j-z_i) 
		\prod_{i=1}^{n-m} \frac{ e^{(w_i-1)t}w_i^{-s_1-m}}{(w_i-1)^{n-m-i+1}} \\
		\times \prod_{i=1}^{n-m-1}(w_i-1)\prod_{1 \leq i < j \leq n-m-1} (w_j-w_i).
	\end{multline}
	The contour in $w_i$ for $1\leq i \leq n-m-1$ has been deformed so that it only surrounds the singularity at $w_i=1$ since the integrand is analytic at $w_i = 0$ as $s_1 \leq -m$.
	
	From here we are able to successively evaluate the $w$-integrals in \eqref{2TASEP Pcross 5}, starting with $w_{n-m-1}$ where there is a simple pole at $w_{n-m-1}=1$. The Vandermonde factor $\prod_{1 \leq i < j \leq n-m-1} (w_j-w_i)$ ensures that integrating in $w_{n-m-k}$ reduces the order of the pole at $w_{n-m-k-1}=1$ so that each integration in $w_i$ for $1\leq i \leq n-m-1$ is only about a simple pole at $w_i =1$. Performing these successive integrals and accounting for the sign yields the result \eqref{2TASEP Pcross 3}. 
\end{proof}

\begin{cor}
	\label{cor:cumtotalcross}
	When $s_1 \leq -m$ we have 
	\begin{equation}
	\PB = \rho^m \oint_{0}\frac{\dd w}{2\pi\ii} \frac{e^{(w-1)t}w^{n-2m-s_2-1}}{w-1} 
	\det(\oint_{0,1,1-\rho}\frac{\dd z}{2\pi\ii} \frac{e^{(z-1)t}z^{i+j-s_2-m-1}}{(z-1)^{m+1}(z-1+\rho)}(w-z))_{1\leq i,j \leq m}.
	\end{equation}
\end{cor}

\begin{proof}
	Starting with \eqref{2TASEP Pcross 3}, we factor $\prod_{i\neq j}(z_j-z_i)$ into the product of two Vandermonde determinants. The result then follows by applying the Cauchy-Binet identity.
\end{proof}

\appendix

\section{Proof of nested Bethe ansatz formula for the transition probability}
\label{2TASEP Green's function proof section}
This section presents a proof of Theorem \ref{2TASEP Green's function Theorem} which was deferred from Section \ref{2TASEP integral formula section}.
\begin{proof}[Proof of Theorem \ref{2TASEP Green's function Theorem}]
	To prove the theorem we show the Green's function \eqref{2TASEP Greens with P} satisfies the following:
	\begin{enumerate}[label=\textbf{\roman*.},wide=0pt, leftmargin=\parindent]
		\item The free evolution \eqref{2TASEP master equation}.
		\item The boundary conditions \eqref{2TASEP BC1},\eqref{2TASEP BC2} and \eqref{2TASEP BC3}.
		\item The initial condition \eqref{2TASEP initial condition}.
	\end{enumerate}
	In what follows we let the 2-TASEP eigenfunction from \eqref{2TASEP_P} take the form
	\begin{equation}
	\label{2TASEP_P_ABC}
	P(\nu,p;t) = \sum_{\pi \in S_n} A_\pi B_\pi(p) \prod_{i=1}^n z_{\pi_i}^{\nu_i}e^{(z_i^{-1}-1)t} \sum_{\sigma \in S_m} C_{\pi,\sigma}(p),
	\end{equation}
	where 
	\begin{align*}
		A_\pi & = (-1)^{\abs{\pi}} \prod_{i=1}^n \left(\frac{1-z_i}{1-z_{\pi_i}} \right)^i,\\
		B_\pi(p) & = \prod_{i=1}^m \prod_{j=1}^{p_i} \frac{1}{1-z_{\pi_j}},\\
		C_{\pi,\sigma}(p) & = (-1)^{\abs{\sigma}} \prod_{i=1}^m \left(\frac{1-u_i}{1-u_{\sigma_i}} \right)^i  \prod_{i=1}^m \prod_{j=1}^{p_i-1} (u_{\sigma_i} - z_{\pi_j}).
	\end{align*}
	We prove here that the eigenfunction $\eqref{2TASEP_P_ABC}$ satisfies the free evolution and the boundary conditions, which automatically imply the the transition probability \eqref{2TASEP Greens with P} also satisfies them.
	
	We also recall that $s_i$ is a simple transposition which swaps $i$ and $i+1$ while leaving all other indices fixed. The signature of the permutation $\pi$ satisfies $(-1)^{\abs{\pi  s_i}}  = -(-1)^{\abs{\pi}}$. 
	
	\medskip\noindent\textbf{i. Free evolution \eqref{2TASEP master equation}:} 
	The only $t$-dependence in the transition probability is from is from the exponential factor in \eqref{2TASEP_P_ABC} so that the integrand of \eqref{2TASEP Greens with P} is differentiable in $t$. The result follows from a straight forward expansion.
	
	\medskip\noindent \textbf{ii. Boundary conditions}\\
	\textbf{Boundary condition \eqref{2TASEP BC1}:}
	Let $l$ be fixed such that $l \in p$ and $l+1 \notin p$. For fixed $\pi \in S_n$ with $\nu_{l+1}=\nu_l$ we note that in \eqref{2TASEP_P_ABC} $\nu_l$ appears only as the exponent for $z_{\pi_l},z_{\pi_{l+1}}$. 
	
	We note here that for all $\sigma \in S_m$, $C_{\pi,\sigma}(p)$ is invariant under swapping $z_{\pi_l}$ and $z_{\pi_{l+1}}$. However this is not true of $B_\pi(p)$, where there is for an extra factor of $(1-z_{\pi_l})^{-1}$. So we consider
	\begin{equation}
	\label{2TASEP BC1 proof_1}
	\left(\frac{1}{1-z_{\pi_l}}\right)^l \left(\frac{1}{1-z_{\pi_{l+1}}}\right)^{l+1} z_{\pi_l}^{\nu_l} z_{\pi_{l+1}}^{\nu_l} \frac{1}{1-z_{\pi_l}} = \left(\frac{1}{1-z_{\pi_l}} \frac{1}{1-z_{\pi_{l+1}}}\right)^{l+1} (z_{\pi_l} z_{\pi_{l+1}})^{\nu_l} 
	\end{equation}
	where the first two factors give the dependence of $A_\pi$ on $z_{\pi_l},z_{\pi_{l+1}}$, and the $(1-z_{\pi_l})^{-1}$ factor comes from $B_\pi(p)$. The right-hand side of \eqref{2TASEP BC1 proof_1} is then invariant under swapping $z_{\pi_l}$ and $z_{\pi_{l+1}}$. Then given the dependence on the signature of the permutation, the term associated with $\pi$ and the one associated with $\pi s_l$ cancel and the result follows.
	
	\medskip\noindent\textbf{Boundary condition \eqref{2TASEP BC2}:} Let $l$ be fixed with $l \notin p$ and $l+1 \in p$. We also let $p' = p - \{l+1\} + \{l\}$. We then consider the appropriate combination of eigenfunctions which factorize as
	\begin{multline}
		\label{2TASEP BC2 proof_1}
		P(\{\nu_1, \dots, \nu_l,\nu_{l+1} = \nu_l , \dots , \nu_n\}, p;t) - P(\{\nu_1, \dots, \nu_l,\nu_{l+1} = \nu_l +1 , \dots , \nu_n\}, p;t) \\
		- P(\{\nu_1, \dots, \nu_l,\nu_{l+1} = \nu_l +1 , \dots , \nu_n\}, p';t) \\ = \sum_{\pi \in S_n} A_\pi B_\pi(p') \prod_{i=1}^n z_{\pi_i}^{\nu_i} \sum_{\sigma \in S_m} C_{\pi,\sigma}(p') 
		\left[\frac{u_{\sigma_r} - z_{\pi_l}}{1-z_{\pi_{l+1}}} - \frac{z_{\pi_{l+1}}(u_{\sigma_r} - z_{\pi_l})}{1-z_{\pi_{l+1}}} - z_{\pi_{l+1}}\right],
	\end{multline}
	where $r$ is the integer such that $p_r = l+1$ and $p'_r = l$. For each $\pi \in S_n, \sigma \in S_m$ the factor in square brackets in \eqref{2TASEP BC2 proof_1} simplifies to be invariant under swapping $z_{\pi_l}$ and $z_{\pi_{l+1}}$
	\begin{equation*}
		\frac{u_{\sigma_r} - z_{\pi_l}}{1-z_{\pi_{l+1}}} - \frac{z_{\pi_{l+1}}(u_{\sigma_r} - z_{\pi_l})}{1-z_{\pi_{l+1}}} - z_{\pi_{l+1}} = u_{\sigma_r} - z_{\pi_l} - z_{\pi_{l+1}}.
	\end{equation*}
	For the term in \eqref{2TASEP BC2 proof_1} corresponding to $\pi \in S_n, \sigma \in S_m$, we only need to consider
	\begin{multline}
		\left(\frac{1}{1-z_{\pi_l}}\right)^l \left(\frac{1}{1-z_{\pi_{l+1}}}\right)^{l+1} z_{\pi_l}^{\nu_l} z_{\pi_{l+1}}^{\nu_l} \frac{1}{1-z_{\pi_l}} (u_{\sigma_r} - z_{\pi_l} - z_{\pi_{l+1}}) \\ = \left(\frac{1}{1-z_{\pi_l}} \frac{1}{1-z_{\pi_{l+1}}}\right)^{l+1} (z_{\pi_l} z_{\pi_{l+1}})^{\nu_l} (u_{\sigma_r} - z_{\pi_l} - z_{\pi_{l+1}})
	\end{multline}
	where the first factors are from the $z_{\pi_l},z_{\pi_{l+1}}$-dependence in $A_\pi \prod_{i} z_{\pi_i}^{\nu_i}$. The $(1-z_{\pi_l})^{-1}$ factor comes from $B_\pi(p')$, as it is the only unpaired factor involving $z_{\pi_l}$ and $z_{\pi_{l+1}}$. The right-hand side of the above is invariant under swapping $z_{\pi_l}$ and $z_{\pi_{l+1}}$, and so the condition is satisfied once the permutation signature of $\pi s_l$ is taken into account.
	
	\medskip\noindent\textbf{Boundary condition \eqref{2TASEP BC3}:}
	Let $l$ be fixed so that either $l,l+1 \in p$ or $l,l+1 \notin p$. We then consider the boundary condition as
	\begin{multline}
		P(\{\nu_1,\dots,\nu_l,\nu_{l+1}=\nu_l,\dots,\nu_n\},p;t)- P(\{\nu_1,\dots,\nu_l,\nu_{l+1}=\nu_l+1,\dots,\nu_n\},p;t) \\
		= \sum_{\pi \in S_n} A_\pi  B_\pi(p) (1-z_{\pi_{l+1}}) \prod_{i=1}^n z_{\pi_i}^{\nu_i} \sum_{\sigma \in S_m}  C_{\pi,\sigma}(p),
	\end{multline}
	where for all $\pi \in S_n,\sigma \in S_m$, $B_\pi(p)$ and $C_{\pi,\sigma}(p)$ are then equal for both terms on the left-hand side above. So for some $\pi \in S_n$, we consider the dependence from $A_\pi \prod_{i} z_{\pi_i}^{\nu_i}$:
	\begin{equation}
	\label{2TASEP BC3 proof_1}
	\left(\frac{1}{1-z_{\pi_l}}\right)^l \left(\frac{1}{1-z_{\pi_{l+1}}}\right)^{l+1} z_{\pi_l}^{\nu_l} z_{\pi_{l+1}}^{\nu_l} (1-z_{\pi_{l+1}}) = \left(\frac{1}{1-z_{\pi_l}} \frac{1}{1-z_{\pi_{l+1}}}\right)^{l} (z_{\pi_l} z_{\pi_{l+1}})^{\nu_l} 
	\end{equation} 
	which is invariant under swapping $z_{\pi_l}$ and $z_{\pi_{l+1}}$, so the result follows when also considering the permutation signature of $\pi s_l$.
	
	\medskip\noindent\textbf{iii. Initial condition \eqref{2TASEP initial condition}:}
	Here we denote the group identity in the symmetric group $S_n$ by $\id_n$. For $\pi \in S_n,\sigma \in S_m$ let us define the integral
	\begin{equation}
	\label{2TASEP inital condition integrand}
	I(\pi,\sigma) = \oint \prod_{i=1}^m \frac{\dd u_i}{2 \pi \ii}  \prod_{i=1}^n \frac{\dd z_i}{2 \pi \ii}   A_\pi B_\pi(p) C_{\pi,\sigma}(p) \prod_{i=1}^n z_{\pi_i}^{\nu_i} z_i^{-\mu_i -1}  \prod_{i=1}^m (1-z_{i})^{m-i+1} \prod_{j=1}^{\inp_i} \frac{1}{u_i - z_j},
	\end{equation}
	so that the initial transition probability \eqref{2TASEP Greens with P} becomes
	\begin{equation}
	\mathbb{P}(\mu\to\nu,\inp\to p;0) = \sum_{\pi \in S_n} \sum_{\sigma \in S_m} I(\pi,\sigma).
	\end{equation}
	We note that when $t=0$ there are no longer essential singularities at $z_i=0$ in the integrand, so we may evaluate the transition probability simply by calculating residues at the origin. To proceed we first consider the following result.
	\begin{lem}
		\label{2TASEP Initial condition Lemma 1 }
		For $\pi \in S_n, \sigma \in S_m$ where either $\pi \neq \id_n$ or $\nu \neq \mu$, then $I(\pi, \sigma) = 0$ for any set of coordinates within the standard regime.
	\end{lem}
	\begin{proof}[Proof of Lemma \ref{2TASEP Initial condition Lemma 1 }.]
		We proceed by induction on the index in the permutation $\pi$. First we consider $\pi \in S_n$ where $\pi_1 = 1$. Then by integrating $z_1$ first in equation \eqref{2TASEP inital condition integrand}, and noting that $A_\pi, B_\pi(p), C_{\pi,\sigma}(p)$ and $\prod_{i=1}^m (1-z_{i})^{m-i+1} \prod_{j=1}^{\inp_i} \frac{1}{u_i - z_j}$ are rational functions in $z_1$ and analytic at $z_1 = 0$, the only contribution to the residue at $z_1 = 0$ is from the factor 
		\[z_1^{\nu_k-\mu_1 -1},\]
		where $k$ is defined as the index satisfying $\pi_k = 1$. In the standard regime we must have $\nu_k \geq \mu_k \geq \mu_1$, where $\mu_k = \mu_1$ if and only if $k=1$. So there is only a contribution to the residue at $z_1 = 0$ if $\nu_1 = \mu_1$ and $\pi_1 = 1$.
		
		Then consider some $\pi \in S_n$ such that $\pi_i = i$ for all $1 \leq i < h$ for some index $h\leq n$. We assume that $I(\pi,\sigma) = 0$ if $\nu_i \neq \mu_i$ for all $1 \leq i < h$. This amounts to expressing the integral \eqref{2TASEP inital condition integrand} in the form
		\begin{equation}
		I(\pi,\sigma) = \prod_{i=1}^{h-1} \delta_{\nu_i,\mu_i} \oint \prod_{i=1}^m \frac{\dd u_i}{2 \pi \ii}  \prod_{i=h}^n \frac{\dd z_i}{2 \pi \ii} D_{\pi,\sigma}^{h-1} \prod_{i=h}^n z_{\pi_i}^{\nu_i} z_i^{\mu_i -1},
		\end{equation}
		where we have performed the integration in the variables $z_1,\dots, z_{h-1}$. Where also
		\[D_{\pi,\sigma}^{h-1} = \lim_{z_{h-1} \rightarrow 0} \dots \lim_{z_{1}\rightarrow 0} A_\pi B_\pi(p) C_{\pi,\sigma}(p) \prod_{j=1}^m (1-z_{j})^{m-j+1} \prod_{k=1}^{\inp_j} \frac{1}{u_j - z_k},  \]
		which comes from evaluating the residues. We note here that $D_{\pi,\sigma}^{h-1}$ is an analytic function in $z_h$ at $z_h = 0$, so that the only contribution to the residue at $z_r = 0$ is from the factor $z_{\pi_k}^{\nu_k} z_h^{-\mu_h -1}$, where $\pi_k = h$. But since $k \geq h$, in the standard regime $\nu_k \geq \nu_h \geq\mu_h$, so we only have a non-zero residue if $k=h$ and $\nu_h = \mu_h$. We conclude if $\pi_h \neq h$ and $\nu_h \neq \mu_h$ then $I(\pi, \sigma) = 0$, so that by induction we arrive at the result. 
	\end{proof} 
	Lemma \ref{2TASEP Initial condition Lemma 1 } implies that 
	\begin{equation}
	\sum_{\pi \in S_n\backslash\{\id_n\}} \sum_{\sigma \in S_m} I(\pi,\sigma) = 0.
	\end{equation}
	Then by performing all the $z$-integrals in $I(\id_n,\sigma)$, we see that 
	\begin{equation}
	\label{2TASEP initial condition proof integral}
	I(\id_n,\sigma) = \prod_{i=1}^{n} \delta_{\nu_i,\mu_i} \oint \prod_{i=1}^m \frac{\dd u_i}{2 \pi \ii} (-1)^{\abs{\sigma}} \prod_{i=1}^{m} \left(\frac{1-u_i}{1-u_{\sigma_i}}\right)^i  u_{\sigma_i}^{p_i} u_i^{-\inp_i-1},
	\end{equation}
	where the $\prod_i u_{\sigma_i}^{p_i} u_i^{-\inp_i-1}$ factors come from evaluating the $z$-residues. Then the integral in \eqref{2TASEP initial condition proof integral} can be identified with the transition probability for the single species TASEP with $m$ particles at time $t=0$ from \cite{tracy_integral_2008}. This TASEP transition probability satisfies a similar delta function initial condition so that with Lemma \ref{2TASEP Initial condition Lemma 1 } we have the desired result
	\begin{equation}
	\mathbb{P}(\mu\to\nu,\inp\to p;0) = \prod_{i=1}^n \delta_{\nu_i,\mu_i}\sum_{\sigma \in S_m} I(\id_n,\sigma) = \prod_{i=1}^{n} \delta_{\nu_i,\mu_i} \prod_{j=1}^m \delta_{p_j,\inp_j}.
	\end{equation}
\end{proof}

\section{Proof of Lemma \ref{2TASEP u_i = z_i residue lemma}}
\label{u_i=z_i proof section}
\begin{proof}
	First we note that for the initial conditions $\inp_i=i$ we have the factor in the transition probability \eqref{2TASEP Greens with P}
	\[\prod_{i=1}^m \prod_{j=1}^{\inp_i} \frac{1}{u_i - z_j} = \prod_{1\leq j \leq i \leq m}\frac{1}{u_i - z_j}.\]
	We then proceed by induction on the index $l$ by performing the integration in $u_l$. The integration is performed in the order $u_
	1,\dots,u_m$. For the base case $l=1$, we perform the $u_1$-integration first and note that in the transition probability integrand we have the factor $\prod_{i=1}^m(u_1-z_i)^{-1}$. We show here that the singularities at $u_i = z_1$ for $i>1$ are removable.
	
	We perform the integration in $u_1$ for the term associated with $\pi\in S_n,\sigma \in S_m$ where for simplicity of notation we set $\tau = \sigma^{-1}$. For some $1<k\leq m$, the $u_1,u_k$-dependence of this term becomes 
	\begin{multline}
		\label{2TASEP u_i=z_i lemma base case 1}
		\oint_{z_1} \frac{\dd u_1}{2\pi\ii} (1-u_1)^{1-\tau_1}(1-u_k)^{k-\tau_k} \prod_{j=1}^{p_{\tau_1}-1}(u_1-z_{\pi_j}) \prod_{j=1}^{p_{\tau_k}-1}(u_k-z_{\pi_j}) \frac{1}{u_1-z_1} \prod_{j=1}^k\frac{1}{u_k-z_j}\\= (1-z_1)^{1-\tau_1}(1-u_k)^{k-\tau_k}\prod_{j=1}^{p_{\tau_1}-1}(z_1-z_{\pi_j}) \prod_{j=1}^{p_{\tau_k}-1}(u_k-z_{\pi_j})\prod_{j=1}^k\frac{1}{u_k-z_j}.
	\end{multline}
	For this permutation $\sigma$ we associate the permutation $\sigma T_{1,k}$, where $T_{1,k}\in S_m$ swaps 1 and $k$ and leaves all other indices fixed. Without loss of generality we are able to choose $\sigma$ such that $\sigma^{-1}_1 > \sigma^{-1}_k$, where we must then consider the cases of $\left(\sigma T_{1,k}\right)^{-1}_1 < \left(\sigma T_{1,k}\right)^{-1}_k$, and later its converse. After noting that $(-1)^\abs{\sigma T_{1,k}}=-(-1)^\abs{\sigma}$ we add the term associated with $\sigma$ form \eqref{2TASEP u_i=z_i lemma base case 1} and the one associated with $\sigma T_{1,k}$ which factors as
	\begin{multline}
		\label{2TASEP u_i=z_i lemma base case 2}
		(1-z_1)^{1-\tau_1-\tau_k}(1-u_k)^{k-\tau_1-\tau_k} \prod_{j=1}^{p_{\tau_k}-1}\left[(z_1-z_{\pi_j})(u_k-z_{\pi_j})\right] \prod_{j=1}^k\frac{1}{u_k-z_j} \\
		\times\left[(1-z_1)^{\tau_k}(1-u_k)^{\tau_1}\prod_{j=p_{\tau_k}}^{p_{\tau_1}-1}(z_1-z_{\pi_j})-(1-z_1)^{\tau_1}(1-u_k)^{\tau_k}\prod_{j=p_{\tau_k}}^{p_{\tau_1}-1}(u_k-z_{\pi_j})\right].
	\end{multline}
	We then are able to note that $u_k=z_1$ is a root of the last factor above, which makes the singularity in the transition probability \eqref{2TASEP Greens with P} at $u_k=z_1$ removable. The case of $\left(\sigma T_{1,k}\right)^{-1}_1 > \left(\sigma T_{1,k}\right)^{-1}_k$ can be performed in a similar manner which also yields the same result.
	
	We then may conclude that this cancels out for all terms associated with a permutation $\sigma$ and thus for the transition probability \eqref{2TASEP Greens with P} in general. Since this analysis was performed for arbitrary $k$ it follows for all $k$ which gives the base case.
	
	For the induction step we assume that the integration in the first $l-1$ $u$-variables has been carried out by evaluating residues of simple poles at $u_i=z_i$ for $1\leq i\leq l-1$. Now we look to integrate in $u_l$ where we may assume that singularities at $u_l = z_i$ for $1\leq i\leq l-1$ are removable, but it remains to show that the poles at $u_k = z_l$ for $l< k\leq m$ are removable. 
	
	In the same way as for the base case, we consider the term associated with the permutation $\sigma$ and the one associated with $\sigma T_{l,k}$ in the integrand of the transition probability. Similarly, without loss of generality we may choose $\sigma \in S_m$ such that $\sigma^{-1}_l > \sigma^{-1}_k$ and choose the case $\left(\sigma T_{l,k}\right)^{-1}_l < \left(\sigma T_{l,k}\right)^{-1}_k$. The sum of these terms will factorize and just as in \eqref{2TASEP u_i=z_i lemma base case 2} as
	\begin{multline}
		\label{2TASEP u_i=z_i lemme induction step 1}
		(1-z_l)^{l-\tau_l-\tau_k}(1-u_k)^{k-\tau_l-\tau_k} \prod_{j=1}^{p_{\tau_k}-1}\left[(z_l-z_{\pi_j})(u_k-z_{\pi_j})\right] \prod_{j=l}^k\frac{1}{u_k-z_j} \\
		\times\left[(1-z_l)^{\tau_k}(1-u_k)^{\tau_l}\prod_{j=p_{\tau_k}}^{p_{\tau_l}-1}(z_l-z_{\pi_j})-(1-z_l)^{\tau_l}(1-u_k)^{\tau_k}\prod_{j=p_{\tau_k}}^{p_{\tau_l}-1}(u_k-z_{\pi_j})\right],
	\end{multline}
	where $\tau = \sigma^{-1}$. We note that $u_k=z_l$ is a root of the last factor above so that $u_k=z_l$ is a removable singularity. Just as with the base case we may perform the same step for the case $\left(\sigma T_{l,k}\right)^{-1}_l > \left(\sigma T_{l,k}\right)^{-1}_k$ which yields the same factor. We then conclude by induction that the only non-removable singularities are those at $u_i=z_i$.
\end{proof}

\section{Proof of Lemma \ref{summation identity 1}}
\label{Geometric series proof appendix}

\begin{proof}[Proof of Lemma \ref{summation identity 1}]
	We proceed by induction on $m$. For $m=1$, \eqref{summation identity 1 equation} holds by evaluating the series
	\begin{equation}
	\label{Geometric series identity proof 1}
	\sum_{\nu_1 = s_2}^\infty z_1^{\nu_i} = \frac{z_1^{s_2}}{1-z_1}.
	\end{equation}
	For the general case we define
	\[f_{m,s_2}(z_1, \dots, z_m) = \sum_{\nu_m = s_2+m-1}^{\infty} \cdots  \sum_{\nu_1 = s_2}^{\nu_2-1}\prod_{i=1}^m z_i^{\nu_i}\]
	which is the left-hand side of equation \eqref{summation identity 1 equation}. Now note 
	\begin{equation}
	f_{m+1,s_2}(z_1, \dots, z_{m+1}) = \sum_{\nu_1=s_2}^\infty z_1^{\nu_1}f_{m,\nu_1+1}(z_2 ,\dots, z_{m+1}).
	\end{equation}
	We then assume \eqref{summation identity 1 equation}, which is our induction hypothesis for $m$. This gives
	\begin{equation}
	f_{m+1,s_2}(z_1, \dots, z_{m+1}) = \sum_{\nu_1=s_2}^\infty z_1^{\nu_1}\prod_{i=2}^{m+1}\left[  \frac{z_i^{\nu_1+i}}{1-\prod_{j=i}^{m+1} z_j}\right].
	\end{equation}
	The using the standard geometric series \eqref{Geometric series identity proof 1} for $z_1 \cdots z_{m+1}$, we arrive at 
	\begin{equation}
	f_{m+1,s_2}(z_1, \dots, z_{m+1}) = \prod_{i=1}^{m+1} \frac{z_i^{s_2+i-1}}{1-\prod_{j=i}^{m+1} z_j},
	\end{equation}
	which gives us the result by induction.
\end{proof}

\section{Proof of Symmetrization Identity Lemma \ref{symmetrization identity 1}}
\label{Symmetrization identity proof appendix}

\begin{proof}[Proof of Lemma \ref{symmetrization identity 1}]
	We proceed to prove the following statement by induction on $m$
	\begin{equation}
	\label{symmetrization proof statement}
	\sum_{\sigma\in S_m}(-1)^\abs{\sigma} \prod_{i=1}^m \left[\frac{z_{\sigma_i}^{i-1}(1-z_{\sigma_i})^{m-i+1}}{1-\prod_{j=i}^m z_{\sigma_j}}\right] = \prod_{1\leq i < j \leq m}\left(z_j-z_i\right),
	\end{equation}
	which is equivalent to \eqref{2TASEP Pcross 1 proof A}.
	
	The statement obviously holds for $m=1$. We then define $f_m(z_1,\dots,z_m)$ as the left-hand side of \eqref{symmetrization proof statement}. Then for each $\sigma \in S_m$ we define $k = \sigma_i$ and decompose the sum over permutations as a sum over $k$ and a sum over a subset of permutations $S^{(k)}_{m} \subset S_{m}$, which consists of all $\pi \in S_m$ such that $\pi_1=k$. Elements of $S_m^{(k)}$ may be regarded as permutations of $\{1,\dots,k-1,k+1,\dots,m\}$ so that $S_m^{(k)} \cong S_{m-1}$. We note that the signature of the permutation $(k,1,\dots,k-1,k+1,\dots,m)$ is $(-1)^{k-1}$, which gives the decomposition as 
	\begin{equation}
	f_m(z_1,\dots,z_m) = \frac{1}{1-\prod_{j=1}^m z_j}\sum_{k=1}^m(-1)^{k-1} (1-z_k)^m\sum_{\pi \in S^{(k)}_{m}} (-1)^{\abs{\pi}} \prod_{i=2}^m \left[\frac{z_{\pi_i}^{i-1}(1-z_{\pi_{i}})^{n-i+1}}{1-\prod_{j=i}^n z_{\pi_j}}\right].
	\end{equation}
	We then recognize that the sum over $S_n^{(k)}$ above as $f_{m-1}(z_1, \dots, z_{k-1},z_{k+1}, \dots, z_m) \prod_{i\neq k} z_i$ so that
	\begin{equation}
	\label{symmetrization proof 1}
	f_m(z_1,\dots,z_m) = {1-\prod_{j=1}^n z_j}\sum_{k=1}^m(-1)^{k-1} (1-z_k)^m  f_{m-1}(z_1, \dots, z_{k-1},z_{k+1}, \dots, z_m) \prod_{i\neq k} z_i.
	\end{equation}
	Then using the induction hypothesis and
	\[\prod_{\substack{1 \leq i<j \leq m\\ i,j\neq k}}(z_j-z_i) = (-1)^{k-1} \frac{\prod_{1 \leq i<j \leq j} (z_j-z_i)}{\prod_{i \neq k}(z_i-z_k)},\]
	equation \eqref{symmetrization proof 1} becomes
	\begin{equation}
	\label{symmetrization proof 2}
	f_m(z_1,\dots,z_m) = \frac{\prod_{1 \leq i<j \leq m}(z_j-z_i)}{1-\prod_{j=1}^m z_j}  \sum_{k=1}^m (1-z_k)^m \prod_{i \neq k}\frac{z_i}{z_i-z_k}.
	\end{equation}
	We then simplify the sum above by defining the integral
	\[
	\sum_{k=1}^m (1-z_k)^m \prod_{i \neq k}\frac{z_i}{z_i-z_k} = \prod_{i=1}^m z_i \oint_{\mathcal{C}} \frac{\dd w}{2\pi\ii w} \frac{(1-w)^m}{\prod_{i=1}^m(z_i-w)},
	\]
	where the closed contour $\mathcal{C}$ is positively oriented and encloses the all of the points $\{z_i\}_{i=1}^m \subset \bbC$ but not the origin. We may then reverse the direction of the contour so that we now have a contour surrounding simple poles at $w=0,\infty$
	\begin{equation}
	\prod_{i=1}^m z_i\oint_{\mathcal{C}} \frac{\dd w}{2\pi\ii w} \frac{(1-w)^m}{\prod_{i=1}^m(z_i-w)} = - \prod_{i=1}^m z_i\oint_{0,\infty} \frac{\dd w}{2\pi\ii w} \frac{(1-w)^m}{\prod_{i=1}^m(z_i-w)}. 
	\end{equation}
	We may then evaluate this integral using residue calculus. Firstly about the origin
	\[-\prod_{i=1}^m z_i\oint_{0} \frac{\dd w}{2\pi\ii w} \frac{(1-w)^m}{\prod_{i=1}^m(z_i-w)} = -1,\]
	and then about $\infty$ with a change of coordinates $w \mapsto v^{-1}$ 
	\[-\prod_{i=1}^m z_i\oint_{\infty} \frac{\dd w}{2\pi\ii w} \frac{(1-w)^m}{\prod_{i=1}^m(z_i-w)} = \prod_{i=1}^n z_i\oint_{0} \frac{\dd v}{2\pi\ii v} \frac{(v-1)^m}{\prod_{i=1}^m(z_iv-1)} = \prod_{i=1}^m z_i.\]
	Using this together with \eqref{symmetrization proof 2} gives the desired result. 
\end{proof}

\newpage

\bibliographystyle{amsplain} 
\addcontentsline{toc}{chapter}{Bibliography}
\bibliography{Library}

\providecommand{\bysame}{\leavevmode\hbox to3em{\hrulefill}\thinspace}
\providecommand{\MR}{\relax\ifhmode\unskip\space\fi MR }
\providecommand{\MRhref}[2]{%
  \href{http://www.ams.org/mathscinet-getitem?mr=#1}{#2}
}
\providecommand{\href}[2]{#2}
\begin{thebibliography}{10}

\bibitem{aggarwal2017convergence}
A.~Aggarwal, \emph{Convergence of the stochastic six-vertex model to the asep},
  Math. Phys. Anal. Geom. \textbf{20} (2017), no.~2, 1--20.

\bibitem{AlcarazBariev00}
F.~C. Alcaraz and R.~Z. Bariev, \emph{Exact solution of asymmetric diffusion
  with $n$ classes of particles of arbitrary size and hierarchical order},
  Braz. J. Phys. \textbf{30} (2000), 655.

\bibitem{Alcaraz_Stochastic_1998}
F.~C. Alcaraz, S.~Dasmahapatra, and V.~Rittenberg, \emph{Stochastic models with
  boundaries and quadratic algebras}, Phys. A: Stat. Mech. \textbf{257} (1998),
  no.~1, 1--9.

\bibitem{ALCARAZ1993377}
F.~C. Alcaraz and V.~Rittenberg, \emph{Reaction-diffusion processes as physical
  realizations of hecke algebras}, Phys. Lett. B. \textbf{314} (1993), no.~3,
  377--380.

\bibitem{Amir2010}
G.~Amir, I.~Corwin, and J.~Quastel, \emph{Probability distribution of the free
  energy of the continuum directed random polymer in 1 + 1 dimensions}, Commun.
  Pure Appl. Math. \textbf{64} (2010), no.~4, 466–537.

\bibitem{arita2009spectrum}
C.~Arita, A.~Kuniba, K.~Sakai, and T.~Sawabe, \emph{Spectrum of a multi-species
  asymmetric simple exclusion process on a ring}, J. Phys. A. \textbf{42}
  (2009), 345002.

\bibitem{BaikRains00}
J.~Baik and E.~M. Rains, \emph{Limiting distributions for a polynuclear growth
  model with external sources}, J. Stat. Phys. (2000), 523–541.

\bibitem{BS1995}
A.~L. Barab{\'a}si and H.~E. Stanley, \emph{Fractal concepts in surface
  growth}, Cambridge University Press, 1995.

\bibitem{Belitsky_Self_2018}
V.~Belitsky and G.~M. Schütz, \emph{Self-duality and shock dynamics in the
  n-species priority asep}, Stoch. Process. Their Appl. \textbf{128} (2018),
  1165--1207.

\bibitem{borodin2017family}
A.~Borodin, \emph{On a family of symmetric rational functions}, Adv. Math.
  \textbf{306} (2017), 973--1018.

\bibitem{borodin2019colorposition}
A.~Borodin and A.~Bufetov, \emph{Color-position symmetry in interacting
  particle systems},  (2019).

\bibitem{BC2014}
A.~Borodin and I.~Corwin, \emph{{Macdonald processes}}, Prob. Th. Rel. Fields
  \textbf{158} (2014), 225--400.

\bibitem{bcs}
A.~Borodin, I.~Corwin, and T.~Sasamoto, \emph{From duality to determinants for
  q-tasep and asep}, Ann. Probab. \textbf{42} (2014), 2314--2382.

\bibitem{borodin2020shiftinvariance}
A.~Borodin, V.~Gorin, and M.~Wheeler, \emph{Shift-invariance for vertex models
  and polymers}, 2020, pp.~182--299.

\bibitem{BP2016}
A.~Borodin and L.~Petrov, \emph{{Lectures on Integrable probability: Stochastic
  vertex models and symmetric functions, 1605.01349}}, Lecture Notes of the Les
  Houches Summer School \textbf{104} (2016).

\bibitem{BorodinPetrov}
\bysame, \emph{Higher spin six vertex model and symmetric rational functions},
  Sel. Math. \textbf{24} (2018).

\bibitem{borodin_coloured_2018}
A.~Borodin and M.~Wheeler, \emph{Coloured stochastic vertex models and their
  spectral theory}, arXiv:1808.01866 [math-ph] (2018), arXiv: 1808.01866.

\bibitem{Bosnjak_2016}
G.~Bosnjak and V.~V. Mangazeev, \emph{Construction ofr-matrices for symmetric
  tensor representations related to ${U}_{q}(\hat{{{sl} }_{n}})$}, J. Phys. A.
  \textbf{49} (2016), no.~49, 495204.

\bibitem{CalabreseDoussal}
P.~Calabrese and P.~Le~Doussal, \emph{Exact solution for the
  kardar-parisi-zhang equation with flat initial conditions}, Phys. Rev. Lett.
  \textbf{106} (2011), 250603.

\bibitem{Cantini_Algebraic_2008}
L.~Cantini, \emph{Algebraic bethe ansatz for the two species asep with
  different hopping rates}, J. Phys. A. \textbf{41} (2008), no.~9, 095001.

\bibitem{Cantini_2015}
L.~Cantini, J.~de~Gier, and M.~Wheeler, \emph{Matrix product formula for
  macdonald polynomials}, J. Phys. A. \textbf{48} (2015), no.~38, 384001.

\bibitem{Chatterjee_2010}
S.~Chatterjee and G.~M. Schütz, \emph{Determinant representation for some
  transition probabilities in the tasep with second class particles}, J. Stat.
  Phys. \textbf{140} (2010), no.~5, 900–916.

\bibitem{chen_exact_2019}
Z.~Chen, \emph{Exact solutions in multi-species exclusion processes}, Ph.D.
  thesis, The University of Melbourne, 2019.

\bibitem{chen_exact_2018}
Z.~Chen, J.~de~Gier, I.~Hiki, and T.~Sasamoto, \emph{Exact confirmation of {1D}
  nonlinear fluctuating hydrodynamics for a two-species exclusion process},
  Phys. Rev. Lett. \textbf{120} (2018), no.~24, 240601, arXiv: 1803.06829.

\bibitem{chen_exact_2021}
Z.~Chen, J.~de~Gier, I.~Hiki, T.~Sasamoto, and M.~Usui, \emph{Limiting
  {Current} {Distribution} for a {Two} {Species} {Asymmetric} {Exclusion}
  {Process}}, Commun. Math. Phys. \textbf{395} (2022), no.~1, 59--142 (en).

\bibitem{ChenGW}
Z.~Chen, J.~de~Gier, and M.~Wheeler, \emph{{Integrable Stochastic Dualities and
  the Deformed Knizhnik–Zamolodchikov Equation}}, Int. Math. Res. Not.
  \textbf{2020} (2018), no.~19, 5872--5925.

\bibitem{Corwin_2015}
I.~Corwin and L.~Petrov, \emph{Stochastic higher spin vertex models on the
  line}, Commun. Math. Phys. \textbf{343} (2015), no.~2, 651–700.

\bibitem{Gier2005}
J.~de~Gier and F.~H.~L. Essler, \emph{Bethe ansatz solution of the asymmetric
  exclusion process with open boundaries}, Phys. Rev. Lett. \textbf{95} (2005),
  no.~24.

\bibitem{Gier2006}
\bysame, \emph{Exact spectral gaps of the asymmetric exclusion process with
  open boundaries}, J. Stat. Mech.: Theory Exp \textbf{2006} (2006), no.~12,
  P12011–P12011.

\bibitem{Gier2008}
\bysame, \emph{Slowest relaxation mode of the partially asymmetric exclusion
  process with open boundaries}, J. Phys. A. \textbf{41} (2008), no.~48,
  485002.

\bibitem{Derrida_exactsolution}
B.~Derrida, S.~A. Janowsky, J.~L. Lebowitz, and E.~R. Speer, \emph{Exact
  solution of the totally asymmetric simple exclusion process: {Shock}
  profiles}, J. Stat. Phys. \textbf{73} (1993), no.~5, 813--842 (en).

\bibitem{ferrari_microscopic_1991}
P.~A. Ferrari, C.~Kipnis, and E.~Saada, \emph{Microscopic structure of
  travelling waves in the asymmetric simple exclusion process}, 1991,
  pp.~226--244.

\bibitem{ferrari2007}
P.~A. Ferrari and J.~B. Martin, \emph{Stationary distributions of multi-type
  totally asymmetric exclusion processes}, Ann. Probab. \textbf{35} (2007),
  no.~3, 807--832.

\bibitem{FNG2019}
P.~L. Ferrari, P.~Nejjar, and P.~Ghosal, \emph{Limit law of a second class
  particle in {TASEP} with non-random initial condition}, Ann. Inst. H.
  Poincare Probab. Statist. \textbf{55} (2019), 1203--1225.

\bibitem{galashin2020symmetries}
P.~Galashin, \emph{Symmetries of stochastic colored vertex models}.

\bibitem{Garbali2017}
A.~Garbali, J.~de~Gier, and M.~Wheeler, \emph{A new generalisation of macdonald
  polynomials}, Commun. Math. Phys. \textbf{352} (2017), no.~2, 773–804.

\bibitem{Golinelli_2006}
O.~Golinelli and K.~Mallick, \emph{The asymmetric simple exclusion process: an
  integrable model for non-equilibrium statistical mechanics}, J. Phys. A.
  \textbf{39} (2006), no.~41, 12679–12705.

\bibitem{PhysRevA.46.844}
L.~H. Gwa and H.~Spohn, \emph{Bethe solution for the dynamical-scaling exponent
  of the noisy burgers equation}, Phys. Rev. A \textbf{46} (1992), 844--854.

\bibitem{PhysRevLett.68.725}
\bysame, \emph{Six-vertex model, roughened surfaces, and an asymmetric spin
  hamiltonian}, Phys. Rev. Lett. \textbf{68} (1992), 725--728.

\bibitem{J2000}
K.~Johansson, \emph{Shape fluctuations and random matrices}, Commun. Math.
  Phys. \textbf{209} (2000), 437--476.

\bibitem{KPZ1986}
M.~Kardar, G.~Parisi, and Y.~C. Zhang, \emph{Dynamic scaling of growing
  interfaces}, Phys. Rev. Lett. \textbf{56} (1986), 889--892.

\bibitem{PhysRevE.59.205}
V.~Karimipour, \emph{Multispecies asymmetric simple exclusion process and its
  relation to traffic flow}, Phys. Rev. E. \textbf{59} (1999), 205--212.

\bibitem{PhysRevE.52.3512}
D.~Kim, \emph{Bethe ansatz solution for crossover scaling functions of the
  asymmetric xxz chain and the kardar-parisi-zhang-type growth model}, Phys.
  Rev. E \textbf{52} (1995), 3512--3524.

\bibitem{Kirillov_1987}
A.~N. Kirillov and N.~Yu~Reshetikhin, \emph{Exact solution of the integrable
  {XXZ} heisenberg model with arbitrary spin. i. the ground state and the
  excitation spectrum}, J. Phys. A. \textbf{20} (1987), no.~6, 1565--1585.

\bibitem{Kuan_2018}
J.~Kuan, \emph{An algebraic construction of duality functions for the
  stochastic ${U_q( A_n^{(1)})}$ vertex model and its degenerations}, Commun.
  Math. Phys. \textbf{359} (2018), no.~1, 121–187.

\bibitem{kuan2019probability}
\bysame, \emph{Probability distributions of multi-species q-tazrp and asep as
  double cosets of parabolic subgroups}, Ann. Henri Poincar{\'e}, vol.~20,
  2019, pp.~1149--1173.

\bibitem{kuan_determinantal_2020}
\bysame, \emph{Determinantal {Expressions} in {Multi}-{Species} {TASEP}}, SIGMA
  (2020), arXiv: 2007.02913.

\bibitem{Kuan_2021}
\bysame, \emph{{Two Dualities: Markov and Schur–Weyl}}, Int. Math. Res. Not.
  (2021), rnaa333.

\bibitem{Kulish1981YangBaxterEA}
P.~P. Kulish, N.~Reshetikhin, and E.~Sklyanin, \emph{Yang-baxter equation and
  representation theory: I}, Lett. Math. Phys. \textbf{5} (1981), 393--403.

\bibitem{Kuniba_2016}
A.~Kuniba, S.~Maruyama, and M.~Okado, \emph{Multispecies {TASEP} and the
  tetrahedron equation}, J. Phys. A.l \textbf{49} (2016), no.~11, 114001.

\bibitem{kuniba2016multispecies}
\bysame, \emph{Multispecies totally asymmetric zero range process: I. multiline
  process and combinatorial $r$}, J. Int. Systems \textbf{1} (2016).

\bibitem{lee2018exact}
E.~Lee, \emph{Exact formulas of the transition probabilities of the
  multi-species asymmetric simple exclusion process}, arXiv preprint
  arXiv:1809.07362 (2018).

\bibitem{Liggett1985}
T.~M. Liggett, \emph{Interacting particle systems}, Springer-Verlag, 1985.

\bibitem{Liggett1999}
\bysame, \emph{Stochastic interacting systems: Contact, voter, and exclusion
  processes}, Springer-Verlag, 1999.

\bibitem{Macdonald69}
C.~T. MacDonald and J.~H. Gibbs, \emph{Concerning the kinetics of polypeptide
  synthesis on polyribosomes}, Biopolymers \textbf{7} (1969), no.~5, 707--725.

\bibitem{Macdonald68}
C.~T. MacDonald, J.~H. Gibbs, and A.~C. Pipkin, \emph{Kinetics of
  biopolymerization on nucleic acid templates}, Biopolymers \textbf{6} (1968),
  no.~1, 1--25.

\bibitem{Mallick_1999}
K.~Mallick, S.~Mallick, and N.~Rajewsky, \emph{Exact solution of an exclusion
  process with three classes of particles and vacancies}, J. Phys. A.
  \textbf{32} (1999), no.~48, 8399--8410.

\bibitem{MANGAZEEV201470}
V.~V. Mangazeev, \emph{On the yang–baxter equation for the six-vertex model},
  Nucl. Phys. B. \textbf{882} (2014), 70 -- 96.

\bibitem{NS2004}
T.~Nagao and T.~Sasamoto, \emph{Asymmetric simple exclusion process and
  modified random matrix ensembles}, Nucl. Phys. B \textbf{699} (2004),
  487--502.

\bibitem{Nejjar2019}
P.~Nejjar, \emph{Kpz statistics of second class particles in asep via mixing},
  Commun. Math. Phys. \textbf{378} (2019), 601--623.

\bibitem{nejjar_kpz_2020}
\bysame, \emph{{KPZ} statistics of second class particles in {ASEP} via
  mixing}, Commun. Math. Phys. \textbf{378} (2020), no.~1, 601--623, arXiv:
  1911.09426.

\bibitem{PerkSchultz}
J.~H.~H. Perk and C.~L. Schultz, \emph{New families of commuting transfer
  matrices in q-state vertex models}, 1981, p.~407.

\bibitem{Popkov_2002}
V.~Popkov, M.~E. Fouladvand, and G.~M. Schütz, \emph{A sufficient criterion
  for integrability of stochastic many-body dynamics and quantum spin chains},
  J. Phys. A.l \textbf{35} (2002), no.~33, 7187–7204.

\bibitem{povolotsky2013}
A.~Povolotsky, \emph{On integrability of zero-range chipping models with
  factorized steady state}, J. Phys. A.l \textbf{46} (2013).

\bibitem{PEM2009}
S.~Prolhac, M.~Evans, and K.~Mallick, \emph{Matrix product solution of the
  multispecies partially asymmetric exclusion process}, J. Phys. A. \textbf{42}
  (2009), 165004.

\bibitem{PS2002}
M.~Prähofer and H.~Spohn, \emph{Current {Fluctuations} for the {Totally}
  {Asymmetric} {Simple} {Exclusion} {Process}},  (2002), 185--204 (en).

\bibitem{Sasamoto2005}
T.~Sasamoto, \emph{Spatial correlations of the 1d kpz surface on a flat
  substrate}, J. Phys. A. \textbf{38} (2005), no.~33, L549--L556.

\bibitem{Sasamoto2010a}
T.~Sasamoto and H.~Spohn, \emph{The crossover regime for the weakly asymmetric
  simple exclusion process}, J. Stat. Phys. \textbf{140} (2010), no.~2,
  209–231.

\bibitem{Sasamoto2010b}
\bysame, \emph{Exact height distributions for the kpz equation with narrow
  wedge initial condition}, Nucl. Phys. B. \textbf{834} (2010), no.~3,
  523–542.

\bibitem{Sasamoto2010c}
\bysame, \emph{One-dimensional kardar-parisi-zhang equation: An exact solution
  and its universality}, Phys. Rev. Lett. \textbf{104} (2010), no.~23.

\bibitem{schutz_exact_1997}
G.~M. Schütz, \emph{Exact {Solution} of the {Master} {Equation} for the
  {Asymmetric} {Exclusion} {Process}}, J. Stat. Phys. \textbf{88} (1997),
  no.~1-2, 427--445, arXiv: cond-mat/9701019.

\bibitem{Schutz_Critical_2003}
\bysame, \emph{Critical phenomena and universal dynamics in one-dimensional
  driven diffusive systems with two species of particles}, J. Phys. A.
  \textbf{36} (2003), no.~36, R339--379.

\bibitem{Spohn1991}
H.~Spohn, \emph{Large scale dynamics of interacting particles}, Springer, 1991.

\bibitem{tracy2008integral}
C.~A. Tracy and H.~Widom, \emph{Integral formulas for the asymmetric simple
  exclusion process}, Commun. Math. Phys. \textbf{279} (2008), 815--844.

\bibitem{tracy_integral_2008}
\bysame, \emph{Integral {Formulas} for the {Asymmetric} {Simple} {Exclusion}
  {Process}}, Commun. Math. Phys. \textbf{279} (2008), no.~3, 815--844, arXiv:
  0704.2633.

\bibitem{TW2009a}
\bysame, \emph{{Asymptotics in ASEP with step initial condition}}, Commun.
  Math. Phys. \textbf{209} (2009), 129--154.

\bibitem{tracy2013asymmetric}
\bysame, \emph{On the asymmetric simple exclusion process with multiple
  species}, J. Stat. Phys. \textbf{150} (2013), 457--470.

\bibitem{YY1966}
C.~N. Yang and C.~P. Yang, \emph{One-dimensional chain of anisotropic spin-spin
  interactions. 1. proof of bethe’s hypothesis for the ground state in a
  finite system}, Phys. Rev. \textbf{150} (1966), 321--327.

\bibitem{Y2004}
H.~T. Yau, \emph{$(\log{t})^{2/3}$ law of the two dimensional asymmetric simple
  exclusion}, Ann. Math. \textbf{159} (2004), 377--405.

\bibitem{Zhang_2019}
X.~Zhang, F.~Wen, and J.~de~Gier, \emph{T-q relations for the integrable
  two-species asymmetric simple exclusion process with open boundaries}, J.
  Stat. Mech.: Theory Exp \textbf{2019} (2019), no.~1, 014001.

\end{thebibliography}

\end{document}